\documentclass[runningheads,USenglish]{llncs}
\usepackage[paperheight=235mm,paperwidth=155mm,textwidth=12.2cm,textheight=19.3cm]{geometry}
\usepackage{bbding}
\usepackage{graphicx}
\makeatletter
\RequirePackage[bookmarks,unicode,colorlinks=true]{hyperref}%
   \def\@citecolor{blue}%
   \def\@urlcolor{blue}%
   \def\@linkcolor{blue}%

\def\orcidID#1{\smash{\href{http://orcid.org/#1}{\protect\raisebox{-1.25pt}{\protect\includegraphics{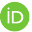}}}}}
\makeatother

\usepackage{booktabs}
\usepackage{siunitx}
\usepackage{amsmath}
\usepackage{skull}
\usepackage{amssymb}
\usepackage{stmaryrd}
\usepackage{thm-restate}
\usepackage{multicol}
\usepackage{color}
\usepackage{algorithm}
\usepackage{float}
\usepackage{multirow}
\usepackage{relsize}
\usepackage{graphicx}
\usepackage{tikz}
\usepackage{mathpartir}
\usepackage{xspace}
\usepackage{paralist}
\usepackage[shortlabels,inline]{enumitem}

\usepackage[capitalise,nameinlink]{cleveref}
\usetikzlibrary{calc}
\usetikzlibrary{arrows,automata}
\usetikzlibrary{positioning}
\usepackage{ifthen}
\usepackage{xargs}
\usepackage{placeins}
\usepackage{pifont}
\usepackage{xcolor}
\usepackage{etoolbox}
\usepackage{adjustbox}
\usepackage{changepage}
\usepackage{ifthen}
\usepackage{cancel}
\usepackage{cite}

\newcommand{\cmark}{{\color{green!70!black}\text{\ding{51}}}}%
\newcommand{\xmark}{{\color{red!70!black}\text{\ding{55}}}}%
\newcommand{\checkmarkt}[1]{%
	\edef\TVALUE{{#1}}%
	\expandafter\ifstrequal\TVALUE{yes}{\cmark}{}%
	\expandafter\ifstrequal\TVALUE{no}{\xmark}{}%
			\expandafter\ifstrequal\TVALUE{no*}{\xmark $^\star$}{}%
}

\spnewtheorem{cremark}{Remark}{\bfseries}{\itshape}
\spnewtheorem{cnotation}[definition]{Notation}{\bfseries}{\itshape}

\crefformat{section}{#2\S{}#1#3}
\Crefname{section}{Section}{Sections}
\Crefformat{section}{Section #2#1#3}
\crefname{corollary}{\text{Corollary}}{\text{Corollaries}}
\Crefname{corollary}{\text{Corollary}}{\text{Corollaries}}
\crefname{lemma}{\text{Lemma}}{\text{Lemmas}}
\Crefname{lemma}{\text{Lemma}}{\text{Lemmas}}
\crefname{proposition}{\text{Prop.}}{\text{Propositions}}
\Crefname{proposition}{\text{Proposition}}{\text{Propositions}}
\crefname{definition}{\text{Def.}}{\text{Definitions}}
\Crefname{definition}{\text{Definition}}{\text{Definitions}}
\crefname{notation}{\text{Notation}}{\text{Notations}}
\Crefname{notation}{\text{Notation}}{\text{Notations}}
\crefname{theorem}{\text{Thm.}}{\text{Theorems}}
\Crefname{theorem}{\text{Theorem}}{\text{Theorems}}
\crefname{figure}{\text{Fig.}}{\text{Figures}}
\Crefname{figure}{\text{Figure}}{\text{Figures}}
\crefname{example}{\text{Example}}{\text{Examples}}
\Crefname{example}{\text{Example}}{\text{Examples}}

\newcommand{\noop}[1]{}

\newcommand{\kw}[1]{\ensuremath{\mathbf{\mathtt{#1}}}}

\newcommand{\ALT}{\;\;|\;\;}

\newcommand{\ie}{{i.e.,} }
\newcommand{\eg}{{e.g.,} }

\newcommand{\wrt}{w.r.t.~}
\newcommand{\aka}{a.k.a.~}

\newcommand{\inarr}[1]{\begin{array}{@{}l@{}}#1\end{array}}

\newcommand{\set}[1]{\{{#1}\}}

\newcommand{\N}{{\mathbb{N}}}
\newcommand{\dom}[1]{\textit{dom}{({#1})}}

\newcommand{\tup}[1]{{\langle{#1}\rangle}}
\newcommand{\nin}{\not\in}
\newcommand{\suq}{\subseteq}

\newcommand{\size}[1]{{|{#1}|}}

\newcommand{\maketil}[1]{{#1}\ldots{#1}}
\newcommand{\til}{\maketil{,}}

\newcommand{\uplustil}{\maketil{\uplus}}

\newcommand{\cdottil}{\cdot\!\ldots\!\cdot}

\newcommand{\rst}[1]{|_{#1}}

\newcommand{\defeq}{\mathrel{\overset{\makebox[0pt]{\mbox{\normalfont\tiny\sffamily def}}}{=}}}

\newcommand{\powerset}[1]{\mathcal{P}({#1})}%

\makeatletter
\newcommand{\raisemath}[1]{\mathpalette{\raisem@th{#1}}}
\newcommand{\raisem@th}[3]{\raisebox{#1}{$#2#3$}}
\makeatother

\newcommandx{\yaHelper}[2][1=\empty]{%
\ifthenelse{\equal{#1}{\empty}}%
  { \ensuremath{ \scriptstyle{ #2 } } } %
  { \raisebox{ #1 }[0pt][0pt]{ \ensuremath{ \scriptstyle{ #2 } } } }  %
}

\newcommandx{\yrightarrow}[4][1=\empty, 2=\empty, 4=\empty, usedefault=@]{%
  \ifthenelse{\equal{#2}{\empty}}
  { \xrightarrow{ \protect{ \yaHelper[ #4 ]{ #3 } } } } %
  { \xrightarrow[ \protect{ \yaHelper[ #2 ]{ #1 } } ]{ \protect{ \yaHelper[ #4 ]{ #3 } } } } %
}

\newcommand{\astep}[1]{\mathrel{\raisebox{-0.8pt}{\ensuremath{\xrightarrow{#1}}}}}
\newcommand{\asteplab}[2]{{}\mathrel{\raisebox{-0.8pt}{\ensuremath{\xrightarrow{#1}}}_{#2}}{}}

\makeatletter
\newcommand{\xRightarrow}[2][]{\ext@arrow 0359\Rightarrowfill@{#1}{#2}}
\makeatother

\colorlet{colorPO}{gray!60!black}
\colorlet{colorRF}{green!60!black}
\colorlet{colorMO}{orange}
\colorlet{colorFR}{purple}
\colorlet{colorECO}{red!80!black}
\colorlet{colorSYN}{green!40!black}
\colorlet{colorHB}{blue}
\colorlet{colorPPO}{magenta}
\colorlet{colorPB}{olive}
\colorlet{colorSBRF}{olive}
\colorlet{colorRMW}{olive!70!black}
\colorlet{colorRSEQ}{blue}
\colorlet{colorSC}{violet}
\colorlet{colorPSC}{violet}
\colorlet{colorREL}{olive}
\colorlet{colorCONFLICT}{olive}
\colorlet{colorRACE}{olive}
\colorlet{colorWB}{orange!70!black}
\colorlet{colorPSC}{violet}
\colorlet{colorSCB}{violet}
\colorlet{colorDEPS}{violet}
\colorlet{colorS}{orange!70!black}
\colorlet{colorTPO}{olive}
\colorlet{colorDTPO}{violet!80!black}

\tikzset{
   every path/.style={>=stealth},
   po/.style={->,color=colorPO,thin,shorten >=-0.5mm,shorten <=-0.5mm},
   sw/.style={->,color=colorSYN,shorten >=-0.5mm,shorten <=-0.5mm},
   rf/.style={->,color=colorRF,dashed,,shorten >=-0.5mm,shorten <=-0.5mm},
   hb/.style={->,color=colorHB,thick,shorten >=-0.5mm,shorten <=-0.5mm},
   mo/.style={->,color=colorMO,dotted,very thick,shorten >=-0.5mm,shorten <=-0.5mm},
   no/.style={->,dotted,thick,shorten >=-0.5mm,shorten <=-0.5mm},
   fr/.style={->,color=colorFR,dotted,thick,shorten >=-0.5mm,shorten <=-0.5mm},
   deps/.style={->,color=colorDEPS,dotted,thick,shorten >=-0.5mm,shorten <=-0.5mm},
   rmw/.style={->,color=colorRMW,thick,shorten >=-0.5mm,shorten <=-0.5mm},
   tpo/.style={->,color=colorTPO,dotted,thick,shorten >=-0.5mm,shorten <=-0.5mm},
   dtpo/.style={->,color=colorDTPO,dotted, thick,shorten >=-0.5mm,shorten <=-0.5mm},
   revisit/.style={inner sep=1pt,rounded corners,fill=phlightcolor},
}

\newcommand{\rlab}[2]{{\lR}({#1},{#2})}
\newcommand{\wlab}[2]{{\lW}({#1},{#2})}
\newcommand{\fllab}[1]{{\lFL}({#1})}
\newcommand{\ulab}[3]{{\lU}({#1},{#2},{#3})}

\newcommand{\folab}[1]{{\lFO}({#1})}
 \newcommand{\sflab}{{\lSF}}
\newcommand{\pflab}[1]{{\lPF}({#1})}

\newcommand{\rexlab}[2]{{\lRex}({#1},{#2})}

\newcommand{\fotlabp}[1]{{\lFOT}({#1})}
\newcommand{\persistlab}{{\lPERSIST}}

\newcommand{\instr}[1]{\green{#1}}

\newcommand{\SFLab}{\mathsf{SFLab}}

\newcommand{\lR}{{\mathtt{R}}}
\newcommand{\lW}{{\mathtt{W}}}

\newcommand{\lQ}{{\mathtt{Q}}}

\newcommand{\lU}{{\mathtt{RMW}}}

\newcommand{\linit}{{\mathtt{q}_\Init}}

\newcommand{\lT}{{\mathtt{T}}}
\newcommand{\lFL}{{\mathtt{FL}}}
\newcommand{\lFO}{{\mathtt{FO}}}
\newcommand{\lPF}{{\mathtt{LSF}}}
 \newcommand{\lSF}{{\mathtt{SF}}}

\newcommand{\lFOT}{{\mathtt{FO}}}
\newcommand{\lRex}{{\mathtt{R}\text{-}\mathtt{ex}}}
\newcommand{\lPERSIST}{{\mathtt{per}}}

\newcommand{\lSigma}{{\mathbf{\Sigma}}}

\newcommand{\lTYP}{{\mathtt{typ}}}
\newcommand{\lLOC}{{\mathtt{var}}}
\newcommand{\lLOCSET}{{\mathtt{varset}}}

\newcommand{\lX}{\mathtt{X}}

\newcommand{\Init}{\mathsf{Init}}
\newcommand{\Tid}{\mathsf{Tid}}
\newcommand{\Loc}{\mathsf{Var}}
\newcommand{\NVLoc}{\mathsf{NV}\Loc}
\newcommand{\VLoc}{\mathsf{V}\Loc}

\newcommand{\Val}{\mathsf{Val}}

\newcommand{\Lab}{\mathsf{Lab}}

\newcommand{\HTLab}{\mathsf{HTLab}}

\newcommand{\Reg}{\mathsf{Reg}}

\newcommand{\sep}{\hspace{0.05em};\,}
\newcommand{\iseq}{I}

\newcommand{\readInst }[2]{#1 \;{:=}\;#2}

\newcommand{\havocInst}{\kw{havoc}}

\newcommand{\pfenceInst}[1]{\kw{lsfence}({#1})}
\newcommand{\sfenceInst}{\kw{sfence}}
\newcommand{\ifGotoInst}[2]{\kw{if} \; #1 \; \kw{goto} \; #2}
\newcommand{\ifGotoInsts}[2]{\kw{if} \; #1 \\ \quad\; \kw{goto} \; #2}

\newcommand{\GotoInst}[1]{\kw{goto} \; #1}
\newcommand{\writeInst}[2]{#1\;{:=}\;#2}
\newcommand{\assignInst}[2]{#1\;{:=}\;#2}
\newcommand{\incInst}[3]{#1 \;{:=}\;\faddInstn({#2},{#3})}

\newcommand{\casInst}[4]{#1 \;{:=}\;\casInstn({#2},{#3},{#4})}
\newcommand{\casInstn}{\kw{CAS}}
\newcommand{\faddInstn}{\kw{FADD}}

\newcommand{\flInst}[1]{\kw{fl}({#1})}
\newcommand{\foInst}[1]{\kw{fo}({#1})}

\newcommand{\callInst}[1]{\kw{call}({#1})}
\newcommand{\returnInst}{\kw{return}}

\newcommand{\myendpbInst}[1]{\kw{endPB}({#1})}
\newcommand{\beginpbInst}[1]{{\tt beginPB}({#1})}
\newcommand{\lbeginpb}{{\sf beginPB}}
\newcommand{\lendpb}{{\sf endPB}}
\newcommand{\beginpblab}[1]{\lbeginpb({#1})}
\newcommand{\myendpblab}[1]{\lendpb({#1})}

\newcommand{\makemodel}[1]{\ensuremath{{\mathsf{#1}}}\xspace}

\newcommand{\SC}{\makemodel{SC}}

\newcommand{\makeP}[1]{{\ensuremath{\mathsf{P}}{#1}}\xspace}

\newcommand{\PSC}{\makeP{\SC}}

\newcommand{\PTSO}{\makemodel{Px86}}

\newcommand{\progstate}{\overline{q}}
\newcommand{\sprogstate}{q}
\newcommand{\memstate}{{M}}

\newcounter{mylabelcounter}

\makeatletter
\newcommand{\labelAxiom}[2]{%
\hfill{\normalfont\textsc{(#1)}}\refstepcounter{mylabelcounter}
\immediate\write\@auxout{%
  \string\newlabel{#2}{{\unexpanded{\normalfont\textsc{#1}}}{\thepage}{{\unexpanded{\normalfont\textsc{#1}}}}{mylabelcounter.\number\value{mylabelcounter}}{}}
}%
}
\makeatother

\newenvironment{myitemize}
{\begin{list}
    {$\bullet$}
    {
      \setlength{\itemsep}{0pt}
      \setlength{\parsep}{0pt}
      \setlength{\topsep}{1pt}
      \setlength{\partopsep}{0pt}
      \setlength{\leftmargin}{1.2em}
      \setlength{\labelwidth}{0.5em}
      \setlength{\labelsep}{0.4em}
    } %
}
{\end{list}}

\newcounter{mycounter}  %
\newenvironment{myenum}
  {\begin{list}
      {\arabic{mycounter}.} %
      {\usecounter{mycounter}   %
        \setlength{\itemsep}{0pt}
        \setlength{\parsep}{0pt}
        \setlength{\topsep}{1pt}
        \setlength{\partopsep}{0pt}
        \setlength{\leftmargin}{1.2em}
        \setlength{\labelsep}{0.4em}
      } %
  }
  {\end{list}}

\definecolor{DarkGreen}{rgb}{0.05, 0.45, 0.05}
\newcommand{\hide}[1]{}

\newcommand{\A}{{A}}

\newcommand{\sprog}{S}
\newcommand{\prog}{\mathit{Pr}}
\newcommand{\locset}{{X}}

\newcommand{\nvlocset}{{\dot{X}}}
\newcommand{\nvlocseta}{{\dot{Y}}}
\newcommand{\loc}{{x}}
\newcommand{\loca}{{y}}

\newcommand{\nvloc}{{\dot{x}}}
\newcommand{\nvloca}{{\dot{y}}}
\newcommand{\vloc}{{\tilde{x}}}
\newcommand{\vloca}{{\tilde{y}}}
\newcommand{\tid}{{\tau}}
\newcommand{\tida}{{\pi}}

\newcommand{\lab}{{l}}
\newcommand{\val}{v}
\newcommand{\vala}{u}

\renewcommand{\exp}{{e}}
\newcommand{\reg}{{r}}
\renewcommand{\instr}{\mathit{inst}}

\newcommand{\pc}{\mathit{pc}}
\newcommand{\lpc}{\mathtt{pc}}

\newcommand{\lphi}{\mathrm{\phi}}

\newcommand{\tidlab}[2]{\tup{{#1},{#2}}}

\newcommand{\asteptidlab}[3]{{}\mathrel{\raisebox{-0.8pt}{\ensuremath{\xrightarrow{#1,#2}}}_{#3}}{}}

\makeatletter
\newcommand{\mylabel}[2]{#2\def\@currentlabel{#2}\label{#1}}
\makeatother

\newcommand{\rulename}[1]{{\textsc{{#1}}}}

\newcommand{\Nthreads}{\mathsf{N}}
\newcommand{\ctid}[1]{\mathtt{T}_#1}
\newcommand{\cloc}[1]{\mathtt{
\ifthenelse{\equal{#1}{1}}{x}{
\ifthenelse{\equal{#1}{2}}{y}{
\ifthenelse{\equal{#1}{3}}{z}{
\ifthenelse{\equal{#1}{4}}{w}{
\problem}}}}}}

\newcommand{\cnvloca}[1]{\dot{\ensuremath{#1}}}
\newcommand{\cnvloc}[1]{\cnvloca{\cloc{#1}}}
\newcommand{\cvloca}[1]{\tilde{\ensuremath{#1}}}
\newcommand{\cvloc}[1]{\cvloca{\cloc{#1}}}

\newcommand{\creg}[1]{\mathtt{
\ifthenelse{\equal{#1}{1}}{a}{
\ifthenelse{\equal{#1}{2}}{b}{
\ifthenelse{\equal{#1}{3}}{c}{
\ifthenelse{\equal{#1}{4}}{d}{
\ifthenelse{\equal{#1}{5}}{e}{
\ifthenelse{\equal{#1}{6}}{f}{
\problem}}}}}}}}

\newcommand{\cval}[1]{\mathtt{
\ifthenelse{\equal{#1}{1}}{v_1}{
\ifthenelse{\equal{#1}{2}}{v_2}{
\ifthenelse{\equal{#1}{3}}{v_3}{
\ifthenelse{\equal{#1}{4}}{v_4}{
\problem}}}}}}

  {\list{}{\leftmargin=#1\rightmargin=#1}\item[]}%
  {\endlist}

\newcommand{\lmem}{\ensuremath{\dot{\mathtt{m}}}}
\newcommand{\mem}{\ensuremath{\dot{m}}}

\newcommand{\lvmem}{{\tilde{\mathtt{m}}}}
\newcommand{\vmem}{{\tilde{m}}}

\newcommand{\pbuff}{\ensuremath{\mathit{p}}}
\newcommand{\Pbuff}{\ensuremath{\mathit{P}}}
\newcommand{\lPbuff}{\ensuremath{\mathtt{P}}}

\newcommand{\epsl}{\ensuremath{\epsilon}\xspace}
\newcommand{\crash}{\lightning}

\newcommand{\tr}{t}

\newcommand{\traces}[1]{\mathsf{traces}({#1})}

\newcommand{\seq}{\mathbin{;}}

\newcommand{\cs}[2]{{#1}{\Join}{#2}}

\newcommand{\diffemph}[1]{{\tightshadetext{\ensuremath{#1}}}}

\definecolor{darkturquoise}{rgb}{0.012, 0.502, 0.486}
\definecolor{lilac}{rgb}{0.580, 0.341, 0.922}
\definecolor{StringRed}{rgb}{.637,0.082,0.082}
\definecolor{CommentGreen}{rgb}{0.0,0.55,0.3}
\definecolor{KeywordBlue}{rgb}{0.0,0.3,0.55}
\definecolor{LinkColor}{rgb}{0.55,0.0,0.3}
\definecolor{CiteColor}{rgb}{0.55,0.0,0.3}
\definecolor{HighlightColor}{rgb}{0.0,0.0,0.0}
\definecolor{grey}{rgb}{0.5,0.5,0.5}
\definecolor{darkgrey}{rgb}{0.4,0.4,0.4}
\definecolor{red}{rgb}{1,0,0}
\definecolor{darkgreen}{rgb}{0.0,0.7,0.0}
\definecolor{mydarkgreen}{rgb}{0.0,0.3,0.0}
\definecolor{darkblue}{rgb}{0.0,0.0,0.5}
\definecolor{darkred}{rgb}{0.7,0.0,0.0}
\definecolor{mygrey}{rgb}{0.7, 0.7, 0.7}
\definecolor{commentgreen}{rgb}{0, 0.3, 0}
\definecolor{darkred}{rgb}{0.5, 0, 0}
\definecolor{nicerhighlightcolor}{HTML}{F0E0F0}
\definecolor{colorPROMOTED}{rgb}{0.906, 0.161, 0.541}
\definecolor{phlightcolor}{rgb}{0.45,0.95,0.78}

\newcommand{\tightshadetext}[2][nicerhighlightcolor]{\setlength{\fboxsep}{0pt}\colorbox{#1}{#2}}

\edef\myindent{\the\parindent}%

\newcommand{\myhrule}{{\color{lightgray}\hrule}}

\newcommand{\lCALL}{{\mathtt{CALL}}}
\newcommand{\lRET}{{\mathtt{RET}}}

\newcommand{\calllab}[2]{{\lCALL}({#1},{#2})}
\newcommand{\retlab}[2]{{\lRET}({#1},{#2})}

\newcommand{\method}{\mathit{f}}
\newcommand{\lmethod}{\mathtt{f}}
\newcommand{\methodset}{{\mathit{F}}}

\newcommand{\lib}{L}
\newcommand{\stub}{\textrm{stub}}
\newcommand{\libstub}{\lib_\stub}

\newcommand{\main}{\mathsf{main}}
\newcommand{\recoverMethod}{\mathsf{recover}}

\newcommand{\pcs}{\pc_\mathsf{s}}
\newcommand{\lpcs}{\mathtt{pc_\mathsf{s}}}

\newcommand{\client}[3][]{
  \ifthenelse{\equal{#1}{}}{
    {{#2}[{#3}]}
  }{
    {{{#2}[{#3}]}\shortparallel{#1}}
  }
}
\newcommand{\MGCn}{\mathit{MGC}}
\newcommand{\MGCfree}{\MGCn_\mathsf{free}}
\newcommand{\MGCrec}{\MGCn_\mathsf{rec}}
\newcommand{\MGC}[2][]{
\client[#1]{\MGCn}{#2}
}

\newcommand{\MethodNames}{\mathsf{F}}

\newcommand{\pblock}{B}
\newcommand{\lpblock}{\ensuremath{\mathtt{B}}}
\newcommand{\pbid}{j}%
\newcommand{\pbidSet}{{\mathit{Bid}}}
\newcommand{\lpbidSet}{\ensuremath{\mathtt{Bid}}}

\newcommand{\ipbwlab}[2]{{#1}{:}{\lW}({#2})}
\newcommand{\pbwlab}[1]{{\lW}({#1})}

\newcommand{\callEv}[3]{
  \ifthenelse{\equal{#3}{}}{
    ({#1}, {\kw{call}({#2})})
  }{
    ({#1}, {\kw{call}({#2},{#3})})
  }
}

\newcommand{\stateType}{Q}
\newcommand{\initStateType}{\ltsstate_\Init}
\newcommand{\labelType}{\Sigma}
\newcommand{\transType}{T}

\newcommand{\selStates}{\lQ}
\newcommand{\selInitStates}{\linit}
\newcommand{\selLabels}{\lSigma}
\newcommand{\selTrans}{\lT}

\newcommand{\ltsstate}{q}

\newcommand{\ltslabel}{\sigma}

\newcommand{\rsthistory}[2]{\mathsf{H}_{#1}(#2)}
\newcommand{\fhistory}[1]{{\mathsf{H}(#1)}}
\newcommand{\history}{h}

\newcommand{\usedlocs}[1]{\Loc({#1})}
\newcommand{\usedclientlocs}[2]{\Loc({#2}\setminus{#1})}
\newcommand{\unusedlocs}[1]{\Loc\setminus\usedlocs{#1}}

\newcommand{\mergemem}[4]{\tup{#1,#2}\uplus\tup{#3,#4}}
\newcommand{\nmergemem}[4]{
\mergemem{{#1}\rst{#2}}{#2}{{#3}\rst{#4}}{#4}
}

\renewcommand{\sharp}{\texttt{\#}}

\newcommand{\cparagraph}[1]{\vspace{0.4em}\noindent\textbf{#1.}}
\newcommand{\citet}[1]{\cite{#1}}

\newif\iflong
\longtrue %

\newif\ifdraft
\draftfalse %

\ifdraft
\usepackage{showframe}
\fi

\newcommand{\mytitle}{Abstraction for Crash-Resilient Objects}
\iflong\renewcommand\mytitle{Abstraction for Crash-Resilient Objects (Extended Version)}\fi
\ifdraft\iflong\renewcommand\mytitle{Abstraction for Crash-Resilient Objects}\fi\fi

\newcommand\mycodeoffset{-2ex}
\newcommand\mycodefactor{0.92}
\SetSymbolFont{stmry}{bold}{U}{stmry}{m}{n}
\begin{document}

\title{\mytitle\thanks{This research was supported by the Israel Science
Foundation (grants 1566/18 and 2005/17) and by the European Research Council
(ERC) under the European Union’s Horizon 2020 research and innovation programme
(grant agreement no. 851811). Additionally, the first author was supported by
the Blavatnik Family Foundation, and the second by the Alon
Young Faculty Fellowship.}}

\author{Artem Khyzha\thanks{Now at Arm Ltd, Cambridge, UK}\orcidID{0000-0002-6781-9665} \and Ori Lahav(\Envelope)\orcidID{0000-0003-4305-6998}}
\institute{Tel Aviv University, Tel Aviv, Israel}
\authorrunning{A. Khyzha and O. Lahav}
\maketitle              %

\begin{abstract}
We study abstraction for crash-resilient concurrent objects using non-volatile memory (NVM). We develop a library-correctness criterion that is sound for ensuring contextual refinement in this setting, thus allowing clients to reason about library behaviors in terms of their abstract specifications, and library developers to verify their implementations against the specifications abstracting away from  particular client programs. As a semantic foundation we employ a recent NVM model, called Persistent Sequential Consistency, and extend its language and operational semantics with useful specification constructs. The proposed correctness criterion accounts for NVM-related interactions between client and library code due to explicit persist instructions, and for calling policies enforced by libraries. We illustrate our approach on two implementations and specifications of simple persistent objects with different prototypical durability guarantees. Our results provide the first approach to formal compositional reasoning under NVM.

\end{abstract}

\newcommand{\edited}[1]{{\color{blue}{#1}}}

\section{Introduction}
\label{sec:intro}

Non-volatile memory, or NVM for short, is an emerging
technology that enables byte addressable and high performant storage
alongside with data persistency across system crashes.
This combination of features allows researchers and practitioners to
develop a variety of efficient crash-resilient data structures (see, \eg~\cite{Friedman-persistent-queue,durable-sets}).
Recently, NVM has started to become available in commodity architectures
of manufacturers such as Intel and ARM~\cite{intel-manual,ARMv8},
and formal (operational and declarative) models of these systems
have been proposed~\cite{pxes-popl,Khyzha2021,Kyeongmin2021}.

Unfortunately, like other new technologies, NVM puts more burden on programmers.
Indeed, to get close to the performance of DRAM,
writes to the NVM are first kept in volatile (\ie losing contents upon crashes) caches,
and only later persist (\ie propagate to the NVM),
possibly not in the order in which they were issued.
This results in counterintuitive behaviors even for sequential programs
and requires careful management using barriers of different kinds, \aka explicit persist instructions,
for guaranteeing that the system recovers to a consistent state upon a failure.
Combined with standard concurrency issues,
programming on such machines is highly challenging.

To tackle the complexity and make NVM widely applicable,
one would naturally want to draw on libraries
encapsulating highly optimized concurrent crash-resilient data structures
(\aka persistent objects).
This approach goes both ways:
programmers should be able to reason about their code
using abstract library specifications that hide the implementation details,
and in turn, library developers should be able to verify ``once and for all'' their implementations against their specifications abstracting away from a particular client program.
From a formal standpoint, this indispensable modularity requires us to have a so-called \emph{(library) abstraction theorem}:
a correctness condition that guarantees the soundness of client reasoning that assumes the specification instead of the implementation.
Put differently, the abstraction theorem should allow one to establish \emph{contextual refinement},
\ie conclude that the specification reproduces the implementation's client-observable behaviors
under any (valid) context.
To the best of our knowledge, while several correctness criteria for persistent objects, akin to classical linearizability,
have been proposed and have been established for multiple sophisticated implementations,
none of them has been formally related to contextual refinement by an abstraction theorem of this kind.

In this paper we formulate and prove an abstraction theorem
for concurrent programs utilizing non-volatile memory.
We target the ``Persistent Sequential Consistency'' model of~\cite{Khyzha2021}, or \PSC,
which enriches the standard sequentially consistent shared-memory with non-volatile storage
using per-location FIFO buffers to account for delayed and out-of-order persistence of writes.
\PSC constitutes a relatively simple model that is very close to developer's informal understanding of NVM.
While existing hardware does not implement \PSC as is, \citet{Khyzha2021} presented %
compiler mappings from \PSC to x86 (based on its persistency model from~\cite{pxes-popl}),
which can be used to ensure \PSC semantics on Intel machines.
Directly supporting relaxed memory models is left for future work.

\iflong

\else

\cparagraph{Auxiliary material}
An extended version, including proofs of theorems stated in the paper, is available at
\url{https://arxiv.org/abs/2111.03881}.
\fi

\section{Key Challenges and Ideas}
\label{sec:key}

We outline the main challenges
and the key ideas in our solutions.
We keep the discussion informal, leaving the formal development
to later sections.

\subsection{Library Specifications}
\label{sec:spec}

A choice of a formalism for specifying library behaviors is integral in
stating a library abstraction theorem. For libraries of concurrent data
structures (\aka concurrent objects), a popular approach is to give
specifications in terms of sequential objects with the help of the classical
notion of linearizability~\cite{herlihy-lin}, which requires every sequence of
method calls and returns that is possible to produce in a concurrent program
to correspond to a sequence that can be generated by the sequential object. In this
approach, a sequential object, represented by a set
of sequences of pairs of method invocations and their associated responses,
constitutes the library specification. Then, abstraction allows the client to
reason about calls to a concurrent library as if they execute atomically on a single thread,
or, equivalently, protected by a global lock~\cite{Filipovic10,Bouajjani15}.

For libraries of crash-resilient objects, there is more than one natural way of
interpreting sequential specifications and adapting the linearizability definition,
and no single notion of correctness \wrt sequential specifications
captures all different options. %
A crash-resilient object may ensure that all methods completed by the moment of crash survive through it, or that some prefix of them does.
It may also choose different possibilities for methods in progress at the moment of crash
(whether they are allowed take their effect at some later point after the crash or not).
Multiple adaptations of linearizability have been proposed,
each relating crash-resilient objects to sequential specifications in a different way.
 This includes:
strict linearizability~\cite{aguilera2003strict},
persistent atomicity~\cite{Guerraoui04},
and durable linearizability and its buffered variant~\cite{persistent-lin}.
Among them, buffered durable linearizability, which allows for efficient implementations,
ended up not being compositional, which means that it may happen that two (non-interacting)  libraries
are both correct, but their combination is not.
In fact, since each of the different notions is useful for particular objects,
one may naturally want to mix different correctness notions in a single client program.
This would force the client to reason with several alternatives
for interpreting sequential specifications, and to make sure that they
compose well with one another.

To approach this variety, we believe it is necessary to follow a different approach,
which is standard in concurrent program verification~(see, \eg \cite{Hawblitzel15,Gu15,Liang14}),
and was applied before for deriving abstraction theorems in different contexts~\cite{
Gotsman11_liveness,Gostman_13_ownership,Burckhardt_2012_library_tso}.
The idea is to take a library's specification to be just another library,
where the latter is intended to have a simpler implementation.
Then, we define a \emph{library correctness condition} stating
what it means for  one library $\lib$ to \emph{refine} another library $\lib^\sharp$
(equivalently, for $\lib^\sharp$ to \emph{abstract} $\lib$),
and prove an abstraction theorem that ensures that
when the library correctness condition is met,
the behaviors of any client using $\lib$
are contained in the behaviors of the client using $\lib^\sharp$.
Such a theorem is only useful if the correctness condition avoids
quantification over all possible clients, which would make the theorem trivial.

Using code for specifying libraries has several advantages over
correctness notions based on sequential specifications. First,
specifications and implementations are expressed and reasoned about in a
unified framework, alleviating the need to interpret the use of sequential
specifications by concurrent programs with system failures. Instead, the
client of the theorem replaces complex library code with simpler specification
code, and thus works with the semantics of a single language. Second,
it enables a layered verification technique for library developers, allowing
them to prove library correctness by introducing one or more intermediate
implementations between $\lib$ and $\lib^\sharp$. Finally, this formulation of
the abstraction theorem is compositional (\aka local) by construction, meaning that objects can be specified and verified in isolation.%

Now, ``code as a specification'' is only useful if
the programming language is sufficiently expressive for desirable specifications.
For concurrent objects, ``atomic blocks'', often included in theoretic
programming languages, provide a handy specification construct. For NVM, one
needs a way to govern the persistence similarly, offering intuitive
specifications for libraries that simplify client reasoning.
For that matter, viewing the out-of-order persistence of writes to different cache
lines as the major source of counterintuitive behaviors in NVM, we propose a
new specification construct, which we call \emph{persistence blocks}. Roughly
speaking, such blocks may only persist in their entirety, so that persistence
blocks ensure an ``all-or-nothing'' persistency behaviors to the writes they
protect.

For example,
when recovering after a crash during a run of the tiny program
$\writeInst{\cnvloc{1}}{1} \sep \writeInst{\cnvloc{2}}{1}$,\footnote{
We use ``overdots'' to denote non-volatile variables.
We assume that all variables are initialized to $0$
and that $\cnvloc{1}$ and $\cnvloc{2}$ lie on different cache lines.}
due to out-of-order persistence (writes to different cache lines
are not guaranteed to persist in the order in which there were issued),
we may reach any combination of values satisfying
$\cnvloc{1}\in\set{0,1}\land\cnvloc{2}\in\set{0,1}$.
In turn, if a persistence block is used, as in
$\beginpbInst{\cnvloc{1},\cnvloc{2}} \sep \writeInst{\cnvloc{1}}{1} \sep \writeInst{\cnvloc{2}}{1}
\sep \myendpbInst{\cnvloc{1},\cnvloc{2}}$,
then only $\cnvloc{1}=\cnvloc{2}=0 \lor \cnvloc{1}=\cnvloc{2}=1$
are possible upon recovery.

Our blocks are closely related to persistent transactions of the PMDK library~\cite{pmdk}
(but we avoid the term transaction, since persistence blocks
do not ensure isolation when executed concurrently).
In our technical development, we extend the \PSC model with instructions for persistence blocks,
and carefully construct their semantics (see \cref{sec:psc_ext}) to allow the abstraction result.
We believe that persistence blocks are a useful \emph{specification construct} for
various data structures, where data consistency naturally involves multiple locations
(often, pointers) being in-sync with one another.

\subsection{Client-Library Interaction Using Explicit Persist Instructions}
\label{sec:sfence}

The key to establishing a library abstraction theorem is in decomposing a program into two interacting sub-parts, a client and a library, and understanding
the interactions between them.
These interactions are usually defined in terms of \emph{histories}, taken to be sequences
of method invocations and responses, along with the values being passed.
The library correctness condition (the premise of the abstraction theorem)
requires that histories produced by using a library $\lib$
are also produced by its
specification $\lib^\sharp$ when both libraries are used by a certain ``most general client'' (MGC, for short)
that concurrently invokes arbitrary methods of $\lib$ an arbitrary number of times
with every possible argument.
The abstraction theorem ensures that if the library correctness condition holds,
then $\lib$ refines $\lib^\sharp$ for \emph{any} client.

Thus, for the abstraction theorem to hold, one has to make sure that the
interactions between any client and the library are fully captured in the
history produced by the library when used by the MGC. In \emph{crash-free}
sequentially consistent shared memory semantics, this is ensured by the
standard assumption that the client and the library manipulate disjoint set of
memory locations. Indeed, this restriction guarantees that clients can
communicate with libraries only via values passed to and returned from
method invocations.

\newcommand{\foo}{\mathit{f}}

However, we observe that under NVM, mutual interactions between the client and the library
go beyond passed values,
even when assuming disjointness of memory locations,
which makes the standard notion of a library history insufficient.
As a simple example, consider an interface with just one method $\foo$,
specified by $\lib^\sharp=[\foo \mapsto \sfenceInst\sep\returnInst]$.
The $\sfenceInst$ instruction, called ``store fence'', is an explicit persist instruction
meant to be used in conjunction with optimized barriers called ``flush-optimal'' (denoted by $\kw{fo}$).
Its role is to guarantee the persistence of previous write instructions that are guarded by flush-optimal instructions.
Concretely, under \PSC (following x86),
after a thread executes $\writeInst{\cnvloc{1}}{1} \sep \foInst{\cnvloc{1}} \sep \sfenceInst$,
we know that the write of $1$ to $\nvloc$ has persisted (\ie been propagated to the NVM),
while without the $\sfenceInst$, it may still sit in the volatile part of the memory system.

In turn, consider an implementation $\lib$, given by $\lib=[\foo \mapsto \returnInst]$,
that implements $\foo$ by doing nothing.
Clearly, $\lib$ does not implement $\lib^\sharp$ correctly.
Indeed, for the (sequential) client program
$\writeInst{\cnvloc{1}}{1} \sep \foInst{\cnvloc{1}} \sep \callInst{\foo} \sep \writeInst{\cnvloc{2}}{1}$
that uses $\lib^\sharp$, we have $\cnvloc{2}=1\implies \cnvloc{1}=1$ as a global invariant:
if the system has crashed and we have $\cnvloc{2}=1$ in the NVM, then
the $\sfenceInst$ ensures that $\cnvloc{1}=1$ is in the NVM as well.
Nevertheless, due to out-of-order persistence,
if we use $\lib$ in this program, we may get $\cnvloc{2}=1 \land \cnvloc{1}=0$ after a crash.
Now, the client and the libraries above mention disjoint locations,
and the histories that $\lib$ may produce for the MGC
are exactly the histories that $\lib^\sharp$ produces
(all well-formed sequences of ``call'' and ``return'').
Thus, when inspecting histories of $\lib$ and of $\lib^\sharp$,
we do not have sufficient information to observe the difference between them.

Generally speaking, the challenge stems from the fact that
certain explicit persist instructions ($\sfenceInst$
and other instructions whose implementation in the hardware contains an implicit store fence, such as RMWs in x86),
which can be executed by the library,
impose conditions on the persistence of writes performed by the client
that ran earlier on the same processor.

We address this challenge in two ways. First, we can sidestep the problem by
weakening the semantics of store fences, making them relative to a set of
locations (those used by the library or those used by the client). To do so,
we extend the programming language with a specification construct similar to a store fence,
but only affecting a given set of locations, and we
restrict its use by each component to mention only the locations it owns. The
use of these localized instructions instead of store fences is sufficient to
ensure that the interaction between client and library is fully captured in
histories, and allows us to establish the expected abstraction theorem.
Libraries that do not intend to provide a store fence functionality to their clients
can readily replace store fences with their localized counterparts.
Doing so gives more freedom to alternative implementations of the same specification,
which may, \eg use alternative persist instructions without
the store fence functionality (such as CLFLUSH in \cite{intel-manual}).

On the other hand, it is possible that in performance-critical systems, clients
would like to rely on a store fence that is executed anyway by the library for
the library's own needs. For that, the library developer needs to use a standard
store fence in the library's specification rather than the localized
counterpart, and the abstraction theorem has to handle store fences with their
standard, non-localized semantics. To do so, we expose in histories not only
method invocations and responses, but also store fences. Roughly
speaking, it means that in addition to the standard requirement on values passed
by method invocations and responses, for $\lib$ to refine $\lib^\sharp$, we
would also require that  $\lib$ performs a store fence whenever $\lib^\sharp$
does (which does not hold for the example above). Our notion of history in
\cref{sec:lib} is set to allow store fences (alongside with their weaker
localized versions), and the abstraction theorem in \cref{sec:abs} shows that
these extended histories are expressive enough for defining the
library-correctness condition.

\subsection{Handling Calling Policies}
\label{sec:policy}

The third challenge we address concerns abstraction for libraries that enforce certain calling policies on their clients.\footnote{
This challenge is not particular to NVM, but, interestingly,
to the best of our knowledge, it has not been addressed in previous work establishing
abstraction theorems.}
For instance, a library implementing a lock may require that the calls of each thread
for acquiring and releasing the lock perfectly interleave,
and a library implementing a single-producer queue may require that only one thread is
calling the enqueue method.
In the context of NVM, libraries often demand
that a distinguished \emph{recovery method} is called after every crash
before invoking any other method of the library.
When the client uses the library in a way that violates the calling policy,
the library developer ensures nothing, and the blame is assigned to the client.

In the presence of calling policies,
the contextual refinement guaranteed by the library abstraction theorem, stating
that all behaviors of a program $\client{\prog}{\lib}$ that uses $\lib$ are also
behaviors of the program $\client{\prog}{\lib^\sharp}$ that uses $\lib^\sharp$,
is only applicable for a program $\prog$ that respects the calling policy.
An interesting compositionality question arises: Are we allowed to
assume the library's specification when checking whether a program adheres to the
calling policy (that is, require that $\client{\prog}{\lib^\sharp}$ adheres to
the policy), or should this obligation be satisfied for the library's
implementation (that is, require that $\client{\prog}{\lib}$ adheres to the
policy)?

The latter option would limit the applicability of the abstraction theorem for client reasoning.
Indeed, it may be the case that establishing that $\client{\prog}{\lib}$ adheres to the policy
depends on the implementation $\lib$, whereas the abstraction theorem should allow
reasoning without knowing the implementation \emph{at all}.
On the other hand, the former option seems circular, as it uses contextual refinement
to establish its own precondition.

In this paper we show that requiring that $\client{\prog}{\lib^\sharp}$ adheres to the policy is actually sufficient for
ensuring contextual refinement. Roughly speaking, our proof avoids circular reasoning
by inspecting a \emph{minimal} contextual refinement violation, for which we are able to establish policy adherence
when using $\lib$, given policy adherence when using $\lib^\sharp$.
To the best of our knowledge, this is a novel argument in the context of library abstraction.
It is akin to DRF (data-race freedom) guarantees in weak memory concurrency,
where often programs are guaranteed to have strong semantics (usually, sequential consistency)
provided that certain race-freedom conditions hold in all runs under the \emph{strong} semantics.

We note that many library's calling policies are ``structural'', namely they only enforce certain
ordering constraints on the clients that do not depend on the values returned by the library
(in particular, ``execute recovery first'' is a structural policy).
In these cases, policy adherence holds even for an over-approximation $\libstub$ of $\lib$
that returns arbitrary values. Certainly, however, this is not always the case.
For example, a library $\lib$ implementing standard list methods, $\mathsf{cons}$ and $\mathsf{head}$,
may require that $\mathsf{head}$ is only called on non-empty lists
(like, \eg $\mathsf{pop\_front}$ in C++ that triggers undefined behavior
if applied to an empty list~\cite{cppref}).
Then, invoking $\mathsf{head}$ with the value returned from $\mathsf{cons}$
does adhere to the calling policy, but this is not the case for the over-approximated library $\libstub$,
which allows $\mathsf{cons}$ to return the empty list.

\newcommand{\istrans}{
\begin{mathpar}
\footnotesize
\inferrule*{
\iseq(\pc)=\assignInst{r}{\exp} \\\\ \phi'=\phi[\reg \mapsto \phi(\exp)]
}{
	\tup{\pc,\phi}%
	\astep{\epsl\vphantom{\lab}}_\iseq
	\tup{\pc+1, \phi'}%
}
\hfill
\inferrule*{
\iseq(\pc)=\ifGotoInst{\exp}{n_1 \shortmid\ldots\shortmid n_m} \\\\
\phi(\exp)\neq 0 \implies \pc' \in \set{n_1 \til n_m} \\\\
\phi(\exp)= 0 \implies \pc' = \pc+1
}{
	\tup{\pc,\phi}%
	\astep{\epsl\vphantom{\lab}}_\iseq
	\tup{\pc' , \phi}%
}
\hfill
\inferrule*{
	\iseq(\pc)=\havocInst
}{
	\tup{\pc,\phi}%
	\astep{\epsl\vphantom{\lab}}_\iseq
	\tup{\pc+1, \phi'}%
}
\\
\inferrule*{
	\iseq(\pc)=\writeInst{\loc}{\exp}
	\\\\ \lab=\wlab{\loc}{\phi(\exp)}
}{
	\tup{\pc,\phi}%
	\asteplab{\lab}{\iseq}
	\tup{\pc+1, \phi}%
}
\hfill
\inferrule*{
	\iseq(\pc)=\readInst{\reg}{\loc}
	\\\\ \lab=\rlab{\loc}{\val}
	\\\\ \phi'=\phi[\reg \mapsto \val]
}{	\tup{\pc,\phi}%
	\asteplab{\lab}{\iseq}
	\tup{\pc+1, \phi'}%
}
\hfill
\inferrule*{
	\iseq(\pc)\in
	{\left\lbrace
		\begin{array}{l}
		\flInst{\_}, \foInst{\_}, \\
		\sfenceInst, \pfenceInst{\_},\\
		\beginpbInst{\_}, \myendpbInst{\_}
		\end{array}
	\right\rbrace } \\\\ {\lab = \mathsf{matching\_label}(\iseq(\pc))}
}{
	\tup{\pc,\phi}%
	\asteplab{\lab}{\iseq}
	\tup{\pc+1, \phi}%
}
\end{mathpar}
}

\section{NVM Programs: Syntax and Semantics}
\label{sec:programs}

In this section we begin to present the formal settings for our results.
As standard in memory models, it is convenient to break the operational semantics
into:
a \emph{program} semantics (\aka thread subsystem)
and a \emph{memory} semantics.
We represent both components as labeled transition systems whose transition labels
correspond to the operations they perform.
We then consider the synchronized runs of the program and the memory,
where
program actions that interact with the memory
are matched by actions executed by the memory system (see \cref{sec:systems}).

Next,
we focus on the program part of the semantics,
presenting both syntax (\cref{sec:program_syntax})
and semantics (\cref{sec:program_semantics}).
We use the following standard notations.

\cparagraph{Notation for finite sequences}
For a finite alphabet $\Sigma$, we denote by $\Sigma^*$ (respectively,
$\Sigma^+$) the set of all (non-empty) sequences over $\Sigma$. We
use $\epsilon$ to denote the empty sequence. The length of a sequence $s$ is
denoted by $\size{s}$. %
We often identify
sequences with their underlying functions (whose domain is $\set{1 \til
\size{s}}$), and write $s(k)$ for the symbol at position $1 \leq k \leq
\size{s}$ in $s$. We write $\sigma\in s$ if $\sigma$ appears in $s$, that is
if $s(k)=\sigma$ for some $1 \leq k \leq \size{s}$.
We use ``$\cdot$'' for concatenating sequences,
and identify symbols with sequences of length $1$.

\subsection{Program Syntax}
\label{sec:program_syntax}

The domains and metavariables used to range over them are as follows:

\begin{tabular}{rll}
\emph{values} & $\quad$ & $\val, \vala \in \Val = \set{0,1,2,\ldots}$ %
\\ \emph{shared non-volatile variables} & $\quad$ & $\nvloc, \nvloca \in \NVLoc = \set{\cnvloc{1},\cnvloc{2},\ldots}$
\\ \emph{shared volatile variables} & $\quad$ & $\vloc, \vloca \in \VLoc = \set{\cvloc{1},\cvloc{2},\ldots}$
\\ \emph{shared variables} & $\quad$ & $\loc, \loca \in \Loc = \NVLoc \cup \VLoc$
\\ \emph{register names} & $\quad$ & $\reg \in \Reg = \set{\creg{1},\creg{2},\ldots}$ %
\\ \emph{thread identifiers} & $\quad$ & $\tid, \tida \in \Tid = \set{\ctid{1},\ctid{2}\til \ctid{\Nthreads}}$
\\ \emph{method names} & $\quad$ & $\method \in \MethodNames \qquad \main\nin\MethodNames$ %
\end{tabular}

\noindent
Thus, there are three kinds of variables: shared non-volatile, shared volatile,
and thread-local ones (called registers), which are also volatile.
A distinguished name $\main$ is reserved for the starting point of the program execution.

For concreteness, we present a simple programming-language syntax.
Its expressions and instructions are given by the following grammar:\footnote{\iflong
In \cref{app:rmw}, we present the required extensions for including read-modify-write instructions.
\else
In the extended version of this paper, we
also include read-modify-write instructions.\fi
}
\[\smaller\smaller
	\begin{array}{@{} l l @{}}
\exp ::=  & \reg \ALT \val \ALT
\exp + \exp \ALT \exp = \exp \ALT \exp \neq \exp \ALT \ldots
\\[0.5ex]
\instr ::=  &
\assignInst{\reg}{\exp}
\ALT \ifGotoInst{\exp}{n_1 \shortmid\ldots\shortmid n_m}
\ALT \havocInst
\ALT
\writeInst{\loc}{\exp}
\ALT \readInst{\reg}{\loc}  \\&
\ALT
\flInst{\nvloc}
\ALT \foInst{\nvloc}
\ALT \sfenceInst
\ALT
\callInst{\method}
\ALT \returnInst \\&
\ALT
\pfenceInst{\nvlocset}
\ALT \beginpbInst{\nvlocset}
\ALT \myendpbInst{\nvlocset}
\end{array}
\]

Expressions are constructed with arithmetic and
boolean operations over registers and values.
Instructions consist of a local assignment
$\assignInst{\reg}{\exp}$;
a conditional $\ifGotoInst{\exp}{n_1 \shortmid\ldots\shortmid n_m}$ for
non-deterministically jumping to a program counter from $\set{n_1 \til n_m}$
when $\exp$ evaluates to non-zero or, otherwise, skipping ($\GotoInst{n_1
\shortmid\ldots\shortmid n_m}$ can be encoded as $\ifGotoInst{1}{n_1
\shortmid\ldots\shortmid n_m}$); $\havocInst$ for arbitrarily modifying all
registers; a write to memory $\writeInst{\loc}{\exp}$;
and a read from memory $\readInst{\reg}{\loc}$.
There are also explicit persist instructions: a \emph{flush}
instruction $\flInst{\nvloc}$ and its optimized version $\foInst{\nvloc}$,
called \emph{flush-optimal} (referred to as {CLFLUSH} and
{CLFLUSHOPT} in \cite{intel-manual}),
as well as the store fence instruction $\sfenceInst$ (see \cref{sec:sfence}).

This standard instruction set is extended to support calling and specifying
library methods. There is a call instruction $\callInst{\method}$ and a return
instruction $\returnInst$.
The novel specification constructs include
the \emph{local store fence} instruction
$\pfenceInst{\nvlocset}$ that relaxes the
semantics of $\sfenceInst$ by only enforcing the persistence ordering for the
given set $\nvlocset$ of variables (thus, $\pfenceInst{\NVLoc}$ is equivalent
to $\sfenceInst$); and instructions to begin and end a {\em
persistence block}, $\beginpbInst{\nvlocset}$ and
$\myendpbInst{\nvlocset}$, respectively. The persistence block demarks
the writes that need to persist simultaneously
after the block ends, either non-deterministically or
triggered by a flush on some variable in $\nvlocset$.

Next, we employ three syntactic categories:
\begin{myitemize}
\item \emph{Instruction sequences} represent the (sequential) implementation
of each method (including $\main$).
Formally, an instruction sequence $\iseq$ is a function from a non-empty finite domain of the form
$\set{0\til n}$ (representing the possible program counters) to the set of instructions. We say that an instruction sequence is \emph{flat} if
it does not include an instruction of the form $\callInst{\_}$.
\item \emph{Sequential programs} consist of a ``main'' method
accompanied with implementations of every method $\method\in\MethodNames$.
Formally, a sequential program $\sprog$ is a function assigning an instruction
sequence to every $\method \in \set{\main } \cup  \MethodNames$. To avoid
modeling a call stack and simplify the presentation, we require that
$\sprog(\method)$ is a flat instruction sequence for every $\method \in \MethodNames$.
\item \emph{Concurrent programs}
are top-level parallel compositions of sequential programs,
all accompanied by the same method implementations.
Formally, a (concurrent) program $\prog$ is a mapping assigning a sequential program to every $\tid\in\Tid$,
with $\prog(\tid)(\method)=\prog(\tida)(\method)$ for every $\tid,\tida\in\Tid$
and $\method\in \MethodNames$.
Below, we write $\prog(\method)$ for
$\prog(\ctid{1})(\method)$.
\end{myitemize}

\subsection{Program Semantics}
\label{sec:program_semantics}

We give semantics to
the syntactic objects
using
labeled transition systems.

\begin{definition}%
{\rm
A \emph{labeled transition system} (LTS) is a tuple
$\A=\tup{\labelType,\stateType,\initStateType,\transType}$,
where
$\labelType$ is a set of \emph{transition labels},
$\stateType$ is a set of \emph{states},
$\initStateType \in \stateType$ is the \emph{initial state},
and $\transType \suq \stateType\times \labelType \times \stateType$ is a set of \emph{transitions}.
We often write $\ltsstate\astep{\ltslabel}\ltsstate'$
to denote a transition $\tup{\ltsstate,\ltslabel,\ltsstate'}$.
We denote by $\A.\selLabels$, $\A.\selStates$, $\A.\selInitStates$, and
$\A.\selTrans$ the components of an LTS $\A$.
We write $\asteplab{\ltslabel}{\A}$
for the relation $\set{\tup{\ltsstate,\ltsstate'} \mid
\ltsstate\astep{\ltslabel}\ltsstate' \in \A.\lT}$ and $\asteplab{}{\A}$ for
$\bigcup_{\ltslabel\in\labelType} \asteplab{\ltslabel}{\A}$. For a sequence $\tr \in
\A.\lSigma^*$, we write $\asteplab{\tr}{\A}$ for the composition
$\asteplab{\tr(1)}{\A} \seq \ldots \seq \asteplab{\tr(\size{\tr})}{\A}$.
A~sequence $\tr \in \A.\lSigma^*$ such that $\A.\linit \asteplab{\tr}{\A}
\ltsstate$ for some $\ltsstate\in\A.\selStates$ is called a
\emph{trace} of $\A$. %
We denote by $\traces{\A}$ the
set of all traces of $\A$. A state $\ltsstate \in \A.\selStates$ is called
\emph{reachable} in $\A$ if $\A.\linit \asteplab{\tr}{\A} \ltsstate$ for some
 $\tr \in \traces{\A}$.
}\end{definition}

Next, we define the LTSs induced by instruction sequences,
sequential programs, and concurrent programs.
We will often identify the syntactic objects with the LTS they induce
(\eg when writing expressions like $\sprog.\lQ$ for a sequential program $\sprog$).
The transition labels of these LTSs feature \emph{action labels}.

\begin{definition}
\label{def:label}
{\rm
An \emph{action label} takes one of the following forms:
a read $\rlab{\loc}{\val}$,
a write $\wlab{\loc}{\val}$,
a flush $\fllab{\nvloc}$,
a flush-opt $\folab{\nvloc}$,
an sfence $\sflab$,
a local sfence $\pflab{\nvlocset}$,
a start $\beginpblab{\nvlocset}$
or an end $\myendpblab{\nvlocset}$ of a persistence block,
a call $\calllab{\method}{\phi}$, or
a return $\retlab{\method}{\phi}$,
where $\loc \in \Loc$, $\val\in \Val$,
$\nvloc\in\NVLoc$,
$\nvlocset\suq\NVLoc$, $\method \in \MethodNames$,  and $\phi : \Reg \to \Val$.
We denote by $\Lab$ the set of all action labels. The functions
$\lTYP$ and $\lLOC$
retrieve (when applicable) the type
($\lR/\lW/\ldots$) and
variable ($\loc$ or $\nvloc$) of an action label.
We write $\lLOCSET(\lab)$ for the set of variables mentioned in $\lab$
(\eg $\lLOCSET(\rlab{\loc}{\val})=\set{\loc}$, $\lLOCSET(\pflab{\nvlocset})=\nvlocset$,
and $\lLOCSET(\sflab)=\emptyset$).
}\end{definition}

Action labels
represent the interactions that a program has
with the memory.

\begin{definition}\label{def:lts_iseq}
{\rm
The LTS induced by an instruction sequence $\iseq$ is given by:
\begin{myitemize}
\item The transition labels are action labels, extended with $\epsl$ for silent transitions.
\item The states are pairs
$\tup{\pc,\phi}$
where $\pc\in\N$, called \emph{program counter}, stores the
current instruction pointer inside the sequence,
and $\phi:\Reg \to \Val$, called \emph{local store},
records the values of the registers.
We assume that local stores
are extended to expressions in the obvious way.
\item The initial state is $\tup{0,\phi_\Init}$,
where $\phi_\Init \defeq \lambda \reg \ldotp 0$.
\item The transitions are as follows:
\istrans
\end{myitemize}
}
\end{definition}

Recall that program semantics is separate from memory semantics, which is why
the transitions above completely ignore the restrictions arising from the
memory system. In particular, the write to memory $\writeInst{\loc}{\exp}$
only announces itself in the label. The read
from memory $\readInst{\reg}{\loc}$ loads an arbitrary value $\val$ into the
destination register $\reg$, announcing that value in the read label.
Other instructions
act as no-ops, and simply announce
themselves in the transition label, using
the function $\mathsf{matching\_label}$ that maps each instruction
to its %
label
($\flInst{\nvloc}\mapsto \fllab{\nvloc},
\foInst{\nvloc}\mapsto \folab{\nvloc}$, and so on).

Finally, $\callInst{\method}$ and $\returnInst$ instructions are not handled in this level,
but receive special semantics at the level of sequential programs, as defined next.

\begin{definition}\label{def:lts_sprog}
	{\rm
The LTS induced by a sequential program $\sprog$ is given by:
\begin{myitemize}
\item The transition labels are action labels, extended with $\epsl$ for silent transitions.

\item The states are tuples
$\sprogstate=\tup{\pc,\phi,\pcs,\method}$, where:
\begin{itemize}[leftmargin=*]
\item $\tup{\pc,\phi}$ is a state of the instruction sequence (see \cref{def:lts_iseq})
storing the state of the sequence currently running.
\item $\pcs \in \N \cup \set{\bot}$,
called \emph{the stored program counter}, is used to remember the program position to jump to when the
current instruction sequence returns, whereas
$\pcs=\bot$ means that the main method is currently running.
(Recall that we  assume that $\sprog(\method)$ is flat for every $\method \in \MethodNames$,
so we do not need to record the call stack.)
\item $\method \in \MethodNames \cup \set{\main}$, called \emph{the active method},
tracks the method that is currently running.
\end{itemize}
We denote by $\sprogstate.\lpc$, $\sprogstate.\lphi$, $\sprogstate.\lpcs$,
and $\sprogstate.\lmethod$
the components of a state $\sprogstate \in \sprog.\lQ$.

\item The initial state is
$\tup{0,\phi_\Init,\bot,\main}$.

\item The transitions are given by:
\begin{mathpar}
\inferrule[normal]{
\lab_\epsl \in \Lab \cup \set{\epsl}\\
\method\in\set{\main}\cup\MethodNames
\\\\ \tup{\pc,\phi}	\asteplab{\lab_\epsl}{\sprog(\method)} 	\tup{\pc', \phi'}
}{
	\tup{\pc,\phi,\pcs,\method}
	\asteplab{\lab_\epsl}{\sprog}
	\tup{\pc', \phi',\pcs,\method}
}
\and
\inferrule[call]{
	\sprog(\main)(\pc) = \callInst{\method}
		\\\\ \lab = \calllab{\method}{\phi}
}{
	\tup{\pc,\phi,\bot,\main}
	\asteplab{\lab}{\sprog}
	\tup{0,\phi, \pc+1,\method}
}
\and
\inferrule[return]{
	\sprog(\method)(\pc) = \returnInst
	\\ \lab = \retlab{\method}{\phi}
}{
	\tup{\pc,\phi, \pcs, \method}
	\asteplab{\lab}{\sprog}
	\tup{\pcs, \phi, \bot, \main}
}
\and
\inferrule[non-det-sfence]{
	\lab = \sflab
}{
	\tup{\pc,\phi, \pcs, \method}
	\asteplab{\lab}{\sprog}
	\tup{\pc,\phi, \pcs, \method}
}
\end{mathpar}
\end{myitemize}
}
\end{definition}

The \rulename{normal} transition lifts the instruction-sequence transition to the level of sequential programs.
Note that the transition applies %
for any method ($\main$ or other).
The \rulename{call} transition passes control from the main method to some other method,
jumping the program counter to the first instruction and storing the return point ($\pc+1$).
The \rulename{return} transition passes control back using the stored return point.
For simplicity, we do not have any argument passing mechanism and use the full register store for that matter.
(If needed, each component may store the values it needs in the memory, and reload them later on.)

Finally, \rulename{non-det-sfence} is a non-standard transition that we find technically convenient to have.
It allows the program to non-deterministically execute an sfence at any point.
Since, as will become apparent when presenting the memory system, sfences only restrict the possible behaviors,
this transition is safe to include in the program semantics.
It is particularly useful for simplifying the library correctness condition
that only considers inclusion of sets of histories (see \cref{sec:lib}).
For instance,  switching the roles of $\lib$ and $\lib^\sharp$ from \cref{sec:sfence},
the library implementing $\foo$ using $\sfenceInst$
should be considered a refinement of the one that simply returns.
For that, we allow the no-op specification to perform non-deterministic sfences that match the ones
executed by the concrete implementation.

Finally, the LTS induced by a concurrent program is defined as follows.

\begin{definition}
{\rm
The LTS induced by a (concurrent) program $\prog$ is given by:

\begin{myitemize}
\item The set of transition labels is given by
$(\Tid \times (\Lab \cup \set{\epsl})) \cup \set{\crash}$.
The functions on action labels (\eg $\lTYP$, $\lLOC$)
are lifted to these labels in the obvious way.

\item The states, %
denoted by $\progstate$, assign %
a state in $\prog(\tid).\lQ$ to every $\tid\in\Tid$.

\item The initial state is composed from the initial state of each thread:
\\
$\progstate_\Init \defeq \tup{\prog(\ctid{1}).\linit \til \prog(\ctid{\Nthreads}).\linit}$.
\item The transitions
are interleaved thread transitions
or crash transitions reinitializing the program state:
\vspace{-0.6ex}
\begin{mathpar}
\inferrule*[left=normal]{
\lab_\epsl \in \Lab \cup \set{\epsl}\\
\progstate(\tid) \asteplab{\lab_\epsl}{\prog(\tid)} \sprogstate'
}{\progstate \asteptidlab{\tid}{\lab_\epsl}{\prog} \progstate[\tid \mapsto \sprogstate']}
\and
\inferrule*[left=crash]{
}{\progstate \asteplab{\crash\vphantom{\lab}}{\prog} \progstate_\Init}
\end{mathpar}
\end{myitemize}
}
\end{definition}

\section{The \PSC Memory System}
\label{sec:psc}

We present \PSC (``Persistent Sequential Consistency''),
the persistency model used as the memory system.
We first introduce the model as it is in \cite{Khyzha2021}
(extended with standard volatile memory alongside with the non-volatile one),
following its operational presentation as an LTS with non-deterministic memory-internal transitions
that flush stores from the volatile part %
to the non-volatile part.
In \cref{sec:systems}, we define the synchronization of programs with the \PSC memory system.
In \cref{sec:psc_ext}, we present the extensions added in this paper
that are useful for library abstraction.
Finally, in \cref{sec:sep}, we establish certain separation properties of \PSC
that are essential in our
proofs.

Roughly speaking, a state in \PSC consists of a non-volatile memory
(mapping from non-volatile variables to values)
and a volatile memory (mapping from volatile variables to values).
The volatile memory works just as a normal sequentially consistent memory, keeping track of
the latest written value to every variable and returning that value for reads.
Upon crash, the contents of the volatile memory is reset to its initial state.
The non-volatile memory behaves observationally the same between crashes,
but its contents survive crashes.
To model delayed and out-of-order persistence of writes, write steps to non-volatile variables
do not alter the non-volatile memory immediately when issued.
Instead, writes first go to volatile per-variable persistence FIFO buffers, which
maintain the writes to each variable that are yet to persist.
Then, \PSC non-deterministically takes \emph{persist steps}
that apply the oldest update from a persistence buffer in the non-volatile memory.
Reads from non-volatile variables retrieve the latest value in the relevant buffer,
or the value from the non-volatile memory if that buffer is empty,
thus providing standard sequentially consistent semantics in the absence of system crashes.
Upon crash the buffers are reset to their initial (empty) state,
but the contents of the non-volatile memory remains intact.

Explicit persist instructions can be used to control the persistence of writes.
A ``flush'' barrier for a certain variable blocks the execution until the relevant persistence buffer is empty,
thus forcing all previous writes to that variable to persist.
Alternatively, a (cheaper) ``flush-optimal'' barrier for a certain variable enqueues a special marker in the persistence buffer of this variable
accompanied by the thread identifier of the thread that issued the barrier.
The effect of flush-optimal is delayed until the same thread performs an sfence,
which blocks the execution until all flush-optimal markers of that thread are dequeued from all buffers.
The fact that the persistence buffers are FIFO ensures that an sfence by some thread
forces the persistence of all writes executed before a flush-optimal issued by the same thread.

\begin{definition}
{\rm
\PSC is the LTS defined as follows:
\begin{myitemize}
\item The transition labels are given by $(\Tid \times \Lab) \cup
\set{\persistlab,\crash}$. That is, a transition label can be a pair of the
thread identifier and the action label of the operation, $\persistlab$ denoting
the internal propagation action, or $\crash$ denoting a system crash.
\item The states are tuples $\memstate=\tup{\mem,\vmem, \Pbuff}$, where:
\begin{itemize}[leftmargin=*]
\item $\mem:\NVLoc \to \Val$ is called the \emph{non-volatile memory}.
\item $\vmem:\VLoc \to \Val$ is called the \emph{volatile memory}.
\item $\Pbuff : \NVLoc \to \mathsf{PLBuff}$ is called the \emph{persistence buffer}.
Here, $\mathsf{PLBuff}$ denotes the set of all \emph{per-location persistence buffers},
each of which is a finite sequence $\pbuff$ of entries
of the form $\pbwlab{\val}$ for $\val\in\Val$ (writes), or $\fotlabp{\tid}$ for $\tid\in\Tid$ (flush optimal markers).
The persistence buffer $\Pbuff$ assigns a per-location persistence buffer
to every non-volatile variable.\footnote{We conservatively assume that writes persist at the location granularity,
rather than at the cache-line granularity as happens in real machines.}
\end{itemize}
We denote by $\memstate.\lmem$, $\memstate.\lvmem$, and $\memstate.\lPbuff$
the components of a state $\memstate \in \PSC.\lQ$,
and write $\memstate[\lX\mapsto Y]$ for the
state obtained from $\memstate$ by setting $\memstate.\lX$ to $Y$.

\item The initial state is $\memstate_\Init \defeq \tup{\mem_\Init, \vmem_\Init, \Pbuff_\Init}$,
where $\mem_\Init \defeq \lambda \nvloc \ldotp 0$,
$\vmem_\Init \defeq \lambda \vloc \ldotp 0$,
and $\Pbuff_\Init \defeq \lambda \nvloc \ldotp \epsl$.

\item The transitions of \PSC are presented in \cref{fig:PSC},
using an auxiliary function for looking up the most recent value of a variable:
we let $\memstate(\loc)$ be $\memstate.\lvmem(\loc)$ for $\loc \in \VLoc$, and, for $\loc \in \NVLoc$, either the value $\val$ of the last write (rightmost) entry $\memstate.\lPbuff(\loc)$ or, when there is no such entry, $\memstate.\lmem(\loc)$.
\end{myitemize}
}\end{definition}

\begin{figure}[t]
\scriptsize
\begin{mathpar}
\inferrule[v-write]{
	\lab = \wlab{\vloc}{\val}
	\\\\
	\vmem'=\memstate.\lvmem[ \vloc \mapsto \val]
}{
	\memstate
	\asteptidlab{\tid}{\lab}{\PSC}
	\memstate[\lvmem \mapsto \vmem']
}
\and
\inferrule[nv-write]{
	\lab = \wlab{\nvloc}{\val}
	\\\\
\pbuff'=\memstate.\lPbuff(\nvloc) \cdot \pbwlab{\val}
\\
	\Pbuff'=\memstate.\lPbuff[ \nvloc \mapsto \pbuff']
}{
	\memstate
	\asteptidlab{\tid}{\lab}{\PSC}
	\memstate[\lPbuff\mapsto \Pbuff']
}
\and
\inferrule[read]{
	\lab = \rlab{\loc}{\val}
	\\\\
	\memstate(\loc) = \val
}{
	\memstate
	\asteptidlab{\tid}{\lab}{\PSC}
	\memstate
}
\and
\hfill
\inferrule[flush]{
	\lab = \fllab{\nvloc}
	\\\\ {\memstate.\lPbuff(\nvloc)=\epsl}
}{
	\memstate
	\asteptidlab{\tid}{\lab}{\PSC}
	\memstate
} \and
\inferrule[flush-opt]{
	\lab = \folab{\nvloc}
\\\\
\pbuff'=\memstate.\lPbuff(\nvloc) \cdot \fotlabp{\tid}
\\
	\Pbuff'=\memstate.\lPbuff[ \nvloc \mapsto \pbuff']
}{
	\memstate
	\asteptidlab{\tid}{\lab}{\PSC}
	\memstate[\lPbuff\mapsto \Pbuff']
} \and
\inferrule[sfence]{
	\lab = \sflab
	\\\\
	\forall \nvloc \ldotp \fotlabp{\tid} \nin \memstate.\lPbuff(\nvloc)
}{
	\memstate
	\asteptidlab{\tid}{\lab}{\PSC}
	\memstate
}
\end{mathpar}
\myhrule
\begin{mathpar}
\inferrule[persist-write]{
	\lab = \persistlab %
    \\
 	\memstate.\lPbuff(\nvloc) = \pbwlab{\val} \cdot \pbuff
	\\\\
	\Pbuff' = \memstate.\lPbuff[\nvloc \mapsto \pbuff]
	\\
	\mem' = \memstate.\lmem[\nvloc \mapsto \val]
}{
	\memstate
	\asteplab{\lab}{\PSC}
	\memstate[\lmem \mapsto \mem', \lPbuff \mapsto \Pbuff']
}
\hfill
\inferrule[persist-fo]{
	\lab = \persistlab %
    \\
 	\memstate.\lPbuff(\nvloc) = \fotlabp{\tid} \cdot \pbuff
	\\\\
	\Pbuff' = \memstate.\lPbuff[\nvloc \mapsto \pbuff]
}{
	\memstate
	\asteplab{\lab}{\PSC}
	\memstate[\lPbuff \mapsto \Pbuff']
}
\hfill
\inferrule[crash]{
	\lab = \crash
}{\memstate \asteplab{\lab}{\PSC}
\memstate_\Init[\lmem \mapsto \memstate.\lmem]}
\end{mathpar}
\caption{Transitions of $\PSC$}
\label{fig:PSC}
\end{figure}

The transitions follow the intuitive account above.
Those corresponding to program transitions
are labeled with pairs in $\Tid \times \Lab$.
For instance, a transition labeled with $\tidlab{\tid}{\rlab{\loc}{\val_\lR}}$
means that thread $\tid$ reads the value $\val_\lR$
from (volatile or non-volatile) shared variable $\loc$.

\subsection{Linking Programs and Memories}
\label{sec:systems}

To give semantics of programs running under \PSC,
the thread system is synchronized with the \PSC memory system.
Formally, the synchronization of a program $\prog$ with \PSC,
is another LTS, denoted by $\cs{\prog}{\PSC}$, defined as follows:

\begin{myitemize}
\item The set of transition labels is $\prog.\lSigma \cup \PSC.\lSigma$,
\ie $(\Tid \times (\Lab \cup \set{\epsl})) \cup \set{\persistlab,\crash}$.

\item The states are pairs $\tup{\progstate,\memstate}\in \prog.\lQ \times \PSC.\lQ$.

\item The initial state is  $\tup{\progstate_\Init,\memstate_\Init}$.

\item The transitions are given by:
\end{myitemize}

\vspace*{-10pt}
\begin{mathpar}
\footnotesize
\inferrule[synchronized]{
\alpha\in  (\Tid \times \Lab) \cup \set{\crash} \\\\
	\progstate {\asteplab{\alpha}{\prog}} \progstate'
	\\ \memstate
			{\asteplab{\alpha}{\PSC}}
			\memstate'
}{
	\tup{\progstate,\memstate}
		{\asteplab{\alpha}{\cs{\prog}{\PSC}}}
		\tup{\progstate',\memstate'}
}
\hfill
\inferrule[program-internal]{
\alpha\in  \Tid \times \set{\epsl}  \\\\
	\progstate {\asteplab{\alpha}{\prog}} \progstate'
}{
	\tup{\progstate,\memstate}
		{\asteplab{\alpha}{\cs{\prog}{\PSC}}}
		\tup{\progstate',\memstate}
}
\hfill
\inferrule[memory-internal]{
\alpha = \persistlab\\\\
\memstate
			{\asteplab{\alpha}{\PSC}}
			\memstate'
}{
	\tup{\progstate,\memstate}
		{\asteplab{\alpha}{\cs{\prog}{\PSC}}}
		\tup{\progstate,\memstate'}
}
\end{mathpar}
The above transitions are ``synchronized transitions'' of $\prog$ and $\PSC$,
using the labels to decide what to synchronize on.
Both the program and the memory take the same step
for transition labels that are common to both LTSs, %
only the program steps for transition labels that are only program transitions, %
and only the memory steps for transition labels that are only memory transitions. %

\subsection{Extending \PSC for Library Abstraction}
\label{sec:psc_ext}

We present the modifications of \PSC for
supporting the new specification constructs:
localized sfences and persistence blocks.
When referring to \PSC in the sequel we mean the following revised version.

\cparagraph{Local store fences}
Localized sfences are straightforwardly supported by the following additional memory transition:
\[
\footnotesize
\inferrule*[left=local sfence]{
	\lab = \pflab{\nvlocset}
	\\
	{\forall \diffemph{\nvloc \in \nvlocset} \ldotp \fotlabp{\tid} \nin \memstate.\lPbuff(\nvloc)}
}{
	\memstate
	\asteptidlab{\tid}{\lab}{\PSC}
	\memstate
}
\]
Here, instead of blocking until all $\fotlabp{\tid}$ entries are removed from all buffers,
we only require that such entries are not present in buffers associated with variables
from a certain set (mentioned in the action label and corresponding to
the argument of the $\pfenceInst{\nvlocset}$ instruction).

\cparagraph{Persistence blocks}
We assume an infinite set $\mathsf{BlockID}$ of block identifiers that
are non-deterministically allocated when blocks are opened.
The state of the memory system keeps track of a mapping
assigning the current open block identifier to every thread and non-volatile variable,
or $\bot$ if the variable is not a part of an open block of the thread.
When writing to non-volatile variables, the associated block identifiers
are attached to the write entry in the per-location persistence buffer.
In turn, the propagation from the buffers to the NVM ensures that blocks
are propagated only after they are not open and only in their entirety.
To do so, we generalize the persist
step of \PSC to allow simultaneous propagation of multiple entries from the buffers.
To respect the per-variable FIFO order, the propagated entries should form a prefix of each buffer.

Formally, this requires the following modifications:
 \begin{myenum}
  \item  Write entries in buffers take the form
$\ipbwlab{\pbid}{\val}$ where $\pbid \in \mathsf{BlockID} \cup \set{\bot}$ and $\val\in\Val$
(instead of $\pbwlab{\val}$). A write entry of the form $\ipbwlab{\bot}{\val}$ means that the corresponding write
was not a part of a persistence block.
\item States are extended to be quintuples $\memstate=\tup{\mem,\vmem, \Pbuff, \pblock, \pbidSet}$, where:
\begin{itemize}[leftmargin=*]
\item $\pblock : \Tid \to \NVLoc \to (\mathsf{BlockID} \cup \set{\bot})$ is called the \emph{active-block mapping}.
It assigns a block identifier (or $\bot$ if there is no active block) to every thread identifier and non-volatile variable.
\item $\pbidSet\suq \mathsf{BlockID} \times \powerset{\NVLoc}$ is called the \emph{block identifiers set}.
It is used to store all persistence block identifiers occurring so far,
each accompanied by the set of non-volatile variables that it protects.
\end{itemize}
We denote by $\memstate.\lpblock$ and
$\memstate.\lpbidSet$ the additional components of a state $\memstate$.
We impose the following well-formedness conditions:
\begin{itemize}[leftmargin=*]
\item If $\ipbwlab{\pbid}{\_} \in \memstate.\lPbuff(\nvloc)$,
then $\tup{\pbid,\set{\nvloc} \cup \nvlocset} \in \memstate.\lpbidSet$
for some $\nvlocset\suq \NVLoc$.
\item If $\memstate.\lpblock(\tid)(\nvloc) \neq \bot$,
then $\tup{\memstate.\lpblock(\tid)(\nvloc) ,\set{\nvloc} \cup \nvlocset} \in \memstate.\lpbidSet$
for some $\nvlocset\suq \NVLoc$.
\end{itemize}

\item The initial state is given by
$\memstate_\Init \defeq \tup{\mem_\Init, \vmem_\Init, \Pbuff_\Init, \pblock_\Init, {\pbidSet}_\Init}$,
where $\pblock_\Init \defeq \lambda \tid \ldotp \lambda \nvloc \ldotp \bot$,
and ${\pbidSet}_\Init  \defeq \emptyset$.
\item The \rulename{nv-write} transition records the current active block in the added entry:
\begin{mathpar}
\footnotesize
\inferrule*[left=nv-write]{
	\lab = \wlab{\nvloc}{\val}
	\\
\pbuff'=\memstate.\lPbuff(\nvloc) \cdot \ipbwlab{\diffemph{\memstate.\lpblock(\tid)(\nvloc)}}{\val}
\\
	\Pbuff'=\memstate.\lPbuff[ \nvloc \mapsto \pbuff']
}{
	\memstate
	\asteptidlab{\tid}{\lab}{\PSC}
	\memstate[\lPbuff\mapsto \Pbuff']
}
\end{mathpar}
\item The following two transitions for opening and closing blocks are added:
\begin{mathpar}
\footnotesize
\inferrule[beginPB]{
	\lab = \beginpblab{\nvlocset}
   \\\\ \forall \nvloc \in \nvlocset \ldotp \memstate.\lpblock(\tid)(\nvloc) =\bot
	\\\\ \pblock' = \memstate.\lpblock\left[
	\tid \mapsto \lambda \nvloc\ldotp \inarr{\text{if } \nvloc\in\nvlocset \text{ then } \pbid \\ \text{else } \memstate.\lpblock(\tid)(\nvloc)}\right]
	\\\\ \tup{\pbid,\_} \nin \memstate.\lpbidSet  \\  \pbidSet' = \memstate.\lpbidSet \cup \set{\tup{\pbid,\nvlocset}}
}{
	\memstate
	\asteptidlab{\tid}{\lab}{\PSC}
	\memstate[\lpblock \mapsto \pblock', \lpbidSet \mapsto \pbidSet']
}
\hfill
\inferrule[endPB]{
	\lab = \myendpblab{\nvlocset}
	\\\\
	\\\\ \pblock' = \memstate.\lpblock\left[
	\tid \mapsto \lambda \nvloc\ldotp \inarr{\text{if } \nvloc\in\nvlocset \text{ then } \bot \\ \text{else } \memstate.\lpblock(\tid)(\nvloc)}\right]
\\\\
}{
	\memstate
	\asteptidlab{\tid}{\lab}{\PSC}
	\memstate[\lpblock \mapsto \pblock']
}
\end{mathpar}
Thus, opening a block allocates a fresh identifier and sets the active-block mapping accordingly.
In turn, closing a block resets the relevant variables in the active-block mapping.

\item The following transition is used \emph{instead} of \rulename{persist-write} and \rulename{persist-fo}.
It generalizes both \rulename{persist-write} and \rulename{persist-fo} by simultaneously
persisting several entries together (each $\pbuff_\nvloc$ below stands for a \emph{sequence} of entries).
\begin{mathpar}
\footnotesize
\inferrule*[left=persist]{
	\lab = \persistlab %
    \\
 	\forall \nvloc \ldotp \memstate.\lPbuff(\nvloc) = \pbuff_\nvloc \cdot \Pbuff'(\nvloc)
	\\\\
	\forall \pbid  \ldotp
		(\exists \nvloc \ldotp \ipbwlab{\pbid}{\_} \in \pbuff_\nvloc)
		\implies
		\forall \nvloc \ldotp (\forall\tid \ldotp \memstate.\lpblock(\tid)(\nvloc)\neq \pbid \land
		 \ipbwlab{\pbid}{\_} \nin \Pbuff'(\nvloc))
\\\\	\mem' =  \lambda \nvloc \ldotp
{	\begin{cases}
		\val &	 \text{last write entry in  $\pbuff_\nvloc$ has value $\val$} \\
		\memstate.\lmem(\nvloc) & 				 \text{there are no write entries in $\pbuff_\nvloc$}
		\end{cases}}
}{
	\memstate
	\asteplab{\lab}{\PSC}
	\memstate[\lmem \mapsto \mem', \lPbuff \mapsto \Pbuff']
}
\end{mathpar}
This step imposes two restrictions.
First, the persisted entries from each buffer ($\pbuff_\nvloc$) should form a prefix of that buffer, so that FIFO semantics
is maintained.
Second, to respect the persistence blocks, if some entry of a given block is persisted ($\exists \nvloc \ldotp \ipbwlab{\pbid}{\_} \in \pbuff_\nvloc$)
then that block should not be currently active by any thread ($\forall \nvloc,\tid \ldotp \memstate.\lpblock(\tid)(\nvloc)\neq \pbid$)
and no entries of that block should remain in the volatile buffers
($\forall \nvloc \ldotp \ipbwlab{\pbid}{\_} \nin \Pbuff'(\nvloc))$).
 \end{myenum}

\smallskip
\noindent
\begin{minipage}{.75\textwidth}
We note that nested and interleaved blocks are allowed.
The program on the right demonstrates such a case.
Here, $\cnvloc{1}=1$ and $\cnvloc{2}=1$ must persist together;
$\cnvloc{3}=1$ and $\cnvloc{4}=1$ must persist together;
but these two pairs can persist independently of each other in any order.
Thus, provided that the client and the library use blocks of their own locations,
the block instructions by each component are invisible to the other.
\end{minipage}\hfill
\begin{minipage}{.2\textwidth}
\vspace*{-10pt}
\small
$$\inarr{
\beginpbInst{\cnvloc{1},\cnvloc{2}} \sep  \\
\writeInst{\cnvloc{1}}{1} \sep \\
\beginpbInst{\cnvloc{3},\cnvloc{4}} \sep  \\
\writeInst{\cnvloc{3}}{1} \sep \writeInst{\cnvloc{4}}{1} \sep \\
\myendpbInst{\cnvloc{3},\cnvloc{4}} \sep  \\
\writeInst{\cnvloc{2}}{1} \sep \\
\myendpbInst{\cnvloc{1},\cnvloc{2}} \sep
}$$
\end{minipage}

\subsection{Separation Properties}
\label{sec:sep}

To enable our library abstraction proof,
the required key property of \PSC, which we preserved in its extensions,
is the ability to separate \PSC states into disjoint parts (the library's part and the client's part)
and capture each memory transition in terms of its effect on the two parts.
Next, we formulate this property, which we will later use to prove library abstraction.
In fact, our arguments for library abstraction rely only on the properties below,
and never ``unfold'' the \PSC-related definitions.
This allows one to refine and extend \PSC,
as long as the separation properties are preserved.

The separation of \PSC states
is stated in terms of the following restriction operator
relative to a set of variables.
For persistence blocks to behave correctly, we need an auxiliary condition
on this set: we say that a set $\nvlocset \suq \NVLoc$ \emph{separates a state}
$\memstate\in \PSC.\lQ$
if for every $\tup{\pbid,\nvlocseta}\in \memstate.\lpbidSet$,
we have $\nvlocseta \suq \nvlocset$
or $\nvlocseta \suq \NVLoc\setminus \nvlocset$.

\begin{definition}
\label{def:restriction_psc}
{\rm
The \emph{restriction} of
$\memstate\in \PSC.\lQ$
onto a set $\locset\suq\Loc$ such that $\locset\cap\NVLoc$ separates $\memstate$,
denoted by $\memstate\rst{\locset}$,
is the state $\memstate'\in \PSC.\lQ$ given by:
\begin{myitemize}
	\item $\memstate'.\lmem(\nvloc)$ is $\memstate.\lmem(\nvloc)$ if $\nvloc\in\NVLoc\cap\locset$,
	or $0$ otherwise.
	\item $\memstate'.\lvmem(\vloc)$ is $\memstate.\lvmem(\vloc)$ if $\vloc\in\VLoc\cap\locset$,
	or $0$ otherwise.
	\item $\memstate'.\lPbuff(\nvloc)$ is $\memstate.\lPbuff(\nvloc)$ if $\nvloc\in\NVLoc\cap\locset$,
	or $\epsilon$ otherwise.
	\item For each $\tid\in\Tid$, $\memstate'.\lpblock(\tid)(\nvloc)$
	is $\memstate.\lpblock(\tid)(\nvloc)$ if $\nvloc\in\NVLoc\cap\locset$,
	or $\bot$ otherwise.
	\item $\memstate'.\lpbidSet= \set{\tup{\pbid,\nvlocseta} \in \memstate.\lpbidSet \mid \nvlocseta \suq \locset }$.
\end{myitemize}
}\end{definition}

The next lemma states the separation property of \PSC, providing a precise
characterization of each \PSC transition in terms of transitions on the
restrictions $\memstate\rst{\locset}$ and $\memstate\rst{\Loc \setminus
\locset}$. A special case is needed for store fence transitions, since taking these transitions enforces conditions on \emph{both}
restrictions.

\begin{lemma}
\label{lem:memory_disjointness}
{\rm
Let $\locset\suq\Loc$ such that $\locset\cap\NVLoc$ separates a
state $\memstate_1$.
\begin{myenum}
\item \label{item:other}
 For every $\tid\in\Tid$ and $\lab\in\Lab\setminus \set{ \sflab}$
with $\lLOCSET(\lab)\suq \locset$,\\
$\memstate_1 \asteptidlab{\tid}{\lab}{\PSC} \memstate_2
\Longleftrightarrow
(\memstate_1\rst{\locset} \asteptidlab{\tid}{\lab}{\PSC} \memstate_2\rst{\locset} \land
\memstate_1\rst{\Loc\setminus\locset}=\memstate_2\rst{\Loc\setminus\locset})
$
\item \label{item:sflab}
 For every $\tid\in\Tid$,\\
$
\memstate_1 \asteptidlab{\tid}{\sflab}{\PSC} \memstate_2
\Longleftrightarrow
(\memstate_1\rst{\locset} \asteptidlab{\tid}{\sflab}{\PSC} \memstate_2\rst{\locset} \land
\memstate_1\rst{\Loc\setminus\locset} \asteptidlab{\tid}{\sflab}{\PSC} \memstate_2\rst{\Loc\setminus\locset})
$
\item \label{item:persist}
$\memstate_1 \asteplab{\persistlab}{\PSC} \memstate_2
\Longleftrightarrow
(\memstate_1\rst{\locset} \asteplab{\persistlab}{\PSC} \memstate_2\rst{\locset} \land
\memstate_1\rst{\Loc\setminus\locset} \asteplab{\persistlab}{\PSC} \memstate_2\rst{\Loc\setminus\locset})
$
\item \label{item:crash}
$\memstate_1 \asteplab{\crash}{\PSC} \memstate_2
\Longleftrightarrow
(\memstate_1\rst{\locset} \asteplab{\crash}{\PSC} \memstate_2\rst{\locset} \land
\memstate_1\rst{\Loc\setminus\locset} \asteplab{\crash}{\PSC} \memstate_2\rst{\Loc\setminus\locset})
$
\end{myenum}
}
\end{lemma}

The proof of \cref{lem:memory_disjointness} proceeds by
standard case analysis ranging over all possible transitions of \PSC.
Finally, the following operation is used below to compose a state from a client and a library components
(see \cref{lem:composition}).

\begin{definition}
\label{def:compose_psc}
{\rm Let $\memstate_1, \memstate_2$ be
states of $\PSC$,
and $\locset_1,\locset_2\suq\Loc$ such that $\locset_1 \cap \locset_2 = \emptyset$.
The \emph{merge of $\memstate_1$ and $\memstate_2$ \wrt
$\locset_1$ and $\locset_2$}, denoted by
$\mergemem{\memstate_1}{\locset_1}{\memstate_2}{\locset_2}$,
is the state $\memstate\in \PSC.\lQ$ defined by:
\begin{mathpar}
\scriptsize
\memstate.\lmem(\nvloc) = {\begin{cases}
\memstate_1.\lmem(\nvloc) & \nvloc\in\locset_1 \\
\memstate_2.\lmem(\nvloc) & \nvloc\in\locset_2 \\
0 & \text{otherwise}
\end{cases}}
\hfill
\inarr{\text{similar definitions}\\ \text{for } \memstate.\lvmem,\memstate.\lPbuff,\memstate.\lpblock}
\hfill
\memstate.\lpbidSet = \inarr{
\set{\tup{\pbid,\nvlocseta} \in \memstate_1.\lpbidSet \mid \nvlocseta \suq \locset_1 } \cup \\
\set{\tup{\pbid,\nvlocseta} \in \memstate_2.\lpbidSet \mid \nvlocseta \suq \locset_2 }}
\end{mathpar}
}\end{definition}

\section{Libraries and Their Clients}
\label{sec:lib}

We present the notions of libraries and clients,
as well as the necessary definitions for stating the abstraction theorem:
histories and most general clients.

\cparagraph{Libraries}
We take a library $\lib$ to be a function assigning to method names in
$\dom{\lib}\suq\MethodNames$ flat instruction sequences representing the method
bodies.
In the context of some library $\lib$,
we refer to the implementations of the methods in $\set{\main}\cup\MethodNames\setminus\dom{\lib}$ in a program $\prog$
as the \emph{client of $\lib$}.

\cparagraph{Client-library composition}
We consider the common case where libraries and their clients
never access the same shared variables.
To formally define this restriction, we use the following notations for sets of locations
used by instruction sequences, libraries, and their clients:
\begin{myitemize}
\item $\usedlocs{\iseq}$ denotes the set of shared variables mentioned
in an instruction sequence $\iseq$  (possibly as a part of a set $\nvlocset$ of variables, \eg in $\beginpbInst{\nvlocset}$).
\item For a library $\lib$, $\usedlocs{\lib} \defeq \bigcup_{\method\in \dom{\lib}} \usedlocs{\lib(\method)}$.
\item For a program $\prog$ and a set $\methodset \suq \MethodNames$,
\\ $\usedclientlocs{\methodset}{\prog} \defeq
\bigcup_{\tid \in \Tid} \usedlocs{\prog(\tid)(\main)}\cup
\bigcup_{\method \in \MethodNames \setminus \methodset} \usedlocs{\prog(\method)}$.
\end{myitemize}

Then, client-library composition is defined as follows.

\begin{definition}
\label{def:safe}
{\rm
A library $\lib$ is \emph{safe} for a program $\prog$ if $\usedlocs{\lib} \cap
\usedclientlocs{\dom{\lib}}{\prog} = \emptyset$. When $\lib$ is safe for
$\prog$, we write $\client{\prog}{\lib}$ for the program obtained from $\prog$
by setting $\prog(\tid)(\method)=\lib(\method)$ for every $\tid\in\Tid$ and
$\method\in\dom{\lib}$.
}\end{definition}
Note that we always have
$\usedclientlocs{\dom{\lib}}{\client{\prog}{\lib}}=\usedclientlocs{\dom{\lib}}{\prog}$.

\cparagraph{Histories}
Histories record the interactions between libraries and clients. Formally, a
\emph{history} $\history$ of a library $\lib$ is a sequence of transition labels
representing a crash, a call to a method of $\lib$, a return from a method of
$\lib$, or an sfence, \ie labels from the set $\HTLab_{\dom{\lib}}$, which is
defined as follows:
 \begin{align*}
\Lab_\methodset & \defeq \set{\sflab} \cup
\set{\calllab{\method}{\phi},\retlab{\method}{\phi} \mid \method\in\methodset, \phi:\Reg\to\Val}
\\
\HTLab_\methodset & \defeq (\Tid\times \Lab_\methodset) \cup \set{\crash}
\end{align*}

\begin{definition}
\label{def:history_of_trace}
{\rm
Let $\tr$ be a trace of $\cs{\prog}{\PSC}$ for some program $\prog$.
The \emph{history} induced by $\tr$ \wrt a set $\methodset\suq\MethodNames$,
denoted by $\rsthistory{\methodset}{\tr}$,
is the subsequence of $\tr$ over $\HTLab_\methodset$
consisting of (in the same order they appear in $\tr$):
call and return labels $\tidlab{\tid}{\calllab{\method}{\phi}}$ and
$\tidlab{\tid}{\retlab{\method}{\phi}}$ with $\method\in \methodset$;
$\lSF$-labels $\tidlab{\tid}{\sflab}$; and
crash labels.
The notation $\rsthistory{\methodset}{\tr}$ is extended to sets of traces in the obvious way.
The set of histories \wrt $\methodset$ of $\prog$,
denoted by $\rsthistory{\methodset}{\prog}$, is given by $\rsthistory{\methodset}{\traces{\cs{\prog}{\PSC}}}$.
When $\methodset=\MethodNames$ (\ie the set of all method names),
we simply write $\fhistory{\tr}$ and $\fhistory{\prog}$.
}\end{definition}

\cparagraph{Most general clients}
We encompass library calling policies (see \cref{sec:policy}) using the notion
of a ``most general client''---a non-deterministic client that invokes the
library methods in the most general way allowed by the policy. Formally, a most
general client $\MGCn$ is given as a (concurrent) program. Adherence to the
calling policy is defined as follows.

\begin{definition}
\label{def:policy}
{\rm
Let $\lib$ be a library, and $\prog$ and $\MGCn$ be programs
such that $\lib$ is safe for both $\prog$ and $\MGCn$.
We say that $\prog$ \emph{correctly calls} $\lib$ \wrt $\MGCn$
if $\rsthistory{\dom{\lib}}{\client{\prog}{\lib}} \suq
\rsthistory{\dom{\lib}}{\MGC{\lib}}$.
}\end{definition}

The policy of a library with no restrictions on its clients (beyond the separation of shared resources)
is expressed by an MGC, called $\MGCfree$,
 that repeatedly invokes arbitrary library methods with arbitrary initial stores.
Often persistent objects include a recovery method meant to
be executed after a crash before any other method is invoked.
We call such a policy $\MGCrec$.
Formally, $\MGCfree$ (for $\dom{\lib}= \set{\method_1 \til \method_n}$)
and $\MGCrec$ (for $\dom{\lib}= \set{\method_1 \til \method_n} \uplus \set{\recoverMethod}$)
assign the following main method to each thread $\tid$:
\[
\footnotesize
\inarr{\resizebox{\mycodefactor\width}{!}{
\begin{array}[t]{l}
\MGCfree(\tid)(\main) = \\
\kw{BEGIN}: \havocInst \sep \\
\GotoInst{\kw{f}_1 \shortmid\ldots\shortmid \kw{f}_n\shortmid\kw{END}} \sep   \\
\kw{f}_1: \callInst{\method_1} \sep \GotoInst{\kw{BEGIN}} \sep \\
\ldots  \\
\kw{f}_n: \callInst{\method_n} \sep \GotoInst{\kw{BEGIN}} \sep \\
\kw{END}:  \\
\end{array}
\qquad
\begin{array}[t]{l}
\MGCrec(\tid)(\main) = \\
\casInst{\creg{1}}{\cvloc{1}}{0}{1} \sep \ifGotoInst{\creg{1}=0}{\kw{REC}} \sep \GotoInst{\kw{WAIT}} \sep  \\
\kw{REC}: \callInst{\recoverMethod} \sep \writeInst{\cvloc{2}}{1} \sep \GotoInst{\kw{BEGIN}} \sep  \\
\kw{WAIT}: \readInst{\creg{1}}{\cvloc{2}} \sep \ifGotoInst{\creg{1}=0}{\kw{WAIT}} \sep \GotoInst{\kw{BEGIN}} \sep \\
\kw{BEGIN}: \ldots \text{rest of the code as in $\MGCfree$} \ldots
\end{array}}}
\]

In $\MGCrec$, using a compare-and-swap, one thread performs the recovery.
All other threads wait until recovery ends to start their %
method invocations.

\newcommand{\scLEM}{\textsc{Compose}}
\newcommand{\lcl}{{\sf cl}}
\newcommand{\llib}{{\sf lib}}

\newcommand{\trc}{\tr_\lcl}
\newcommand{\trl}{\tr_\llib}

\section{The Library Abstraction Theorem}
\label{sec:abs}

In this section we state and prove the library abstraction theorem.
The premise of this theorem, the \emph{library correctness condition}, is formulated as follows.

\begin{definition}
\label{def:refines}
{\rm
Let $\lib$ and $\lib^\sharp$ be libraries,
both safe for a program $\MGCn$.
We say that $\lib$ \emph{refines} $\lib^\sharp$ \wrt $\MGCn$,
denoted by $\lib \sqsubseteq_\MGCn \lib^\sharp$,
if both libraries implement the same methods
and
$\fhistory{\MGC{\lib}} \suq \fhistory{\MGC{\lib^\sharp}}$.
}\end{definition}

Next, the abstraction theorem states that $\lib \sqsubseteq_\MGCn \lib^\sharp$
ensures that any client adhering to the library's calling policy
may safely use the implementation
$\lib$ while reasoning about possible behaviors in terms of
the specification $\lib^\sharp$.
Our notion of ``a behavior'' includes the generated histories,
as well as the reachable states, by the composition of the program
and the memory system.
Including reachable states is intended to assist safety verification.
Clearly, we cannot require that the program states match
for threads that are currently executing a method of $\lib$.
In addition, since $\lib$ and $\lib^\sharp$ may update the memory differently
(\eg use different variables), we should only consider the variables of the client
when inspecting the memory states.
This leads us to the following statement.

\begin{restatable}[Abstraction]{theorem}{abstraction}
\label{thm:abs}
Suppose that $\lib \sqsubseteq_\MGCn \lib^\sharp$.
Let $\MGCn$ and $\prog$ be programs
such that both $\lib$ and $\lib^\sharp$ are safe for $\MGCn$ and $\prog$,
and $\prog$ correctly calls $\lib^\sharp$ \wrt $\MGCn$.
If $\tup{\progstate_\Init,\memstate_\Init}
{\asteplab{\tr}{\cs{\client{\prog}{\lib}}{\PSC}}}
\tup{\progstate,\memstate}$, then
there exist $\tr^\sharp$ and
$\tup{\progstate^\sharp,\memstate^\sharp}$  such that the following hold:
\begin{myitemize}%
\item $\tup{\progstate_\Init,\memstate_\Init}
{\asteplab{\tr^\sharp}{\cs{\client{\prog}{\lib^\sharp}}{\PSC}}}
\tup{\progstate^\sharp,\memstate^\sharp}$.
\item $\fhistory{\tr^\sharp} = \fhistory{\tr}$.
\item For every $\tid\in\Tid$, if $\progstate(\tid).\lmethod \nin \dom{\lib}$, then
$\progstate^\sharp(\tid)=\progstate(\tid)$.
\item $\memstate^\sharp\rst{\usedclientlocs{\dom{\lib}}{\prog}}=
\memstate\rst{\usedclientlocs{\dom{\lib}}{\prog}}$ (see \cref{def:restriction_psc}).
\end{myitemize}
\end{restatable}

Note that $\lib \sqsubseteq_\MGCn \lib^\sharp$ is necessary for the conclusion
to hold: otherwise, $\MGCn$ itself is a client that can observe behaviors of
$\lib$ that are impossible for $\lib^\sharp$. Following
\cref{sec:policy}, we also note that policy adherence is required \wrt to
$\lib^\sharp$.

To prove the abstraction theorem, the following key lemma is used multiple times
(with different arguments). It allows us to compose the client's part from one trace
with the library's part from another into one combined trace.

\begin{restatable}[Composition]{lemma}{composition}
\label{lem:composition}
{\rm
Let $\lib$ and $\lib'$ be libraries implementing the same set $\methodset$ of methods
such that both are safe for a program $\prog$,
and $\lib$ is also safe for a program $\prog'$.
Suppose that
$\tup{\progstate_\Init,\memstate_\Init}
{\asteplab{\tr_\lcl}{\cs{\client{\prog}{\lib'}}{\PSC}}}
\tup{\progstate_\lcl,\memstate_\lcl}$,
$\tup{\progstate_\Init,\memstate_\Init}
{\asteplab{\tr_\llib}{\cs{\client{\prog'}{\lib}}{\PSC}}}
\tup{\progstate_\llib,\memstate_\llib}$,
and
$\rsthistory{\methodset}{\tr_\lcl} = \rsthistory{\methodset}{\tr_\llib}$.
Then, there exists a trace $\tr$ such that
$\fhistory{\tr} = \fhistory{\tr_\lcl}$
and $\tup{\progstate_\Init,\memstate_\Init}
\asteplab{\tr}{\cs{\client{\prog}{\lib}}{\PSC}} \tup{\progstate,\memstate}$,
for:
\begin{myitemize}
\item $\progstate =
\lambda \tid \ldotp
\begin{cases}
\tup{
\progstate_\llib(\tid).\lpc,
\progstate_\llib(\tid).\lphi,
\progstate_\lcl(\tid).\lpcs,
\progstate_\lcl(\tid).\lmethod}
& \progstate_\lcl(\tid).\lmethod \in \methodset
\\
\progstate_\lcl(\tid)
& \mbox{otherwise}
\end{cases}$
\item
$\memstate=\nmergemem{\memstate_\lcl}{\usedclientlocs{\methodset}{\prog}}
{\memstate_\llib}{\usedlocs{\lib}}$ (see \cref{def:compose_psc}).
\end{myitemize}
}
\end{restatable}

The
proof of  \cref{lem:composition}  is based on the inherent disjointness in client-library composition
provided by a library safe for its client program, which we leverage in the
following two ways.

Firstly, we extract \emph{client-local} and \emph{library-local} transition
properties from all transitions of $\cs{\client{\prog}{\lib'}}{\PSC}$ and
$\cs{\client{\prog'}{\lib}}{\PSC}$. Thus, when we consider a transition by
$\cs{\client{\prog}{\lib'}}{\PSC}$ corresponding to an instruction outside of
a method of $\lib'$, we show that an analogous transition is
possible with the same program state, but with memory state zeroing out
locations used by the library $\lib'$. Similarly, when we consider a
transition by $\cs{\client{\prog'}{\lib}}{\PSC}$ corresponding to an
instruction in a method of $\lib$, we show that an analogous
transition is possible with almost the same program state, except we
alter its stored program counter, and with memory state zeroing out locations
used by the client $\prog'$. The justifications for these steps follow
by the ($\Rightarrow$) directions of \cref{lem:memory_disjointness}.

Secondly, we compose the \emph{client-local} transition properties $\prog$
exhibits in $\tr_\lcl$ and the \emph{library-local} transition properties
$\lib$ exhibits in $\tr_\llib$ while constructing transitions of
$\cs{\client{\prog}{\lib}}{\PSC}$ for a trace $\tr$. Knowing that $\lib$ is
safe for $\prog$, we consider client-local transition properties from
$\tr_\lcl$ corresponding to transitions we wish to recreate in $\tr$, and
replace zeroed-out memory locations with locations of $\lib$. Dually, we
consider library-local transition properties from $\tr_\llib$ corresponding to
transitions we wish to recreate in $\tr$, and replace zeroed-out memory
locations with locations of $\prog$. The ($\Leftarrow$) directions of
\cref{lem:memory_disjointness} justify such transformations.
For instance,  non-$\sflab$-transitions
can be composed, provided that the
client program preserves the library memory state, and vice versa; while
crashes and $\sflab$-transitions record an interaction
between a client program and a library and therefore need to be performed in
synchrony.

We use these two ideas in proving \cref{lem:composition} by induction on the sum
of lengths of $\tr_\lcl$ and $\tr_\llib$, and use their local transition
properties to justify composing them in synchrony. For
the base case, we can simply take $\tr = \epsilon$. For the induction step, we
consider the last labels in $\tr_\lcl$ and $\tr_\llib$, as well as the cases
when one of the traces is empty. When $\tr_\lcl = \_ \cdot
\alpha_\lcl$ and $\tr_\llib = \_ \cdot \alpha_\llib$, we use $\tr'$ from the
induction hypothesis for $\tr_\lcl$ and $\tr_\llib$ with the last action removed
from either or both of them, and let $\tr = \tr' \cdot \alpha_\lcl$ or $\tr =
\tr' \cdot \alpha_\llib$.

Then, the abstraction theorem is proved as follows.

\cparagraph{Proof outline for \cref{thm:abs}}
It suffices to show $\fhistory{\client{\prog}{\lib}} \suq \fhistory{\client{\prog}{\lib^\sharp}}$;
then the claim follows using \cref{lem:composition} by letting $\lib:=\lib^\sharp$, $\lib':=\lib$, $\prog:=\prog$, and $\prog':=\prog$.
Suppose otherwise, and let $\history$ be a shortest history in
$\fhistory{\client{\prog}{\lib}} \setminus \fhistory{\client{\prog}{\lib^\sharp}}$.
Let $\tr$ be a shortest trace in $\traces{\cs{\client{\prog}{\lib}}{\PSC}}$ with $\fhistory{\tr}=\history$.
Consider the last transition label $\alpha$ in $\tr$.
The minimality of $\history$ and $\tr$ ensures that
$\alpha$ must be a return transition label for some $\method\in\dom{\lib}$.
Indeed, otherwise, we can show that $\alpha$ is enabled in the end of a corresponding trace of
$\cs{\client{\prog}{\lib^\sharp}}{\PSC}$, which contradicts the fact that
$\history \nin \fhistory{\client{\prog}{\lib^\sharp}}$.
(The full argument here requires applying  \cref{lem:composition} with
$\lib:=\lib^\sharp$, $\lib':=\lib$, $\prog:=\prog$, and $\prog':=\prog$.)

Now, using the fact that
$\prog$ correctly calls $\lib^\sharp$ \wrt $\MGCn$,
we again apply \cref{lem:composition}
with $\lib:=\lib$, $\lib':=\lib^\sharp$, $\prog:=\MGCn$, and $\prog':=\prog$,
and derive that $\alpha$ is enabled in the end of a corresponding trace of $\cs{\MGC{\lib}}{\PSC}$.
Then,
$\lib \sqsubseteq_\MGCn \lib^\sharp$ ensures that
$\rsthistory{\dom{\lib}}{\tr}\in \rsthistory{\dom{\lib}}{\MGC{\lib^\sharp}}$.
Using \cref{lem:composition} for the last time (applied with $\lib:=\lib^\sharp$, $\lib':=\lib$, $\prog:=\prog$, and $\prog':=\MGCn$),
we obtain that $\history=\fhistory{\tr}\in \fhistory{\client{\prog}{\lib^\sharp}}$,
which contradicts our assumption.
\qed\vspace{0.5ex}

The following corollary of \cref{thm:abs} states that, like
classical linearizability, our correctness condition is compositional (\aka local),
meaning that a library consisting of several
(non-interacting) libraries can be abstracted by considering each sub-library separately.
Formally,
the composition of libraries $\lib_1 \til \lib_n$ with pairwise disjoint sets of declared methods,
denoted by $\lib_1 \uplustil \lib_n$, is defined to be the library obtained by taking the union of $\lib_1 \til \lib_n$.
Compositionality is formulated as follows.

\begin{restatable}[Compositionality]{corollary}{compositionality}
\label{cor:compositionality}
{\rm
The following two conditions together imply that
$\lib_1 \uplustil \lib_n \sqsubseteq_\MGCn \lib_1^\sharp \uplustil \lib_n^\sharp$:
\begin{myenum}
\item $\usedlocs{\lib_1} \til \usedlocs{\lib_n},
\usedlocs{\lib_1^\sharp} \til  \usedlocs{\lib_n^\sharp},
\usedclientlocs{\dom{\lib_1 \uplustil \lib_n}}{\MGCn}$ are
pairwise disjoint.
\item
 For all $i$,
 $\lib_i \sqsubseteq_{\MGCn_i} \lib_i^\sharp$
 for $\MGCn_i=\MGC{\lib_1^\sharp \uplustil \lib_{i-1}^\sharp \uplus \lib_{i+1}^\sharp
\uplustil \lib_{n}^\sharp}$.
\end{myenum}
}
\end{restatable}

To end this section, we provide a simple lemma
that allows one to establish $\lib \sqsubseteq_\MGCn \lib^\sharp$
by applying standard simulation arguments
for \emph{crashless} traces
(with observable transitions being those that induce history labels).
For that matter, we require a simulation relation on non-volatile memories
generated by $\cs{\MGC{\lib}}{\PSC}$ and $\cs{\MGC{\lib^\sharp}}{\PSC}$
that holds for the very initial memory
and preserved %
during crashless executions.

\begin{restatable}{lemma}{sim}
\label{lem:sim}
{\rm
A trace $\tr$ %
is
\emph{$\mem_0$-to-$\mem$} if
$\tup{\progstate_\Init,\memstate_\Init[\lmem\mapsto\mem_0]}
\asteplab{\tr}{\cs{\prog}{\PSC}}
\tup{\progstate,\memstate[\lmem\mapsto\mem]} $
for some $\progstate$ and $\memstate$.
Suppose that some relation
$R$ on $\NVLoc\to\Val$ satisfies:
\begin{myitemize}
\item $\tup{\mem_\Init,\mem_\Init} \in R$.
\item If $\tup{\mem_0,\mem^\sharp_0} \in R$,
then for every $\mem_0$-to-$\mem$ crashless trace $\tr$ of $\cs{\MGC{\lib}}{\PSC}$,
there exist a non-volatile memory $\mem^\sharp$ and
an $\mem^\sharp_0$-to-$\mem^\sharp$ crashless trace $\tr^\sharp$ of $\cs{\MGC{\lib^\sharp}}{\PSC}$,
such that
$\tup{\mem,\mem^\sharp} \in R$
and
$\fhistory{\tr} = \fhistory{\tr^\sharp}$.
\end{myitemize}
Then, assuming $\dom{\lib}=\dom{\lib^\sharp}$, we have that $\lib \sqsubseteq_\MGCn \lib^\sharp$.
}
\end{restatable}

Furthermore, if $\MGC{\lib^\sharp}$ has no $\foInst{\cdot}$ and $\sfenceInst$ instructions, then
$\cs{\MGC{\lib^\sharp}}{\PSC}$
can take non-deterministic sfence steps (see \cref{sec:programs})
when $\cs{\MGC{\lib}}{\PSC}$ takes $\lSF$-steps,
so store fences can be ignored when checking
$\fhistory{\tr} = \fhistory{\tr^\sharp}$.

\newcommand{\codelabels}[1]{\kw{{#1}\texttt{:}\;}}
\newcommand{\codelabel}[1]{\kw{{#1}}}

\newcommand{\jobs}{\cnvloca{w}}
\newcommand{\xlockrw}{\tilde{\ensuremath{\mathtt{l}}}_{\mathtt{rw}}}

\newcommand{\mread}{\mathsf{read}}
\newcommand{\mwrite}{\mathsf{write}}
\newcommand{\mflush}{\mathsf{sync}}
\newcommand{\mpair}{\mathrm{pair}}
\newcommand{\mbpair}{\mathrm{bpair}}

\newcommand{\casInsta}[3]{\casInstn({#1},{#2},{#3})}

\newcommand{\xone}{\cnvloc{1}_\mathtt{1}}
\newcommand{\xtwo}{\cnvloc{1}_\mathtt{2}}
\newcommand{\xoneold}{\cnvloc{1}_\mathtt{1}^\mathtt{old}}
\newcommand{\xtwoold}{\cnvloc{1}_\mathtt{2}^\mathtt{old}}
\newcommand{\xonenew}{\cnvloc{1}_\mathtt{1}^\mathtt{new}}
\newcommand{\xtwonew}{\cnvloc{1}_\mathtt{2}^\mathtt{new}}
\newcommand{\xc}{\dot{\ensuremath{\mathtt{s}}}}
\newcommand{\xlock}{\tilde{\ensuremath{\mathtt{l}}}}
\newcommand{\xflag}{\dot{\ensuremath{\mathtt{f}}}}

\newcommand{\vxone}{\cvloc{1}_\mathtt{1}}
\newcommand{\vxtwo}{\cvloc{1}_\mathtt{2}}
\newcommand{\vxc}{\tilde{\ensuremath{\mathtt{s}}}}

\newcommand{\xoneprev}{\cnvloc{1}_\mathtt{1}^\mathtt{prev}}
\newcommand{\xtwoprev}{\cnvloc{1}_\mathtt{2}^\mathtt{prev}}
\newcommand{\xonenext}{\cnvloc{1}_\mathtt{1}^\mathtt{next}}
\newcommand{\xtwonext}{\cnvloc{1}_\mathtt{2}^\mathtt{next}}

\section{An Application: Persistent Pairs}
\label{sec:seqlock}

We illustrate the use of the library abstraction theorem for
a simple concurrent and persistent data structure---a pair of values that supports write and read operations.
We present two specifications and an implementation for each specification.
Both specifications ensure atomicity (\ie linearizability if the system does not crash),
and ``data consistency''
(reads return values written by a single write invocation),
but they differ in their persistency guarantees.
For the concurrency aspect, the implementations
follow the sequence lock (seqlock, for short) mechanism,
which uses a version counter along with the pair
and allows readers to avoid blocking~\cite{seqlock}.
For durability, the implementations employ different techniques:
one uses a ``redo log'' and the other is based on ``checkpoints''.

\cparagraph{A durable pair}
The first specification, a library we denote by $\lib^\sharp_\mpair$,
consists of three methods: $\mwrite$ for writing the two values of the pair,
$\mread$ for reading the pair,
and $\recoverMethod$ for recovering from a crash.
The specification is as follows:\footnote{Our simplified language has no mechanism for argument passing.
We assume that $\mwrite$ receives arguments ($
\mread$ returns results) via designated registers, $\creg{1}_1$ and
$\creg{1}_2$.}

{\smaller
\resizebox{\mycodefactor\width}{!}{
\begin{minipage}[t]{0.45\textwidth}
$$\inarr{
\underline{\mwrite{:}} \\
\codelabels{LOCK} \ifGotoInsts{\casInsta{\xlock}{0}{1}}{\codelabel{LOCK}} \sep \\
\beginpbInst{\xone,\xtwo} \sep  \\
\writeInst{\xone}{\creg{1}_1} \sep \writeInst{\xtwo}{\creg{1}_2} \sep \\
\myendpbInst{\xone,\xtwo} \sep  \\
\flInst{\xone} \sep\\
\codelabels{UNLOCK} \writeInst{\xlock}{0} \sep \\
\returnInst{} \sep
}$$
\end{minipage}}
\hfill
\resizebox{\mycodefactor\width}{!}{
\begin{minipage}[t]{0.45\textwidth}
$$\inarr{
\underline{\mread:} \\
\codelabels{LOCK} \ifGotoInsts{\casInsta{\xlock}{0}{1}}{\codelabel{LOCK}} \sep \\
\readInst{\creg{1}_1}{\xone} \sep\readInst{\creg{1}_2}{\xtwo} \sep \\
\codelabels{UNLOCK} \writeInst{\xlock}{0} \sep \\
\returnInst{} \sep
\\
\\
\underline{\recoverMethod:} \\
\returnInst{} \sep
}$$
\end{minipage}}
}

\noindent
A volatile lock ($\xlock$) is used to ensure atomicity.
For durability, writes use persistence blocks, which ensure that the two parts of the pair
persist simultaneously. After the block is ended,  $\flInst{\xone}$ (equivalent here
to $\flInst{\xtwo}$ due to the persistence block) ensures that the block persists.
If the system crashes after a write completed, the written values are guaranteed to
survive the crash. Thus, there is nothing to be done at recovery.
Nevertheless, aiming to allow implementations,
the library policy requires that recovery is executed after every crash before other methods are invoked
($\MGCrec$ in \cref{sec:lib}).%

Next, we present an implementation of $\lib^\sharp_\mpair$, which we denote by
$\lib_\mpair$. We write $\writeInst{\loc}{\loca}$ instead of
a read of $\loca$ (to some fresh register) followed by a write to $\loc$. We
also omit some necessary register bookkeeping: since histories record the whole register
store in call/return labels, strictly speaking, implementations must unroll
changes to registers not used to pass return values.

{
\vspace{\mycodeoffset}\vspace{\mycodeoffset}
\smaller
\noindent
\resizebox{\mycodefactor\width}{!}{
\begin{minipage}[t]{0.45\textwidth}
$$\inarr{
\underline{\mwrite{:}} \\
\codelabels{LOCK} \ifGotoInsts{\casInsta{\xlock}{0}{1}}{\codelabel{LOCK}} \sep \\
\writeInst{\xonenew}{\creg{1}_1} \sep \foInst{\xonenew} \sep
\writeInst{\xtwonew}{\creg{1}_2} \sep \foInst{\xtwonew} \sep \\
\sfenceInst \sep \\
\writeInst{\xc}{\xc+1} \sep \flInst{\xc} \sep \\
\writeInst{\xone}{\creg{1}_1} \sep \foInst{\xone} \sep
\writeInst{\xtwo}{\creg{1}_2} \sep \foInst{\xtwo} \sep \\
\sfenceInst \sep \\
\writeInst{\xc}{\xc+1} \sep \\
\codelabels{UNLOCK} \writeInst{\xlock}{0} \sep \\
\returnInst{} \sep
}$$
\end{minipage}}
\hfill
\resizebox{\mycodefactor\width}{!}{
\begin{minipage}[t]{0.2\textwidth}
$$\inarr{
\underline{\mread:} \\
\codelabels{BEGIN}\readInst{\creg{1}}{\xc} \sep \\
\ifGotoInsts{\mathsf{odd}(\creg{1})}\codelabel{BEGIN} \sep \\
\readInst{\creg{1}_1}{\xone} \sep\readInst{\creg{1}_2}{\xtwo} \sep \\
\ifGotoInsts{\xc \neq \creg{1}}\codelabel{BEGIN} \sep  \\
\returnInst{} \sep
}$$
\end{minipage}}
\hfill
\resizebox{\mycodefactor\width}{!}{
\begin{minipage}[t]{0.25\textwidth}
$$\inarr{
\underline{\recoverMethod:} \\
\ifGotoInsts{\mathsf{even}(\xc)}{\codelabel{END}} \sep \\
\writeInst{\xone}{\xonenew} \sep \foInst{\xone} \sep \\
\writeInst{\xtwo}{\xtwonew} \sep \foInst{\xtwo} \sep \\
\sfenceInst \sep \\
\codelabels{END} \writeInst{\xc}{0} \sep \\
\returnInst{} \sep
}$$
\end{minipage}}
}

\smallskip

\noindent
Ignoring crashes, atomicity is guaranteed here using a seqlock.
As for persistency, observe first that writing directly to the NVM is wrong
since we cannot control the non-deterministic propagation: if a crash occurs
during the execution of $\mwrite$, it is possible that only one part
of the pair has persisted, and the recovery method will not have sufficient information
for reinitializing the pair correctly.
Instead, $\mwrite$ first records its ``job'' in $\tup{\xonenew,\xtwonew}$.
Then, if a crash happens and the write was in the middle of updating
$\tup{\xone,\xtwo}$ (as identified via observing an odd version number),
 the recovery will complete the job of the writer.
 We note that the (rather extensive) use of flushes
 (or flush-optimals followed by an sfence)
 is necessary here in order to restrict the out-of-order
 persistence.
 The final write to $\xc$ in $\mwrite$ does not have to be explicitly persisted.
 Indeed, if a crash happens between this write and its persistence,
 recovery will redo the (idempotent) job.

\begin{restatable}{theorem}{pair}
    \label{thm:pair}
    $\lib_\mpair \sqsubseteq_{\MGCrec} \lib_\mpair^\sharp$.
\end{restatable}
\noindent Our proof sketch uses \cref{lem:sim}, letting
$\tup{\mem,\mem^\sharp}\in R$ if the following hold:
\begin{myitemize}
    \item If $\mem(\xc)$ is even, then $\mem(\xone)= \mem^\sharp(\xone)$ and $\mem(\xtwo)= \mem^\sharp(\xtwo)$.
    \item If $\mem(\xc)$ is odd,  then $\mem(\xonenew)= \mem^\sharp(\xone)$ and $\mem(\xtwonew)= \mem^\sharp(\xtwo)$.
\end{myitemize}

Using the abstraction theorem, we obtain that for a program $\prog$
that uses $\lib_\mpair$ correctly (\ie calls recovery first after every crash),
for every state $\tup{\progstate,\memstate}$
that is reachable in $\cs{\client{\prog}{\lib_\mpair}}{\PSC}$,
there exists a state $\tup{\progstate^\sharp,\memstate^\sharp}$
reachable in $\cs{\client{\prog}{\lib_\mpair^\sharp}}{\PSC}$
and indistinguishable from $\tup{\progstate,\memstate}$ from the client perspective.

\cparagraph{A buffered durable pair}
A second specification,
denoted by $\lib^\sharp_\mbpair$,
allows for ``buffered'' behaviors, which enable faster implementations by weakening persistency guarantees~\cite{persistent-lin}.
Instead of requiring operations to persist before returning, it only requires that operations are ``persistently ordered'' before returning.

{\vspace{\mycodeoffset}
\smaller\noindent
\resizebox{\mycodefactor\width}{!}{
\begin{minipage}[t]{0.38\textwidth}
$$\inarr{
\underline{\mwrite{:}} \\
\codelabels{LOCK} \ifGotoInsts{\casInsta{\xlock}{0}{1}}{\codelabel{LOCK}} \sep \\
\beginpbInst{\xone,\xtwo} \sep  \\
\writeInst{\xone}{\creg{1}_1} \sep \writeInst{\xtwo}{\creg{1}_2} \sep \\
\myendpbInst{\xone,\xtwo} \sep \\
\codelabels{UNLOCK} \writeInst{\xlock}{0} \sep \\
\returnInst{} \sep
}$$
\end{minipage}}
\hfill
\resizebox{\mycodefactor\width}{!}{
\begin{minipage}[t]{0.36\textwidth}
$$\inarr{
\underline{\mread{:}} \\
\codelabels{LOCK} \ifGotoInsts{\casInsta{\xlock}{0}{1}}{\codelabel{LOCK}} \sep \\
\readInst{\creg{1}_1}{\xone} \sep\readInst{\creg{1}_2}{\xtwo} \sep \\
\codelabels{UNLOCK} \writeInst{\xlock}{0} \sep \\
\returnInst \sep
}$$
\end{minipage}}
\hfill
\resizebox{\mycodefactor\width}{!}{
\begin{minipage}[t]{0.28\textwidth}
$$\inarr{
\underline{\recoverMethod{:}} \\
\returnInst{} \sep
\\
\\
\underline{\mflush{:}} \\
\flInst{\xone} \sep\\
\returnInst{} \sep
}$$
\end{minipage}}
}

Compared to $\lib^\sharp_\mpair$, the explicit flush instruction $\flInst{\xone}$ from the write method is omitted,
which means that a crash after a completed write may take the pair back to its state before the write.
Thus, the state after a crash need not necessarily be fully up-to-date.
An additional method, called $\mflush$, can used to
ensure that previous writes have persisted.
Without $\mflush$, an implementation could simply ignore persistency and
store the pair in the volatile memory, which corresponds to an execution of $\lib^\sharp_\mbpair$
in which the persistency buffers are never being flushed.

An implementation can be obtained as follows:

{\vspace{\mycodeoffset}\vspace{\mycodeoffset}\smaller
\noindent
\resizebox{\mycodefactor\width}{!}{
\begin{minipage}[t]{0.24\textwidth}
$$\inarr{
\underline{\mwrite{:}} \\
\codelabels{LOCK} \ifGotoInsts{\casInsta{\xlock}{0}{1}}{\codelabel{LOCK}} \sep \\
\writeInst{\vxc}{\vxc+1} \sep \\
\writeInst{\vxone}{\creg{1}_1} \sep
\writeInst{\vxtwo}{\creg{1}_2} \sep  \\
\writeInst{\vxc}{\vxc+1} \sep \\
\codelabels{UNLOCK} \writeInst{\xlock}{0} \sep \\
\returnInst{} \sep
}$$
\end{minipage}}
\hfill
\resizebox{\mycodefactor\width}{!}{
\begin{minipage}[t]{0.38\textwidth}
$$\inarr{
\underline{\mread{:}} \\
\codelabels{BEGIN}\readInst{\creg{1}}{\vxc} \sep \\
\ifGotoInst{\mathsf{odd}(\creg{1})}\codelabel{BEGIN} \sep \\
\readInst{\creg{1}_1}{\vxone} \sep\readInst{\creg{1}_2}{\vxtwo} \sep \\
\ifGotoInst{\vxc \neq \creg{1}}\codelabel{BEGIN} \sep  \\
\returnInst \sep
\\[1ex]
\underline{\recoverMethod{:}} \\
\ifGotoInst{\xflag=1}{\codelabel{PREV}} \sep \\
\writeInst{\vxone}{\xonenext} \sep \writeInst{\vxtwo}{\xtwonext} \sep \\
\returnInst \sep \\
\codelabels{PREV}
\writeInst{\vxone}{\xoneprev} \sep \writeInst{\vxtwo}{\xtwoprev} \sep \\
\writeInst{\xflag}{0} \sep \flInst{\xflag} \sep \\
\returnInst \sep
}$$
\end{minipage}}
\hfill
\resizebox{\mycodefactor\width}{!}{
\begin{minipage}[t]{0.38\textwidth}
$$\inarr{
\underline{\mflush{:}} \\
\codelabels{LOCK} \ifGotoInsts{\casInsta{\xlock}{0}{1}}{\codelabel{LOCK}} \sep \\
\readInst{\creg{1}_1}{\vxone} \sep\readInst{\creg{1}_2}{\vxtwo} \sep \\
\writeInst{\xoneprev}{\xonenext} \sep \foInst{\xoneprev} \sep \\
\writeInst{\xtwoprev}{\xtwonext} \sep \foInst{\xtwoprev} \sep \\
\sfenceInst \sep \\
\writeInst{\xflag}{1} \sep \flInst{\xflag} \sep \\
\codelabels{NEXT} \writeInst{\xonenext}{\creg{1}_1} \sep \foInst{\xonenext} \sep \\
\writeInst{\xtwonext}{\creg{1}_2} \sep \foInst{\xtwonext} \sep \\
\sfenceInst \sep \\
\writeInst{\xflag}{0} \sep \flInst{\xflag} \sep \\
\codelabels{UNLOCK} \writeInst{\xlock}{0} \sep \\
\returnInst \sep
}$$
\end{minipage}}
}

\smallskip
This implementation
exploits the freedom allowed by the specification.
Writes and reads again employ a seqlock,
but this time they only use volatile variables.
In turn, $\mflush$ sets a ``checkpoint'',
and recovery rolls the state back to the latest complete checkpoint.
For that matter, a non-volatile flag $\xflag$ is used to detect crashes during the setting the checkpoint $\tup{\xonenext,\xtwonext}$.
Thus, before storing the checkpoint, the previous checkpoint is stored in the non-volatile variables $\tup{\xoneprev,\xtwoprev}$.
Upon recovery, given the value of the flag, we know if we can restore the state from the current stored checkpoint,
or, if a crash happened during the store of this checkpoint (which means that $\mflush$ did not return),
set the pair to the previous stored one.

\begin{restatable}{theorem}{bpair}
\label{thm:bpair}
$\lib_\mbpair \sqsubseteq_{\MGCrec} \lib_\mbpair^\sharp$.
\end{restatable}

\noindent Our proof sketch uses \cref{lem:sim}, letting
$\tup{\mem,\mem^\sharp}\in R$ if the following hold:
\begin{myitemize}
\item If $\mem(\xflag)=0$, then $\mem(\xonenext)= \mem^\sharp(\xone)$ and $\mem(\xtwonext)= \mem^\sharp(\xtwo)$.
\item If $\mem(\xflag)=1$,  then $\mem(\xoneprev)= \mem^\sharp(\xone)$ and $\mem(\xtwoprev)= \mem^\sharp(\xtwo)$.
\end{myitemize}

\section{Related and Future Work}
\label{sec:related}

\cparagraph{Library abstraction theorems}
Previous work has developed library abstraction theorems for crashless shared memory concurrency.
First, \citet{Filipovic10} formalized the
intuition that standard linearizability as defined in~\cite{herlihy-lin} corresponds to contextual refinement
(and also proved a completeness result: the converse also holds provided that threads
have other means of interaction besides the library).
Later, \citet{Bouajjani15} refined and formulated this result using history inclusion instead of linearizability,
which is closer to our formalization.
Other abstraction results account for liveness~\cite{Gotsman11_liveness}, resource-transferring programs~\cite{Gostman_13_ownership},
and x86-TSO~\cite{Burckhardt_2012_library_tso}.
Our composition lemma (\cref{lem:composition}) is inspired by \cite{Burckhardt_2012_library_tso},
which addresses a challenge that is close to the challenge posed by store fence instructions in NVM,
where actions of the client and the library affect each other even if they access to distinct locations.
To do so, the notion of a history is extended to expose events that correspond
to the flushing certain entries from the x86-TSO store buffers, which is
close to what we do to handle store fences.
Our alternative approach to this problem, \ie introducing a relaxed version of the store fence, is novel.

While our framework is operational, library abstraction was also studied before
for declarative shared memory concurrency semantics, particularly in the context
of the C11 weak memory model~\cite{Batty:2013,yacovet}.

\cparagraph{Linearizability notions for persistent objects}
Different approaches for adapting the standard linearizability criterion that is based on
crash-free sequential specifications~\cite{herlihy-lin} were proposed before~\cite{aguilera2003strict,Guerraoui04,persistent-lin},
but were not formally related to contextual refinement.
Since methods like $\recoverMethod$ and $\mflush$ (see \cref{sec:seqlock})
are meaningless in crash-free sequential specifications, they require an ad-hoc external treatment in these linearizability adaptations.
The variety of approaches to interpret crash-free sequential specifications for crash-resilient concurrent objects
makes it hard, in particular, to combine libraries with different linearizability guarantees in a single program.

In turn, these existing notions are typically expressible in the refinement framework that we employ.
For example, in the \emph{crashless} setting,
by wrapping each method of a sequential implementation $S$ of some object inside a global lock,
one obtains an abstract library $\lib_{S}^\sharp$ for that object that corresponds to the conditions imposed
by standard linearizability~\cite{Bouajjani15}
(a library $\lib$ is linearizable \wrt $S$ iff
every crashless history induced by a trace of $\MGC{\lib}$
is also induced by some trace of $\MGC{\lib_S^\sharp}$).
Now, when crashes are involved, by wrapping each method of $S$
inside a global lock \emph{and a persistence block} followed by an explicit flush instruction
(like $\lib_\mpair^\sharp$ in \cref{sec:seqlock}),
one obtains an abstract library $\lib_{S\crash}^\sharp$ that corresponds to the conditions imposed
by strict linearizability of \citet{aguilera2003strict}
($\lib$ is strictly linearizable \wrt $S$ iff
$\lib \sqsubseteq_\MGCn \lib_{S\crash}^\sharp$).
Thus, our results can be used to derive contextual refinement (using
$\lib_{S\crash}^\sharp$ as a specification) from strictly linearizable objects.
We note that while the original definition of   strict linearizability was for a model with per-processor failure,
what we consider here is its application for full system crashes.

Durable linearizability~\cite{persistent-lin} weakens strict linearizability by
allowing methods that were active during a crash to take their effect at any
later point in the execution (or never), instead of requiring that the effect of
such methods is visible immediately after the crash (or never).
This weakening aims to allow
lazy recovery for large structures, where either the recovery procedure is
executed in parallel to other methods after a crash, or the methods themselves
participate in recovering the data structure when they are further executed.
 This notion can be also expressible as an abstract implementation in our
 language.
For this matter, every update method in the specification would: first record its task in a
work-set; remove the task from the work-set; flush the updated work-set; and
perform the task like in $\lib_{S\crash}^\sharp$ described above. In turn, every
query method may choose to complete any task it finds in the work-set, since
the method performing such a task has crashed
during its invocation. For persistent pairs (see
\cref{sec:seqlock}), this is illustrated by the specification below. The non-volatile
variable $\jobs$ is the multiset holding the work-set with atomic add and remove
operations, and $\xlockrw$ is an abstract
multiple-readers-single-writer lock
used to resolve races on the work-set.

\vspace{\mycodeoffset}
{\smaller\noindent
\resizebox{\mycodefactor\width}{!}{
\begin{minipage}{0.4\textwidth}
\[\inarr{
\underline{\mwrite{:}} \\
\codelabels{LOCK1} \text{acquire }\xlockrw \text{ as a reader} \sep \\
\text{add }\tup{\creg{1}_1,\creg{1}_2}\text{ to } \jobs \sep \\
\text{remove }\tup{\creg{1}_1,\creg{1}_2}\text{ from } \jobs \sep \\
\flInst{\jobs} \sep\\
\codelabels{UNLOCK1} \text{release }\xlockrw\sep \\
\ldots \text{continue as in $\mwrite$ of $\lib_\mpair^\sharp$ (\cref{sec:seqlock})} \ldots
\\
\underline{\recoverMethod:} \\
\returnInst{} \sep
}\]
\end{minipage}}
\hfill
\resizebox{\mycodefactor\width}{!}{
\begin{minipage}{0.6\textwidth}
$$\inarr{
\underline{\mread{:}} \\
\GotoInst{\set{\codelabel{LOCK1},\codelabel{BEGIN}}} \sep \\
\codelabels{LOCK1} \text{acquire }\xlockrw \text{ as a writer} \sep \\
\text{pick\ some\ }\tup{\creg{1}_1,\creg{1}_2}\in\jobs \sep \\
\text{remove }\tup{\creg{1}_1,\creg{1}_2}\text{ from } \jobs \sep \\
\flInst{\jobs} \sep\\
\ldots \text{write $\tup{\creg{1}_1,\creg{1}_2}$ to $\tup{\cloc{1},\cloc{2}}$ as in $\mwrite$ of $\lib_\mpair^\sharp$ (\cref{sec:seqlock})} \ldots \\
\codelabels{UNLOCK1} \text{release }\xlockrw\sep \\
\codelabels{BEGIN} \ldots \text{continue as in $\mread$ of $\lib_\mpair^\sharp$ (\cref{sec:seqlock})} \ldots
}$$
\end{minipage}}
}
\smallskip

A ``buffered'' version of strict linearizability, which only requires the
existence of a prefix of the completed invocations to be observed after a crash,
is also naturally derived by considering $\lib_{S\crash\mathrm{b}}^\sharp$ which
is obtained from a sequential implementation $S$ by wrapping each method of $S$
inside a global lock and a persistence block (\emph{without} an explicit flush
instruction) and ensuring that there is a single non-volatile variable that is
written to by all library methods (introducing such a variable if needed).\footnote{Since
the corresponding ``buffered'' correctness notion is not compositional, while
the refinement-based notion is (see \cref{cor:compositionality}), one cannot
expect to have a per-object translation of a sequential implementation $S$ into
a concurrent and persistent implementation $\lib_{S\crash\mathrm{b}}^\sharp$.
Indeed, the addition of a single non-volatile variable that is written to by all
library methods is a not a per-object translation (\ie for two sequential
library implementations implementing disjoint sets of methods and operating on
disjoint variables, $S_1$ and $S_2$, we will \emph{not} have $\lib_{S_1 \cup
S_2\crash\mathrm{b}}^\sharp= \lib_{S_1\crash\mathrm{b}}^\sharp \cup
\lib_{S_2\crash\mathrm{b}}^\sharp$).}

An alternative operational characterization of durable linearizability using Input/Output automata was
developed in \cite{derrick2021verifying} and used to formally establish this property
for the persistent queue of \cite{Friedman-persistent-queue}
by providing a full-blown simulation proof using the KIV proof assistant.\footnote{
See \url{https://kiv.isse.de/projects/Durable-Queue.html}.}
Nevertheless, this work does not relate the proved correctness criterion to contextual refinement.

\cparagraph{Persistency models}
The underlying model we assume is \PSC by~\cite{Khyzha2021},
a strengthening of \PTSO~\cite{pxes-popl} that formalizes the Intel-x86 persistency.
The paper \cite{Khyzha2021} provided compiler mappings that ensure \PSC semantics on machines guaranteeing \PTSO semantics.
We extended the general semantic framework with libraries,
and extended \PSC with local store fences and persistence blocks.

\cparagraph{Future work}
Future work includes extending our proof method and results for weaker persistency models,
such as persistent x86-TSO~\cite{pxes-popl} and ARM~\cite{Kyeongmin2021};
handling random access shared memory with allocations and deallocations
(instead of the simplified shared variables model we employ);
and lifting the strict condition that libraries and clients live in disjoint address spaces
by allowing them to transfer ownership of certain locations (as
was done in \cite{Gostman_13_ownership} for standard volatile memory).

In addition, extending and adapting methods for refinement verification under volatile memory
is needed in order to provide library developers with means to validate our library-correctness conditions.
Such methods may include automated checking by approximation~\cite{Bouajjani15},
layered interactive verification in the style of~\cite{Hawblitzel15,Lorch20},
and formal logics as the one in~\cite{Liang14}.
Similarly, developing formal methods and tools that allow using library specifications for client reasoning
is left for future work, including decidable reachability analysis~\cite{AHH+2013},
program logics~\cite{Raad20_pog}, and principled testing~\cite{Gorjiara21}.
Finally, it is interesting to see how logical atomicity notions established by program logics, such as~\cite{daRochaPinto2014TaDA,Svendsen13},
which has been extended to cover crashes in disk-based storage systems~\cite{Chajed21},
can be adapted for establishing our correctness condition and/or for client reasoning.

 \bibliographystyle{splncs04}

\bibliography{biblio}

\iflong

\clearpage
\appendix

\section{Extension with RMWs}
\label{app:rmw}

In this section we present the extension of our definitions with
two atomic read-modify-write instructions (RMWs).

\begin{enumerate}

\item The programming language is extended with two additional instructions:

\[
	\begin{array}{@{} l l @{}}
\instr ::=  &
... \ALT \incInst{r}{\loc}{\exp}
\ALT \casInst{r}{\loc}{\exp}{\exp}
\end{array}
\]

Intuitively, a fetch-and-add $\incInst{\reg}{\loc}{\exp}$
loads the value from a variable $\loc$ into $\reg$ and increments the value in memory;
and a compare-and-swap $\casInst{\reg}{\loc}{\exp_1}{\exp_2}$ that
loads the value from $\loc$ into $\reg$ and overwrites it with the second
expression in case if the loaded value coincides with the first expression.

\item Action labels (\cref{def:label}) are extended to include additional labels:

\smallskip
\begin{tabular}{rlrl}
	a \emph{read-modify-write} & $\ulab{\loc}{\val_\lR}{\val_\lW}$ &
	\qquad a \emph{failed CAS} & $\rexlab{\loc}{\val_\lR}$
\end{tabular}

\smallskip
where $\loc \in \Loc$ and $\val_\lR,\val_\lW \in \Val$.

Fetch-and-add and successful compare-and-swap instructions are captured by an RMW
label ($\ulab{\loc}{\val_\lR}{\val_\lW}$), whereas compare-and-swaps
that did not read the expected value correspond to special read label
$\rexlab{\loc}{\val_\lR}$, which allows us to distinguish such transitions
from plain reads and provide them with stronger semantics.

\item The transitions of the LTS induced by an instruction sequence $\iseq$ are (\cref{def:lts_iseq})
are extended with:

\begin{mathpar}
\inferrule*{
	\iseq(\pc)=\incInst{\reg}{\loc}{\exp}
	\\\\ \lab = \ulab{\loc}{\val}{\val+\phi(\exp)}
	\\\\ \phi'=\phi[\reg \mapsto \val]
}{
	\tup{\pc,\phi}%
	\asteplab{\lab}{\iseq}
	\tup{\pc+1, \phi'}%
}
\and
\inferrule*{
	\iseq(\pc)=\casInst{\reg}{\loc}{\exp_\lR}{\exp_\lW}
	\\\\ \lab = \ulab{\loc}{\phi(\exp_\lR)}{\phi(\exp_\lW)}
	\\\\
	\\\\ \phi'=\phi[\reg \mapsto \phi(\exp_\lR)]
}{
	\tup{\pc,\phi}%
	\asteplab{\lab}{\iseq}
	\tup{\pc+1,\phi'}%
}
\and
\inferrule*{
	\iseq(\pc)=\casInst{\reg}{\loc}{\exp_\lR}{\exp_\lW}
	\\\\ \lab = \rexlab{\loc}{\val}
	\\\\ \val\neq \phi(\exp_\lR)
	\\\\ \phi'=\phi[\reg \mapsto \val]
}{
	\tup{\pc,\phi}%
	\asteplab{\lab}{\iseq}
	\tup{\pc+1, \phi'}%
}
\end{mathpar}

The
fetch-and-add $\incInst{\reg}{\loc}{\exp}$ loads an arbitrary value into the
destination register $\reg$, announcing in the label
$\ulab{\loc}{\val}{\val+\phi(\exp)}$ that the value loaded is $\val$ and the
value stored is $\val + \phi(\exp)$.
The two transitions for the
compare-and-swap $\casInst{\reg}{\loc}{\exp_\lR}{\exp_\lW}$ both load an
arbitrary value $\val$ into the destination register $\reg$, except the
successful compare-and-swap transition announces in the label
$\ulab{\loc}{\phi(\exp_\lR)}{\phi(\exp_\lW)}$ that $\phi(\exp_\lR)$ is the
value loaded and $\phi(\exp_\lW)$ is the value stored, while the failed
compare-and-swap announces in the label $\rexlab{\loc}{\val}$ that the value
loaded that is different from $\phi(\exp_\lR)$.

\item The transitions of \PSC are extended with:

  \begin{mathpar}
  \inferrule[v-rmw]{
  	\lab = \ulab{\vloc}{\val_\lR}{\val_\lW}
  \\\\
  	\memstate
  	\asteptidlab{\tid}{\rlab{\vloc}{\val_\lR}}{\PSC}
  	\asteptidlab{\tid}{\wlab{\vloc}{\val_\lW}}{\PSC}
  	\memstate'
  }{
  	\memstate
  	\asteptidlab{\tid}{\lab}{\PSC}
  	\memstate'
  }
  \and
  \inferrule[v-rmw-fail]{
  	\lab = \rexlab{\vloc}{\val}
  \\\\
  	\memstate
  	\asteptidlab{\tid}{\rlab{\vloc}{\val}}{\PSC}
  	\memstate'
  }{
  	\memstate
  	\asteptidlab{\tid}{\lab}{\PSC}
  	\memstate'
  }
\\
  \inferrule[nv-rmw]{
  	\lab = \ulab{\nvloc}{\val_\lR}{\val_\lW}
  \\\\
  	\memstate
  	\asteptidlab{\tid}{\sflab}{\PSC}
  	\asteptidlab{\tid}{\rlab{\nvloc}{\val_\lR}}{\PSC}
  	\asteptidlab{\tid}{\wlab{\nvloc}{\val_\lW}}{\PSC}
  	\memstate'
  }{
  	\memstate
  	\asteptidlab{\tid}{\lab}{\PSC}
  	\memstate'
  }
  \and
  \inferrule[nv-rmw-fail]{
  	\lab = \rexlab{\nvloc}{\val}
  \\\\
  	\memstate
  	\asteptidlab{\tid}{\sflab}{\PSC}
  	\asteptidlab{\tid}{\rlab{\nvloc}{\val}}{\PSC}
  	\memstate'
  }{
  	\memstate
  	\asteptidlab{\tid}{\lab}{\PSC}
  	\memstate'
  }
  \end{mathpar}

We note that RMWs to non-volatile variables
(including those arising from failed compare-and-swap operations)
include an implicit sfence transition.

\item In \cref{lem:memory_disjointness}, we have to consider
RMWs to non-volatile variables as store fence transitions, since they
induce a store fence.

Thus, we define:

$$\SFLab \defeq \set{ \sflab} \cup
 \set{ \lab\in\Lab \mid \lTYP(\lab) \in \set{\lU,\lRex} \land \lLOC(\lab)\in\NVLoc}$$

and adapt the first two items of the lemma as follows:

\begin{enumerate}
\item
 For every $\tid\in\Tid$ and $\lab\in\Lab\setminus\SFLab$
with $\lLOCSET(\lab)\suq \locset$,\\
$\memstate_1 \asteptidlab{\tid}{\lab}{\PSC} \memstate_2
\Longleftrightarrow
(\memstate_1\rst{\locset} \asteptidlab{\tid}{\lab}{\PSC} \memstate_2\rst{\locset} \land
\memstate_1\rst{\Loc\setminus\locset}=\memstate_2\rst{\Loc\setminus\locset})
$
\item
 For every $\tid\in\Tid$ and $\lab\in\SFLab$
which is either $\sflab$ or has $\lLOC(\lab) \in \locset$,\\
$
\memstate_1 \asteptidlab{\tid}{\lab}{\PSC} \memstate_2
\Longleftrightarrow
(\memstate_1\rst{\locset} \asteptidlab{\tid}{\lab}{\PSC} \memstate_2\rst{\locset} \land
\memstate_1\rst{\Loc\setminus\locset} \asteptidlab{\tid}{\sflab}{\PSC} \memstate_2\rst{\Loc\setminus\locset})
$
\end{enumerate}

\item In the definition of the history induced by a trace of $\cs{\prog}{\PSC}$ (\cref{def:history_of_trace}),
we need to include an $\lSF$-label $\tidlab{\tid}{\sflab}$ for every store-fence inducing
label ($\tidlab{\tid}{\lab}$ with $\lab \in\SFLab$), rather than only for the $\lSF$ label.

\end{enumerate}

\clearpage %

\section{Proofs}
\label{app:proofs}

The following propositions are used in the following proofs.
The all easily follow from our definitions.

\begin{proposition}
\label{prop:h_prefix}
If $\history\in \rsthistory{\methodset}{\prog}$, then $\history'\in \rsthistory{\methodset}{\prog}$
for every prefix $\history'$ of $\history$.
\end{proposition}

\begin{proposition}
\label{prop:h_crash}
If $\history\in \rsthistory{\methodset}{\prog}$, then $\history \cdot \crash \in \rsthistory{\methodset}{\prog}$.
\end{proposition}

\begin{proposition}
\label{prop:h_sf}
If $\history\in \rsthistory{\methodset}{\prog}$, then $\history \cdot \tidlab{\tid}{\sflab} \in \rsthistory{\methodset}{\prog}$
for every $\tid\in\Tid$.
\end{proposition}

\begin{proposition}
\label{prop:same_f}
Suppose that $\tup{\progstate, \memstate} \asteplab{\tr_1}{\cs{\prog}{\PSC}} \tup{\progstate_1, \memstate_1}$
and $\tup{\progstate, \memstate} \asteplab{\tr_2}{\cs{\prog'}{\PSC}} \tup{\progstate_2, \memstate_2}$.
If $\rsthistory{\methodset}{\tr_1}=\rsthistory{\methodset}{\tr_2}$, then
for every $\tid$ we have $\progstate_1(\tid).\lmethod \in \methodset \Longleftrightarrow  \progstate_2(\tid).\lmethod \in \methodset$.
\end{proposition}

The following properties all assume a library $\lib$ that is safe for a program $\prog$.

\begin{proposition}
\label{prop:P_or_L1}
If $\progstate \asteptidlab{\tid}{\lab_\epsl}{\client{\prog}{\lib}} \progstate'$
and $\progstate(\tid).\lmethod \nin \dom{\lib}$,
then $\progstate \asteptidlab{\tid}{\lab_\epsl}{\prog} \progstate'$.
\end{proposition}

\begin{proposition}
\label{prop:prog_lib_separates}
For every state $\tup{\progstate,\memstate}$
reachable in $\cs{\client{\prog}{\lib}}{\PSC}$,
we have that both $\usedlocs{\lib} \cap\NVLoc$
and $\usedclientlocs{\dom{\lib}}{\prog} \cap\NVLoc$
separate $\memstate$.
\end{proposition}

\begin{proposition}
\label{prop:local}
The following hold whenever $\progstate \asteptidlab{\tid}{\lab}{\client{\prog}{\lib}} \progstate'$:
\begin{itemize}
\item
If $\progstate(\tid).\lmethod \in \dom{\lib}$,
then $\lLOCSET(\lab) \subseteq \usedlocs{\lib}$.
\item
If $\progstate(\tid).\lmethod \nin \dom{\lib}$,
then $\lLOCSET(\lab) \subseteq \usedclientlocs{\dom{\lib}}{\prog}$.
\end{itemize}
\end{proposition}

The following propositions easily follow from the definitions in \cref{sec:psc}.

\begin{proposition}
\label{prop:separates_complement}
A set $\nvlocset \suq \NVLoc$ separates $\memstate$
iff $\NVLoc\setminus \nvlocset$ separates $\memstate$.
\end{proposition}

\begin{proposition}
\label{prop:seperates_preserved}
If $\nvlocset\suq\NVLoc$ separates $\memstate_1$
and $\memstate_1 \asteplab{\alpha}{\PSC} \memstate_2$
with $\lLOCSET(\alpha)\suq \nvlocset$,
then $\nvlocset$ separates $\memstate_2$.
\end{proposition}

Under the conditions of \cref{def:compose_psc}, we always have the following properties:

\begin{lemma}\label{lem:merge}
Suppose $\locset_1,\locset_2\suq\Loc$ is such that
$\locset_1 \cap \locset_2 = \emptyset$. Then
\begin{enumerate}[(a)]
\item\label{item:merge1}
$\mergemem{\memstate_1}{\locset_1}{\memstate_2}{\locset_2}
=
\mergemem{\memstate_2}{\locset_2}{\memstate_1}{\locset_1}$.
\item\label{item:merge2}
$\mergemem{\memstate_1}{\locset_1}{\memstate_2}{\locset_2}
=
\mergemem{\memstate_1\rst{\locset_1}}{\locset_1}{\memstate_2}{\locset_2}$.
\item\label{item:merge3}
$(\mergemem{\memstate_1}{\locset_1}{\memstate_2}{\locset_2})\rst{Y}
=
\memstate_1\rst{\locset_1}$, for any $Y$ such that $\locset_1 \subseteq Y \subseteq \Loc\setminus\locset_2$.
\end{enumerate}
\end{lemma}

\composition*

\begin{proof}
Consider two libraries $\lib$ and $\lib'$ implementing the same method set
$\methodset$, both safe for a program $\prog$, with $\lib$ being also safe for
a program $\prog'$. For traces $\tr_\lcl$ and $\tr_\llib$, let
$\scLEM(\tr_\lcl, \tr_\llib)$ denote the rest of the statement of the lemma.
We prove $(\forall \tr_\lcl, \tr_\llib \ldotp \scLEM(\tr_\lcl,
\tr_\llib))$ by induction on the sum of lengths of $\tr_\lcl$ and $\tr_\llib$.

The base of induction is to show $\scLEM(\tr_\lcl, \tr_\llib)$ when
$\size{\tr_\lcl}+\size{\tr_\llib}=0$; then
$\tup{\progstate,\memstate}=\tup{\progstate_\Init,\memstate_\Init}$, and we
can simply take $\tr=\epsilon$.

The step of induction is to show $\scLEM(\tr_\lcl, \tr_\llib)$, assuming that
$\scLEM(\tr'_\lcl, \tr'_\llib)$ holds for every $\tr'_\lcl$ and $\tr'_\llib$
with $\size{\tr'_\lcl}+\size{\tr'_\llib} < \size{\tr_\lcl}+\size{\tr_\llib}$.
We split the rest of the proof into the following cases:
\begin{itemize}
\item[(I)] $\tr_\llib$ is non-empty and ends with a label $\alpha_\llib$ that
does not contribute to $\rsthistory{\methodset}{\tr_\llib}$, \ie one of the
following holds:
\begin{itemize}
	\item $\alpha_\llib = \persistlab$
	\item $\alpha_\llib \in \Tid \times \set{\epsilon}$
	\item $\alpha_\llib \in \Tid \times
		\set{\calllab{\method}{\phi}
			, \retlab{\method}{\phi} \in \Lab
			\mid \method \nin \methodset
		}$
	\item $\alpha_\llib \in \Tid \times \set{\lab \in \Lab \mid \lTYP(\lab)\nin\set{\lCALL,\lRET} \land \lab \nin \SFLab}$
\end{itemize}
\item[(II)] $\tr_\lcl$ is non-empty and ends with a label $\alpha_\lcl$ that
does not contribute to $\rsthistory{\methodset}{\tr_\lcl}$;
\begin{itemize}
	\item $\alpha_\lcl = \persistlab$
	\item $\alpha_\lcl \in \Tid \times \set{\epsilon}$
	\item $\alpha_\lcl \in \Tid \times
		\set{\calllab{\method}{\phi}
			, \retlab{\method}{\phi} \in \Lab
			\mid \method \nin \methodset
		}$
	\item $\alpha_\lcl \in \Tid \times \set{\lab \in \Lab \mid \lTYP(\lab)\nin\set{\lCALL,\lRET} \land \lab \nin \SFLab}$
\end{itemize}
\item[(III)] both $\tr_\lcl$ and $\tr_\llib$ are non-empty and end with
labels $\alpha_\lcl$ and $\alpha_\llib$ contributing to histories
$\rsthistory{\methodset}{\tr_\lcl}$ and $\rsthistory{\methodset}{\tr_\llib}$,
\ie one of the following holds:
\begin{itemize}
	\item $\alpha_\lcl = \alpha_\llib \in
		\Tid \times \set{\calllab{\method}{\phi} \in \Lab
											\mid \method \in \methodset}$
	\item $\alpha_\lcl = \alpha_\llib \in
		\Tid \times \set{\retlab{\method}{\phi} \in \Lab
											\mid \method \in \methodset}$
	\item $\alpha_\lcl = \alpha_\llib = \crash$
	\item $\alpha_\lcl, \alpha_\llib \in \Tid \times \SFLab$
\end{itemize}
\end{itemize}
It is easy to see that these three cases exhaust all possibilities for
$\tr_\llib$ and $\tr_\lcl$. For instance, suppose that $\tr_\llib$ is
non-empty, but ends with a label corresponding to a history label. Let $\tr_\llib = \_
\cdot \alpha_\llib$ and $\rsthistory{\methodset}{\tr_\llib} = \_ \cdot
\rsthistory{\methodset}{\alpha_\llib} $. By the lemma's premise,
$\rsthistory{\methodset}{\tr_\lcl} = \rsthistory{\methodset}{\tr_\llib}$.
Therefore, it must be that $\tr_\lcl = \_ \cdot \alpha_\lcl \cdot \tr'_\lcl$
and $\rsthistory{\methodset}{\tr_\lcl} = \_ \cdot
\rsthistory{\methodset}{\alpha_\lcl}$. However, when $\tr'_\lcl$ is non-empty,
such a possibility is already covered by Case II, and when $\tr'_\lcl$ is
empty, such a possibility is already covered by Case III.

\noindent\underline{\textsc{Case I.}}
Suppose that $\tr_\llib$ is non-empty and ends with a label $\alpha_\llib$ not
corresponding to a history label. Let $\tr_\llib = \tr'_\llib \cdot
\alpha_\llib$, and consider any state
$\tup{\progstate'_\llib,\memstate'_\llib}$ for which there are the following
transitions:
\begin{equation}\label{eq:comp1}\tag{I}
\tup{\progstate_\Init,\memstate_\Init}
{\asteplab{\tr'_\llib}{\cs{\client{\prog'}{\lib}}{\PSC}}}
\tup{\progstate'_\llib,\memstate'_\llib}
\asteplab{\alpha_\llib}{\cs{\client{\prog'}{\lib}}{\PSC}}
\tup{\progstate_\llib,\memstate_\llib}
\end{equation}
In the following, we consider differently various cases for $\alpha_\llib$ in
order to construct $\tr$.

{Suppose $\alpha_\llib = \persistlab$.}
The last transition of \cref{eq:comp1} is memory-internal, so
$\progstate'_\llib = \progstate_\llib$. Then $\progstate = \progstate'$ holds
by construction.
We deduce from \cref{eq:comp1} that $\memstate'_\llib
\asteplab{\persistlab}{\PSC}
\memstate_\llib$ holds. By \cref{prop:prog_lib_separates}, since $\lib$ is safe for $\prog'$, $\usedlocs{\lib} \cap \NVLoc$ separates $\memstate'_\llib$, which allows us to deduce $\memstate'_\llib\rst{\usedlocs{\lib}}
\asteplab{\persistlab}{\PSC}
\memstate_\llib\rst{\usedlocs{\lib}}$
by \cref{lem:memory_disjointness}(\ref{item:persist}). From the induction
hypothesis $\scLEM(\tr_\lcl, \tr'_\llib)$, we know that
$\memstate'=\nmergemem{\memstate_\lcl}{\usedclientlocs{\methodset}{\prog}}{\memstate'_\llib}{\usedlocs{\lib}}$,
so overall we obtain $\memstate'\rst{\usedlocs{\lib}}
\asteplab{\persistlab}{\PSC} \memstate\rst{\usedlocs{\lib}}$.
Also, $\memstate'_\lcl\rst{\usedclientlocs{\methodset}{\prog}}
\asteplab{\persistlab}{\PSC}
\memstate'_\lcl\rst{\usedclientlocs{\methodset}{\prog}}$
holds trivially. Note that by \cref{lem:merge}\ref{item:merge3},
$\memstate'\rst{\unusedlocs{\lib}} = \memstate\rst{\unusedlocs{\lib}} =
\memstate_\lcl\rst{\usedclientlocs{\methodset}{\prog}}$. Therefore, we get
$\memstate'\rst{\unusedlocs{\lib}}
\asteplab{\persistlab}{\PSC}
\memstate\rst{\unusedlocs{\lib}}$.
Since $\tup{\progstate',
\memstate'}$ is reachable, by \cref{prop:prog_lib_separates},
$\usedlocs{\lib} \cap \NVLoc$ separates $\memstate'$. Then, by
\cref{lem:memory_disjointness}(\ref{item:persist}), we obtain $\memstate'
\asteplab{\persistlab}{\PSC} \memstate$. The transition is memory-internal, so we get
$\tup{\progstate',\memstate'}
{\asteplab{\persistlab}{\cs{\client{\prog}{\lib}}{\PSC}}}
\tup{\progstate,\memstate}$,

By the induction hypothesis $\scLEM(\tr_\lcl,
\tr'_\llib)$, for $\tup{\progstate_\lcl,\memstate_\lcl}$,
$\tup{\progstate'_\llib,\memstate'_\llib}$ there there exists $\tr'$ such that
$\fhistory{\tr'} =
\fhistory{\tr_\lcl}$ and
$\tup{\progstate_\Init,\memstate_\Init}
{\asteplab{\tr'}{\cs{\client{\prog}{\lib}}{\PSC}}}
\tup{\progstate',\memstate'}$. It is easy to see that
$\fhistory{\tr' \cdot \persistlab}
= \fhistory{\tr_\lcl}$,
and we have shown that:
\[
\tup{\progstate_\Init,\memstate_\Init}
{\asteplab{\tr'}{\cs{\client{\prog}{\lib}}{\PSC}}}
\tup{\progstate',\memstate'}
{\asteplab{\persistlab}{\cs{\client{\prog}{\lib}}{\PSC}}}
\tup{\progstate,\memstate}
\]
To conclude the proof for this case, we let $\tr = \tr' \cdot \persistlab$.

{Suppose $\alpha_\llib = \tidlab{\tid}{\epsilon}$.}
The transition is program-internal,
so $\memstate'_\llib = \memstate_\llib$. Therefore, by construction,
$\memstate' = \memstate$. The histories of $\tr_\lcl$, $\tr'_\llib$ and
$\tr_\llib$ coincide, so there are two possibilities: either the execution in
all of them is outside of a method of $\lib$ (and then
$\progstate_\lcl(\tid).\lmethod, \progstate'_\llib(\tid).\lmethod,
\progstate_\llib(\tid).\lmethod \notin \methodset$ holds) or inside of the
same method (and then $\progstate_\lcl(\tid).\lmethod,
\progstate'_\llib(\tid).\lmethod,
\progstate_\llib(\tid).\lmethod \in \methodset$ holds).
We first assume that $\progstate_\lcl(\tid).\lmethod,
\progstate'_\llib(\tid).\lmethod,
\progstate_\llib(\tid).\lmethod \nin \methodset$.
Then $\progstate(\tid) =
\progstate'(\tid) = \progstate_\lcl(\tid)$ holds. We let $\tr =
\tr'$; then the induction step immediately follows from the induction
hypothesis $\scLEM(\tr_\lcl, \tr'_\llib)$.
We now consider the possibility of $\progstate_\lcl(\tid).\lmethod,
\progstate'_\llib(\tid).\lmethod,
\progstate_\llib(\tid).\lmethod \in \methodset$.
From \cref{eq:comp1} we deduce that $\tup{\progstate'_\llib(\tid).\lpc,
\progstate'_\llib(\tid).\lphi}
\asteplab{\epsilon}{\lib(\progstate_\llib(\tid).\lmethod)}
\tup{\progstate_\llib(\tid).\lpc, \progstate_\llib(\tid).\lphi}$ holds. By
construction of $\progstate$ and $\progstate'$, and the concurrent program
transition rules we obtain that $\progstate' \asteptidlab{\tid}{\epsilon}{\prog}
\progstate$ holds. Since the latter is a program-internal transition, we get
$\tup{\progstate',\memstate'}
{\asteptidlab{\tid}{\epsilon}{\cs{\client{\prog}{\lib}}{\PSC}}}
\tup{\progstate,\memstate}$.

By the induction hypothesis $\scLEM(\tr_\lcl,
\tr'_\llib)$, for $\tup{\progstate_\lcl,\memstate_\lcl}$,
$\tup{\progstate'_\llib,\memstate'_\llib}$ there there exists $\tr'$ such that
$\fhistory{\tr'} = \fhistory{\tr_\lcl}$ and
$\tup{\progstate_\Init,\memstate_\Init}
{\asteplab{\tr'}{\cs{\client{\prog}{\lib}}{\PSC}}}
\tup{\progstate',\memstate'}$. It is easy to see that
$\fhistory{\tr' \cdot \tidlab{\tid}{\epsilon}}
= \fhistory{\tr_\lcl}$,
and we have shown that:
\[
\tup{\progstate_\Init,\memstate_\Init}
{\asteplab{\tr'}{\cs{\client{\prog}{\lib}}{\PSC}}}
\tup{\progstate',\memstate'}
{\asteptidlab{\tid}{\epsilon}{\cs{\client{\prog}{\lib}}{\PSC}}}
\tup{\progstate,\memstate}
\]
To conclude the proof for this case, we let $\tr = \tr' \cdot \tidlab{\tid}{\epsilon}$.

{Suppose $\alpha_\llib = \tidlab{\tid}{\calllab{\method}{\phi}}$ (or
$\alpha_\llib = \tidlab{\tid}{\retlab{\method}{\phi}}$) and $\method \nin
\methodset$.} The transition is program-internal, so $\memstate'_\llib =
\memstate_\llib$. Therefore, by construction, $\memstate' = \memstate$. It is
also a call (or return) transition, so necessarily
$\progstate'_\llib(\tid).\lmethod = \main \notin \methodset$ and
$\progstate_\llib(\tid).\lmethod = \method \notin \methodset$ (or
$\progstate'_\llib(\tid).\lmethod = \method \notin \methodset$ and
$\progstate_\llib(\tid).\lmethod = \main \notin \methodset$). {By the
premise of the induction step, $\rsthistory{\methodset}{\tr_\lcl} =
\rsthistory{\methodset}{\tr_\llib}$}, so it can only be that
$\progstate_\lcl(\tid).\lmethod \notin \methodset$ holds. Then, by
construction, $\progstate(\tid) = \progstate'(\tid) = \progstate_\lcl(\tid)$
holds. We let $\tr = \tr'$; then the induction step immediately follows from
the induction hypothesis $\scLEM(\tr_\lcl, \tr'_\llib)$.

{Suppose $\alpha_\llib = \tidlab{\tid}{\lab}$ and $\lTYP(\lab)\nin\set{\lCALL,\lRET} \land \lab \nin \SFLab$.}

{By the
premise of the induction, $\rsthistory{\methodset}{\tr_\lcl} =
\rsthistory{\methodset}{\tr_\llib}$}. Also, $\rsthistory{\methodset}{\tr_\llib} =
\rsthistory{\methodset}{\tr'_\llib}$. Hence, either the execution in
all of the traces is outside of a method of $\lib$ (and then
$\progstate_\lcl(\tid).\lmethod, \progstate'_\llib(\tid).\lmethod,
\progstate_\llib(\tid).\lmethod \notin \methodset$ holds) or inside of the
same method (and then $\progstate_\lcl(\tid).\lmethod,
\progstate'_\llib(\tid).\lmethod,
\progstate_\llib(\tid).\lmethod \in \methodset$ holds).
We first assume that $\progstate_\lcl(\tid).\lmethod,
\progstate'_\llib(\tid).\lmethod,
\progstate_\llib(\tid).\lmethod \nin \methodset$.
Then $\progstate = \progstate' = \progstate_\lcl$ holds. By \cref{prop:local},
$\lLOCSET(\lab) \suq \unusedlocs{\lib}$. From \cref{eq:comp1} we deduce
that $\memstate'_\llib
\asteptidlab{\tid}{\lab}{\PSC}
\memstate_\llib$ holds.
Since $\tup{\progstate'_\llib, \memstate'_\llib}$ is reachable, by
\cref{prop:prog_lib_separates}, we have that $\usedlocs{\lib}$ separates
$\memstate'$.
By \cref{lem:memory_disjointness}(\ref{item:other}),
$\memstate'_\llib\rst{\unusedlocs{\lib}}
\asteptidlab{\tid}{\lab}{\PSC}
\memstate_\llib\rst{\unusedlocs{\lib}}$ and
$\memstate'_\llib\rst{\usedlocs{\lib}} =
\memstate_\llib\rst{\usedlocs{\lib}}$. The latter in particular implies that,
by construction, $\memstate' = \memstate$.
We let $\tr = \tr'$; then the induction step immediately follows from the induction hypothesis $\scLEM(\tr_\lcl, \tr'_\llib)$.

We now consider the possibility of $\progstate_\lcl(\tid).\lmethod, \progstate'_\llib(\tid).\lmethod,
\progstate_\llib(\tid).\lmethod \in \methodset$.
From \cref{eq:comp1} we deduce that $\tup{\progstate'_\llib(\tid).\lpc,
\progstate'_\llib(\tid).\lphi}
\asteplab{\lab}{\lib(\progstate_\llib(\tid).\lmethod)}
\tup{\progstate_\llib(\tid).\lpc, \progstate_\llib(\tid).\lphi}$ holds. By
construction of $\progstate$ and $\progstate'$, and the concurrent program
transition rules we obtain that $\progstate' \asteptidlab{\tid}{\lab}{\prog}
\progstate$ holds.
Secondly, by $\scLEM(\tr_\lcl,
\tr'_\llib)$,
 $\memstate'=\nmergemem{\memstate_\lcl}{\usedclientlocs{\methodset}{\prog}}{\memstate'_\llib}{\usedlocs{\lib}}$. By \cref{prop:local},
$\lLOCSET(\lab) \suq \unusedlocs{\lib}$.
We deduce from \cref{eq:comp1} $\memstate'_\llib
\asteptidlab{\tid}{\lab}{\PSC}
\memstate_\llib$. By \cref{prop:prog_lib_separates}, since $\lib$ is safe for $\prog'$, $\usedlocs{\lib} \cap \NVLoc$ separates $\memstate'_\llib$, which allows us to deduce $\memstate'_\llib\rst{\usedlocs{\lib}}
\asteptidlab{\tid}{\lab}{\PSC}
\memstate_\llib\rst{\usedlocs{\lib}}$
by \cref{lem:memory_disjointness}(\ref{item:persist}). From the induction
hypothesis $\scLEM(\tr_\lcl, \tr'_\llib)$, we know that
$\memstate'=\nmergemem{\memstate_\lcl}{\usedclientlocs{\methodset}{\prog}}{\memstate'_\llib}{\usedlocs{\lib}}$,
so overall we obtain $\memstate'\rst{\usedlocs{\lib}}
\asteptidlab{\tid}{\lab}{\PSC} \memstate\rst{\usedlocs{\lib}}$.
Also, note that by \cref{lem:merge}\ref{item:merge3} and by contruction of the
merges, $\memstate'\rst{\unusedlocs{\lib}} = \memstate\rst{\unusedlocs{\lib}}
= \memstate_\lcl\rst{\usedclientlocs{\methodset}{\prog}}$. Since
$\tup{\progstate', \memstate'}$ is reachable, by
\cref{prop:prog_lib_separates}, we have that $\usedlocs{\lib}$ separates
$\memstate'$. Then, by \cref{lem:memory_disjointness}(\ref{item:other}), we
have $\memstate'
\asteptidlab{\tid}{\lab}{\PSC} \memstate$. By synchronizing the latter with
$\progstate' \asteptidlab{\tid}{\lab}{\prog}
\progstate$, we get:
$\tup{\progstate',\memstate'}
{\asteptidlab{\tid}{\lab}{\cs{\client{\prog}{\lib}}{\PSC}}}
\tup{\progstate,\memstate}$.

By the induction hypothesis $\scLEM(\tr_\lcl,
\tr'_\llib)$, for $\tup{\progstate_\lcl,\memstate_\lcl}$,
$\tup{\progstate'_\llib,\memstate'_\llib}$ there there exists $\tr'$ such that
$\fhistory{\tr'} =
\fhistory{\tr_\lcl}$ and
$\tup{\progstate_\Init,\memstate_\Init}
{\asteplab{\tr'}{\cs{\client{\prog}{\lib}}{\PSC}}}
\tup{\progstate',\memstate'}$. It is easy to see that
$\fhistory{\tr' \cdot \tidlab{\tid}{\lab}}
= \fhistory{\tr_\lcl}$,
and we have shown that:
\[
\tup{\progstate_\Init,\memstate_\Init}
{\asteplab{\tr'}{\cs{\client{\prog}{\lib}}{\PSC}}}
\tup{\progstate',\memstate'}
{\asteptidlab{\tid}{\lab}{\cs{\client{\prog}{\lib}}{\PSC}}}
\tup{\progstate,\memstate}
\]
To conclude the proof for this case, we let $\tr = \tr' \cdot \tidlab{\tid}{\lab}$.

\noindent\underline{\textsc{Case II.}}
Suppose that $\tr_\lcl$ is non-empty and ends with a label $\alpha_\lcl$ not
corresponding to a history label. Let $\tr_\lcl = \tr'_\lcl \cdot
\alpha_\lcl$, and consider any state
$\tup{\progstate'_\lcl,\memstate'_\lcl}$ for which there are the following
transitions:
\begin{equation}\label{eq:comp2}\tag{II}
\tup{\progstate_\Init,\memstate_\Init}
{\asteplab{\tr'_\lcl}{\cs{\client{\prog}{\lib'}}{\PSC}}}
\tup{\progstate'_\lcl,\memstate'_\lcl}
\asteplab{\alpha_\lcl}{\cs{\client{\prog}{\lib'}}{\PSC}}
\tup{\progstate_\lcl,\memstate_\lcl}
\end{equation}
Like in Case I, we consider separately various cases for $\alpha_\lcl$ in
order to construct $\tr$. We only give a proof for the case of call and return labels here, since the other cases are analogous to Case I.

{Suppose $\alpha_\lcl = \tidlab{\tid}{\calllab{\method}{\phi}}$
(or $\alpha_\lcl = \tidlab{\tid}{\retlab{\method}{\phi}}$) and $\method \nin \methodset$.}
The transition is program-internal,
so $\memstate'_\lcl = \memstate_\lcl$. It is also a call transition into a method not in $\methodset$, so
$\progstate'_\lcl(\tid).\lmethod = \main \notin \methodset$ and
$\progstate'_\lcl(\tid).\lmethod \notin \methodset$.
Then, by construction, $\progstate'(\tid) = \progstate'_\lcl$ and
$\progstate(\tid) = \progstate_\lcl(\tid)$.
We deduce from \cref{eq:comp2} that
$\progstate'_\lcl
{\asteptidlab{\tid}{\calllab{\method}{\phi}}{\prog}}
\progstate_\lcl$ and, therefore,
$\progstate'
{\asteptidlab{\tid}{\calllab{\method}{\phi}}{\prog}}
\progstate$.
Since the transition is program-internal,  we have
$\tup{\progstate',\memstate'}
{\asteptidlab{\tid}{\calllab{\method}{\phi}}{\cs{\client{\prog}{\lib}}{\PSC}}}
\tup{\progstate,\memstate}$.

By the induction hypothesis $\scLEM(\tr'_\lcl,
\tr_\llib)$, for $\tup{\progstate'_\lcl,\memstate'_\lcl}$,
$\tup{\progstate_\llib,\memstate_\llib}$ there exists $\tr'$ such that
$\fhistory{\tr'} = \fhistory{\tr'_\lcl}$ and
$\tup{\progstate_\Init,\memstate_\Init}
{\asteplab{\tr'}{\cs{\client{\prog}{\lib}}{\PSC}}}
\tup{\progstate',\memstate'}$. It is easy to see that
$\fhistory{\tr' \cdot \tidlab{\tid}{\calllab{\method}{\phi}}}
= \fhistory{\tr'_\lcl \cdot \tidlab{\tid}{\calllab{\method}{\phi}}}
= \fhistory{\tr_\lcl}$,
and we have shown that:
\[
\tup{\progstate_\Init,\memstate_\Init}
{\asteplab{\tr'}{\cs{\client{\prog}{\lib}}{\PSC}}}
\tup{\progstate',\memstate'}
{\asteptidlab{\tid}{\calllab{\method}{\phi}}{\cs{\client{\prog}{\lib}}{\PSC}}}
\tup{\progstate,\memstate}
\]
To conclude the proof for this case, we let $\tr = \tr' \cdot \tidlab{\tid}{\calllab{\method}{\phi}}$.

{Suppose $\alpha_\lcl = \tidlab{\tid}{\retlab{\method}{\phi}}$ and $\method \nin \methodset$.}
The transition is program-internal,
so $\memstate'_\lcl = \memstate_\lcl$. It is also a return transition from a method not in $\methodset$, so
$\progstate'_\lcl(\tid).\lmethod \notin \methodset$ and
$\progstate'_\lcl(\tid).\lmethod = \main \notin \methodset$.
Then, by construction, $\progstate'(\tid) = \progstate'_\lcl$ and
$\progstate(\tid) = \progstate_\lcl(\tid)$.
We deduce from \cref{eq:comp2} that
$\progstate'_\lcl
{\asteptidlab{\tid}{\retlab{\method}{\phi}}{\prog}}
\progstate_\lcl$ and, therefore,
$\progstate'
{\asteptidlab{\tid}{\retlab{\method}{\phi}}{\prog}}
\progstate$.
Since the transition is program-internal,
$\tup{\progstate',\memstate'}
{\asteptidlab{\tid}{\retlab{\method}{\phi}}{\cs{\client{\prog}{\lib}}{\PSC}}}
\tup{\progstate,\memstate}$.

By the induction hypothesis $\scLEM(\tr'_\lcl,
\tr_\llib)$, for $\tup{\progstate'_\lcl,\memstate'_\lcl}$,
$\tup{\progstate_\llib,\memstate_\llib}$ there exists $\tr'$ such that
$\fhistory{\tr'} = \fhistory{\tr'_\lcl}$ and
$\tup{\progstate_\Init,\memstate_\Init}
{\asteplab{\tr'}{\cs{\client{\prog}{\lib}}{\PSC}}}
\tup{\progstate',\memstate'}$. It is easy to see that
$\fhistory{\tr' \cdot \tidlab{\tid}{\retlab{\method}{\phi}}}
= \fhistory{\tr'_\lcl \cdot \tidlab{\tid}{\retlab{\method}{\phi}}}
= \fhistory{\tr_\lcl}$,
and we have shown that:
\[
\tup{\progstate_\Init,\memstate_\Init}
{\asteplab{\tr'}{\cs{\client{\prog}{\lib}}{\PSC}}}
\tup{\progstate',\memstate'}
{\asteptidlab{\tid}{\retlab{\method}{\phi}}{\cs{\client{\prog}{\lib}}{\PSC}}}
\tup{\progstate,\memstate}
\]
To conclude the proof for this case, we let $\tr = \tr' \cdot \tidlab{\tid}{\retlab{\method}{\phi}}$.

\noindent\underline{\textsc{Case III.}} Suppose both $\tr_\lcl$ and $\tr_\llib$ are
non-empty and end with a label corresponding to a history label. Let $\tr_\lcl
= \tr'_\lcl \cdot \alpha_\lcl$ and $\tr_\llib = \tr'_\llib
\cdot \alpha_\llib$. By the premise of the induction step,
$\rsthistory{\methodset}{\tr_\lcl} = \rsthistory{\methodset}{\tr_\llib}$
holds; hence, $\rsthistory{\methodset}{\alpha_\lcl} =
\rsthistory{\methodset}{\alpha_\llib}$, and we refer to that history action
label as $\alpha$.
Let $\tup{\progstate'_\llib,\memstate'_\llib}$ and
$\tup{\progstate'_\lcl,\memstate'_\lcl}$ be any states for which there are the
following transitions:
\begin{equation}\label{eq:comp3}\tag{III}
\begin{array}{c}
\tup{\progstate_\Init,\memstate_\Init}
{\asteplab{\tr'_\llib}{\cs{\client{\prog'}{\lib}}{\PSC}}}
\tup{\progstate'_\llib,\memstate'_\llib}
\asteplab{\alpha_\llib}{\cs{\client{\prog'}{\lib}}{\PSC}}
\tup{\progstate_\llib,\memstate_\llib}
\\
\tup{\progstate_\Init,\memstate_\Init}
{\asteplab{\tr'_\lcl}{\cs{\client{\prog}{\lib'}}{\PSC}}}
\tup{\progstate'_\lcl,\memstate'_\lcl}
\asteplab{\alpha_\lcl}{\cs{\client{\prog}{\lib'}}{\PSC}}
\tup{\progstate_\lcl,\memstate_\lcl}
\end{array}
\end{equation}
In the following, we consider different combinations of $\alpha_\lcl$ and
$\alpha_\llib$ in order to construct $\tr$.

{Suppose $\alpha = \alpha_\lcl = \alpha_\llib = \tidlab{\tid}{\calllab{\method}{\phi}}$.}
Let $\progstate'_\lcl(\tid).\lpc = \pc$. The
transition is program-internal, so $\memstate'_\llib = \memstate_\llib$ and
$\memstate'_\lcl = \memstate_\lcl$; therefore, by construction, $\memstate' =
\memstate$. It is a call transition, so
$\prog(\tid)(\main)(\pc) = \callInst{\method}$,
$\progstate'_\lcl(\tid).\lphi = \phi$,
$\progstate'_\lcl(\tid).\lpcs = \bot$ and
$\progstate'_\lcl(\tid).\lmethod =
\main$, and also
$\progstate_\llib(\tid).\lpc = 0$ and
$\progstate_\llib(\tid).\lphi = \phi$,
$\progstate_\lcl(\tid).\lpcs = \pc+1$ and
$\progstate_\lcl(\tid).\lmethod = \method$.
By construction, $\progstate'(\tid) = \tup{\pc, \phi, \bot, \main}$ and
$\progstate(\tid) = \tup{0, \phi, \pc+1, \method}$.
Then
$\progstate'(\tid)
{\asteplab{\calllab{\method}{\phi}}{\prog(\tid)}}
\progstate(\tid)$ and, therefore,
$\progstate'
{\asteptidlab{\tid}{\calllab{\method}{\phi}}{\prog}}
\progstate$.
Since the transition is program-internal,
we have
$\tup{\progstate',\memstate'}
{\asteptidlab{\tid}{\calllab{\method}{\phi}}{\cs{\client{\prog}{\lib}}{\PSC}}}
\tup{\progstate,\memstate}$.

By the induction hypothesis $\scLEM(\tr'_\lcl,
\tr'_\llib)$, for $\tup{\progstate'_\lcl,\memstate'_\lcl}$,
$\tup{\progstate'_\llib,\memstate'_\llib}$ there there exists $\tr'$ such that
$\fhistory{\tr'} =
\fhistory{\tr'_\lcl}$ and
$\tup{\progstate_\Init,\memstate_\Init}
{\asteplab{\tr'}{\cs{\client{\prog}{\lib}}{\PSC}}}
\tup{\progstate',\memstate'}$. It is easy to see that
$\fhistory{\tr' \cdot \tidlab{\tid}{\calllab{\method}{\phi}}}
= \fhistory{\tr'_\lcl \cdot \tidlab{\tid}{\calllab{\method}{\phi}}}
= \fhistory{\tr_\lcl}$,
and we have shown that:
\[
\tup{\progstate_\Init,\memstate_\Init}
{\asteplab{\tr'}{\cs{\client{\prog}{\lib}}{\PSC}}}
\tup{\progstate',\memstate'}
{\asteptidlab{\tid}{\calllab{\method}{\phi}}{\cs{\client{\prog}{\lib}}{\PSC}}}
\tup{\progstate,\memstate}
\]
To conclude the proof for this case, we let $\tr = \tr' \cdot \tidlab{\tid}{\calllab{\method}{\phi}}$.

{Suppose $\alpha = \alpha_\lcl = \alpha_\llib = \tidlab{\tid}{\retlab{\method}{\phi}}$.}
Let $\progstate'_\llib(\tid).\lpc = \pc_\llib$ and $\progstate'_\lcl(\tid).\lpcs = \pcs$. The transition is program-internal,
so $\memstate'_\llib = \memstate_\llib$ and $\memstate'_\lcl = \memstate_\lcl$; therefore, by construction, $\memstate' = \memstate$.
It is a return transition, so
$\prog(\tid)(\method)(\pc_\llib) = \returnInst$,
$\progstate'_\llib(\tid).\lphi = \phi$ and
$\progstate'_\lcl(\tid).\lmethod = \method$,
and also
$\progstate_\lcl(\tid).\lpc = \pcs$ and
$\progstate_\lcl(\tid).\lphi = \phi$,
$\progstate_\lcl(\tid).\lpcs = \bot$ and
$\progstate_\lcl(\tid).\lmethod = \main$.
By construction, $\progstate'(\tid) = \tup{\pc_\llib, \phi, \pcs, \method}$ and $\progstate(\tid) = \tup{
\pcs,
\phi,
\bot,
\main}$. Then
$\progstate'(\tid)
{\asteplab{\retlab{\method}{\phi}}{\prog(\tid)}}
\progstate(\tid)$, therefore,
$\progstate'
{\asteptidlab{\tid}{\retlab{\method}{\phi}}{\prog}}
\progstate$.
Since the transition is program-internal,
$\tup{\progstate',\memstate'}
{\asteptidlab{\tid}{\retlab{\method}{\phi}}{\cs{\client{\prog}{\lib}}{\PSC}}}
\tup{\progstate,\memstate}$.

By the induction hypothesis $\scLEM(\tr'_\lcl,
\tr'_\llib)$, for $\tup{\progstate'_\lcl,\memstate'_\lcl}$,
$\tup{\progstate'_\llib,\memstate'_\llib}$ there there exists $\tr'$ such that
$\fhistory{\tr'} =
\fhistory{\tr'_\lcl} = \rsthistory{\methodset}{\tr'_\llib}$ and
$\tup{\progstate_\Init,\memstate_\Init}
{\asteplab{\tr'}{\cs{\client{\prog}{\lib}}{\PSC}}}
\tup{\progstate',\memstate'}$. It is easy to see that
$\fhistory{\tr' \cdot \tidlab{\tid}{\retlab{\method}{\phi}}}
= \fhistory{\tr'_\lcl \cdot \tidlab{\tid}{\retlab{\method}{\phi}}}
= \fhistory{\tr_\lcl}$,
and we have shown that:
\[
\tup{\progstate_\Init,\memstate_\Init}
{\asteplab{\tr'}{\cs{\client{\prog}{\lib}}{\PSC}}}
\tup{\progstate',\memstate'}
{\asteptidlab{\tid}{\retlab{\method}{\phi}}{\cs{\client{\prog}{\lib}}{\PSC}}}
\tup{\progstate,\memstate}
\]
To conclude the proof for this case, we let $\tr = \tr' \cdot \tidlab{\tid}{\retlab{\method}{\phi}}$.

{Suppose $\alpha_\lcl = \alpha_\llib = \alpha = \crash$.}
From \cref{eq:comp3} we deduce that $\progstate_\lcl = \progstate_\llib =
\progstate_\Init$. By construction, $\progstate = \progstate_\Init$. We also deduce that $\memstate'_\lcl \asteplab{\crash}{\PSC}
\memstate_\lcl$ and $\memstate'_\llib \asteplab{\crash}{\PSC}
\memstate_\llib$ hold.
We consider $\memstate'_\lcl \asteplab{\crash}{\PSC}
\memstate_\lcl$ first.
By \cref{prop:prog_lib_separates}, since $\lib'$ is safe for $\prog$, $\usedclientlocs{\methodset}{\prog} \cap \NVLoc$ separates $\memstate'_\lcl$, which allows us to deduce $\memstate'_\lcl\rst{\usedclientlocs{\methodset}{\prog}} \asteplab{\crash}{\PSC}
\memstate_\lcl\rst{\usedclientlocs{\methodset}{\prog}}$ by
\cref{lem:memory_disjointness}(\ref{item:crash}).
From the induction
hypothesis $\scLEM(\tr'_\lcl, \tr'_\llib)$, we know that
$\memstate'=\nmergemem{\memstate'_\lcl}{\usedclientlocs{\methodset}{\prog}}
{\memstate'_\llib}{\usedlocs{\lib}}$
Note that by \cref{lem:merge}\ref{item:merge3},
$\memstate'\rst{\unusedlocs{\lib}} = \memstate\rst{\unusedlocs{\lib}} =
\memstate_\lcl\rst{\usedclientlocs{\methodset}{\prog}}$. Therefore, we get
$\memstate'\rst{\unusedlocs{\lib}}
\asteplab{\persistlab}{\PSC}
\memstate\rst{\unusedlocs{\lib}}$.
We consider $\memstate'_\llib \asteplab{\crash}{\PSC}
\memstate_\llib$ now. By \cref{prop:prog_lib_separates}, since $\lib$ is safe for $\prog'$, $\usedlocs{\lib} \cap \NVLoc$ separates $\memstate'_\llib$, which allows us to deduce $\memstate'_\llib\rst{\usedlocs{\lib}}
\asteplab{\crash}{\PSC}
\memstate_\llib\rst{\usedlocs{\lib}}$
by \cref{lem:memory_disjointness}(\ref{item:crash}). From the induction
hypothesis $\scLEM(\tr'_\lcl, \tr'_\llib)$, we know that
$\memstate'=\nmergemem{\memstate'_\lcl}{\usedclientlocs{\methodset}{\prog}}
{\memstate'_\llib}{\usedlocs{\lib}}$,
so we obtain $\memstate'\rst{\usedlocs{\lib}}
\asteplab{\crash}{\PSC} \memstate\rst{\usedlocs{\lib}}$.
Overall, we have $\memstate'\rst{\unusedlocs{\prog}}
\asteplab{\crash}{\PSC}
\memstate\rst{\unusedlocs{\prog}}$ and $\memstate'\rst{\usedlocs{\lib}}
\asteplab{\crash}{\PSC}
\memstate\rst{\usedlocs{\lib}}$ hold.
By \cref{lem:memory_disjointness}(\ref{item:crash}), $\memstate'
\asteplab{\crash}{\PSC} \memstate$, which gives us $\tup{\progstate',\memstate'}
{\asteplab{\crash}{\cs{\client{\prog}{\lib}}{\PSC}}}
\tup{\progstate,\memstate}$

By the induction hypothesis $\scLEM(\tr'_\lcl,
\tr'_\llib)$, for $\tup{\progstate'_\lcl,\memstate'_\lcl}$,
$\tup{\progstate'_\llib,\memstate'_\llib}$ there there exists $\tr'$ such that
$\fhistory{\tr'} =
\fhistory{\tr'_\lcl}$ and
$\tup{\progstate_\Init,\memstate_\Init}
{\asteplab{\tr'}{\cs{\client{\prog}{\lib}}{\PSC}}}
\tup{\progstate',\memstate'}$. It is easy to see that
$\fhistory{\tr' \cdot \crash}
= \fhistory{\tr'_\lcl \cdot \crash}
= \fhistory{\tr_\lcl}$,
and we have shown that:
\[
\tup{\progstate_\Init,\memstate_\Init}
{\asteplab{\tr'}{\cs{\client{\prog}{\lib}}{\PSC}}}
\tup{\progstate',\memstate'}
{\asteplab{\crash}{\cs{\client{\prog}{\lib}}{\PSC}}}
\tup{\progstate,\memstate}
\]
To conclude the proof for this case, we let $\tr = \tr' \cdot \crash$.

{Suppose $\alpha_\lcl = \tidlab{\tid}{\lab_\lcl}$,
$\alpha_\llib = \tidlab{\tid}{\lab_\llib}$, and $\lab_\lcl,\lab_\llib \in \SFLab$.}
Let us assume that $\progstate'(\tid).\lmethod \in \methodset$ (the case when $\progstate'(\tid).\lmethod \nin \methodset$ is analogous).
We deduce from \cref{eq:comp3} that $\tup{\progstate'_\llib(\tid).\lpc,
\progstate'_\llib(\tid).\lphi}
\asteplab{\lab_\llib}{\lib(\progstate'_\llib(\tid).\lmethod)}
\tup{\progstate_\llib(\tid).\lpc, \progstate_\llib(\tid).\lphi}$ holds.
By
construction of $\progstate$ and $\progstate'$, and the concurrent program
transition rules we obtain that $\progstate' \asteplab{\alpha_\llib}{\prog}
\progstate$ holds. We also deduce from \cref{eq:comp3} that
$\memstate'_\lcl \asteplab{\alpha_\lcl}{\PSC}
\memstate_\lcl$ and $\memstate'_\llib \asteplab{\alpha_\llib}{\PSC}
\memstate_\llib$.
We first consider
$\memstate'_\lcl \asteplab{\alpha_\lcl}{\PSC}
\memstate_\lcl$.
By \cref{prop:local},
$\lLOCSET(\alpha_\lcl) \suq \usedlocs{\lib'}$.
Since $\lib'$ is safe for $\prog$, firstly,
$\lLOCSET(\alpha_\lcl) \suq \Loc\setminus\usedclientlocs{\methodset}{\prog}$, and secondly,
by \cref{prop:prog_lib_separates,prop:separates_complement}, $(\Loc\setminus\usedclientlocs{\methodset}{\prog}) \cap \NVLoc$ separates $\memstate'_\lcl$, which allows us to deduce $\memstate'_\lcl\rst{\usedclientlocs{\methodset}{\prog}} \asteptidlab{\tid}{\lSF}{\PSC}
\memstate_\lcl\rst{\usedclientlocs{\methodset}{\prog}}$ by
\cref{lem:memory_disjointness}(\ref{item:sflab}).
From the induction
hypothesis $\scLEM(\tr'_\lcl, \tr'_\llib)$, we know that
$\memstate'=\nmergemem{\memstate'_\lcl}{\usedclientlocs{\methodset}{\prog}}
{\memstate'_\llib}{\usedlocs{\lib}}$.
Note that by \cref{lem:merge}\ref{item:merge3},
$\memstate'\rst{\unusedlocs{\lib}} = \memstate\rst{\unusedlocs{\lib}} =
\memstate_\lcl\rst{\usedclientlocs{\methodset}{\prog}}$. Therefore, we get
$\memstate'\rst{\unusedlocs{\lib}}
\asteptidlab{\tid}{\lSF}{\PSC}
\memstate\rst{\unusedlocs{\lib}}$.
We now consider
$\memstate'_\llib \asteplab{\alpha_\llib}{\PSC}
\memstate_\llib$.
By \cref{prop:prog_lib_separates}, since $\lib$ is safe for $\prog'$, $\usedlocs{\lib} \cap \NVLoc$ separates $\memstate'_\llib$, which allows us to deduce $\memstate'_\llib\rst{\usedlocs{\lib}}
\asteplab{\alpha_\llib}{\PSC}
\memstate_\llib\rst{\usedlocs{\lib}}$
by \cref{lem:memory_disjointness}(\ref{item:sflab}). From the induction
hypothesis $\scLEM(\tr'_\lcl, \tr'_\llib)$, we know that
$\memstate'=\nmergemem{\memstate'_\lcl}{\usedclientlocs{\methodset}{\prog}}
{\memstate'_\llib}{\usedlocs{\lib}}$,
so we obtain $\memstate'\rst{\usedlocs{\lib}}
\asteplab{\alpha_\llib}{\PSC} \memstate\rst{\usedlocs{\lib}}$.
Overall,
$\memstate'\rst{\unusedlocs{\lib}}
\asteptidlab{\tid}{\lSF}{\PSC}
\memstate\rst{\unusedlocs{\lib}}$
and $\memstate'\rst{\usedlocs{\lib}}
\asteplab{\alpha_\llib}{\PSC}
\memstate\rst{\usedlocs{\lib}}$. Since $\tup{\progstate', \memstate'}$ is
reachable, by \cref{prop:prog_lib_separates}, we have that $\usedlocs{\lib}$
separates $\memstate'$. By
\cref{lem:memory_disjointness} (\ref{item:sflab}), $\memstate'
\asteplab{\alpha}{\PSC} \memstate$, which gives us
$\tup{\progstate',\memstate'}
{\asteplab{\alpha_\llib}{\cs{\client{\prog}{\lib}}{\PSC}}}
\tup{\progstate,\memstate}$.

By the induction hypothesis $\scLEM(\tr'_\lcl,
\tr'_\llib)$, for $\tup{\progstate'_\lcl,\memstate'_\lcl}$,
$\tup{\progstate'_\llib,\memstate'_\llib}$ there there exists $\tr'$ such that
$\fhistory{\tr'} =
\fhistory{\tr'_\lcl}$ and
$\tup{\progstate_\Init,\memstate_\Init}
{\asteplab{\tr'}{\cs{\client{\prog}{\lib}}{\PSC}}}
\tup{\progstate',\memstate'}$. Recalling that $\rsthistory{\methodset}{\alpha_\lcl} =
\rsthistory{\methodset}{\alpha_\llib}$, it is easy to see that
$\fhistory{\tr' \cdot \alpha_\llib}
= \fhistory{\tr'_\lcl \cdot \alpha_\lcl}
= \fhistory{\tr_\lcl}$,
and we have shown that:
\[
\tup{\progstate_\Init,\memstate_\Init}
{\asteplab{\tr'}{\cs{\client{\prog}{\lib}}{\PSC}}}
\tup{\progstate',\memstate'}
{\asteplab{\alpha_\llib}{\cs{\client{\prog}{\lib}}{\PSC}}}
\tup{\progstate,\memstate}
\]
To conclude the proof for this case, we let $\tr = \tr' \cdot \alpha_\llib$.
\end{proof}

\abstraction*

\begin{proof}
Let $\methodset=\dom{\lib}$.
It suffices to show $\fhistory{\client{\prog}{\lib}} \suq \fhistory{\client{\prog}{\lib^\sharp}}$.
Then, the claim follows using \cref{lem:composition} (applied with $\lib:=\lib^\sharp$, $\lib':=\lib$, $\prog:=\prog$, and $\prog':=\prog$).
Suppose otherwise, and let $\history$ be a shortest history in
$\fhistory{\client{\prog}{\lib}} \setminus \fhistory{\client{\prog}{\lib^\sharp}}$.
Let $\tr$ be a shortest trace in $\traces{\cs{\client{\prog}{\lib}}{\PSC}}$ with $\fhistory{\tr}=\history$.
Let $\tup{\progstate,\memstate}$ such that
$\tup{\progstate_\Init,\memstate_\Init} \asteplab{\tr}{\cs{\client{\prog}{\lib}}{\PSC}} \tup{\progstate,\memstate}$.
Clearly, $\tr$ cannot be empty (since the empty history is a history of any program).
Consider the last transition in $\tr$,
and let $\tr'$, $\alpha$, and $\tup{\progstate',\memstate'}$,
such that $\tr = \tr' \cdot \alpha$
and $\tup{\progstate_\Init,\memstate_\Init} \asteplab{\tr'}{\cs{\client{\prog}{\lib}}{\PSC}} \tup{\progstate',\memstate'}
\asteplab{\alpha}{\cs{\client{\prog}{\lib}}{\PSC}} \tup{\progstate,\memstate}$.
Let $\history'=\fhistory{\tr'}$.
The minimality of $\tr$ ensures that $\history'$ is a proper prefix of $\history$,
and thus $\alpha$ must correspond to a history label.
In turn, using \cref{prop:h_prefix}, the minimality of $\history$ ensures that
$\history' \in \fhistory{\client{\prog}{\lib^\sharp}}$.

Now, if $\alpha=\crash$, then using \cref{prop:h_crash},
we would have $\history = \fhistory{\tr} = \fhistory{\tr' \cdot\alpha} =
\history' \cdot \crash \in \fhistory{\client{\prog}{\lib^\sharp}}$.
Similarly, if $\alpha=\tidlab{\tid}{\lab}$ for some $\tid\in\Tid$ and $\lab\in\SFLab$, then using \cref{prop:h_sf},
we would have $\history = \fhistory{\tr} = \fhistory{\tr' \cdot\alpha} =
\history' \cdot \tidlab{\tid}{\sflab} \in \fhistory{\client{\prog}{\lib^\sharp}}$.

Hence, we have $\alpha = \tidlab{\tid}{\calllab{\method}{\phi}}$ or
$\alpha = \tidlab{\tid}{\retlab{\method}{\phi}}$
for some $\tid\in\Tid$, $\method\in\MethodNames$, and $\phi:\Reg\to\Val$.
It also follows that
$\progstate'\asteplab{\alpha}{\client{\prog}{\lib}} \progstate$
(since $\tup{\progstate',\memstate'}
\asteplab{\alpha}{\cs{\client{\prog}{\lib}}{\PSC}} \tup{\progstate,\memstate}$
and $\alpha\neq \persistlab$).

We claim that $\progstate'(\tid).\lmethod \in \methodset$ (and so, it must be the case that
$\alpha = \tidlab{\tid}{\retlab{\method}{\phi}}$ for $\method\in\methodset$).
Indeed, suppose otherwise.
Let $\tr'_\sharp$ and $\tup{\progstate'_\sharp,\memstate'_\sharp}$ such that
$\fhistory{\tr'_\sharp}=\history'$
and $\tup{\progstate_\Init,\memstate_\Init} \asteplab{\tr'_\sharp}{\cs{\client{\prog}{\lib^\sharp}}{\PSC}} \tup{\progstate'_\sharp,\memstate'_\sharp}$.
Using \cref{lem:composition} (applied with $\lib:=\lib^\sharp$, $\lib':=\lib$, $\prog:=\prog$, and $\prog':=\prog$),
there exist $\tr''_\sharp$
and $\tup{\progstate''_\sharp,\memstate''_\sharp}$
such that
$\fhistory{\tr''_\sharp} = \history'$,
$\tup{\progstate_\Init,\memstate_\Init}
\asteplab{\tr''_\sharp}{\cs{\client{\prog}{\lib^\sharp}}{\PSC}} \tup{\progstate''_\sharp,\memstate''_\sharp}$,
and $\progstate''_\sharp(\tida) =\progstate'(\tida)$
for every $\tida$ such that
$\progstate'(\tida).\lmethod \nin \methodset$.
Now, since $\progstate'\asteplab{\alpha}{\client{\prog}{\lib}} \progstate$
and $\progstate'(\tid).\lmethod \nin \methodset$,
by \cref{prop:P_or_L1}, we have that $\progstate'\asteplab{\alpha}{\client{\prog}{\lib^\sharp}} \progstate$.
In addition, since $\progstate'(\tid).\lmethod \nin \methodset$,
we also have $\progstate''_\sharp(\tid)=\progstate'_\sharp(\tid)$.
Hence, $\alpha$ is enabled in $\progstate''_\sharp$ (in the LTS $\client{\prog}{\lib^\sharp}$),
and so it is also enabled in $\tup{\progstate''_\sharp,\memstate''_\sharp}$
(in the LTS $\cs{\client{\prog}{\lib^\sharp}}{\PSC}$).
It follows that $\tr''_\sharp\cdot\alpha\in \traces{\cs{\client{\prog}{\lib^\sharp}}{\PSC}}$,
but since $\fhistory{\tr''_\sharp\cdot\alpha} = \history$,
this contradicts the fact that
$\history \nin \fhistory{\client{\prog}{\lib^\sharp}}$.

Since $\prog$ correctly calls $\lib^\sharp$ \wrt $\MGCn$, we have
$\rsthistory{\methodset}{\tr'_\sharp} \in \rsthistory{\methodset}{\MGC{\lib^\sharp}}$.
Let $\tr^*_\sharp$ and $\tup{\progstate^*_\sharp,\memstate^*_\sharp}$ such that
$\rsthistory{\methodset}{\tr^*_\sharp}=\rsthistory{\methodset}{\tr'_\sharp}$
and $\tup{\progstate_\Init,\memstate_\Init} \asteplab{\tr^*_\sharp}{\cs{\MGC{\lib^\sharp}}{\PSC}} \tup{\progstate^*_\sharp,\memstate^*_\sharp}$.
Using \cref{lem:composition} (applied with $\lib:=\lib$, $\lib':=\lib^\sharp$, $\prog:=\MGCn$, and $\prog':=\prog$),
there exist $\tr^*$ and $\tup{\progstate^*,\memstate^*}$
such that
$\rsthistory{\methodset}{\tr^*} = \rsthistory{\methodset}{\tr'_\sharp}$,
$\tup{\progstate_\Init,\memstate_\Init}
\asteplab{\tr^*}{\cs{\MGC{\lib}}{\PSC}} \tup{\progstate^*,\memstate^*}$,
and $\progstate^*(\tida) =
\tup{
\progstate'(\tida).\lpc,
\progstate'(\tida).\lphi,
\progstate^*_\sharp(\tida).\lpcs,
\progstate^*_\sharp(\tida).\lmethod}$
for every $\tida$ such that
$\progstate^*_\sharp(\tida).\lmethod \in \methodset$.
Since $\rsthistory{\methodset}{\tr^*} = \rsthistory{\methodset}{\tr'_\sharp}= \rsthistory{\methodset}{\tr'}$
and $\progstate'(\tid).\lmethod \in \methodset$,
by \cref{prop:same_f}, we have  $\progstate^*_\sharp(\tid).\lmethod \in \methodset$,
and so $\progstate^*(\tid)= \tup{
\progstate'(\tid).\lpc,
\progstate'(\tid).\lphi,
\progstate^*_\sharp(\tid).\lpcs,
\progstate^*_\sharp(\tid).\lmethod}$.
Since $\progstate'\asteplab{\alpha}{\client{\prog}{\lib}} \progstate$,
it follows that $\alpha$ is enabled in $\progstate^*$ (in the LTS $\MGC{\lib}$),
and so it is also enabled in $\tup{\progstate^*,\memstate^*}$
(in the LTS $\cs{\MGC{\lib}}{\PSC}$).

Therefore, we have
$\tr^* \cdot \alpha \in \traces{\cs{\MGC{\lib}}{\PSC}}$,
and so $\rsthistory{\methodset}{\tr} = \rsthistory{\methodset}{\tr'} \cdot \alpha = \rsthistory{\methodset}{\tr^*} \cdot \alpha =
\rsthistory{\methodset}{\tr^* \cdot \alpha} \in \rsthistory{\methodset}{\MGC{\lib}}$.
Then, the assumption that $\lib \sqsubseteq_\MGCn \lib^\sharp$, ensures that
$\rsthistory{\methodset}{\tr}\in \rsthistory{\methodset}{\MGC{\lib^\sharp}}$.

Finally, using \cref{lem:composition} (applied with $\lib:=\lib^\sharp$, $\lib':=\lib$, $\prog:=\prog$, and $\prog':=\MGCn$),
we obtain that $\history=\fhistory{\tr}\in \fhistory{\client{\prog}{\lib^\sharp}}$,
which contradicts our assumption.
\end{proof}

\compositionality*

\begin{proof}
We prove the claim by induction on $n$.
For $n=1$, the claim trivially follows.

For the induction step, let $\lib_1 \til \lib_n, \lib_1^\sharp \til \lib_n^\sharp$ be libraries,
and let $\MGCn$ be a program satisfying the required conditions.
For $\MGCn'=\MGC{\lib_n^\sharp}$, we have that
$$\usedlocs{\lib_1} \til \usedlocs{\lib_{n-1}},
\usedlocs{\lib_1^\sharp} \til  \usedlocs{\lib_{n-1}^\sharp},
\usedclientlocs{\dom{\lib_1 \uplustil \lib_{n-1}}}{\MGCn'}$$ are
pairwise disjoint.
In addition, for every $1\leq i\leq n-1$,
\begin{multline*}
\client{\MGCn'}
{\lib_1^\sharp \uplustil \lib_{i-1}^\sharp \uplus \lib_{i+1}^\sharp\uplustil \lib_{n-1}^\sharp}= \\
\client{\MGC{\lib_n^\sharp}}
{\lib_1^\sharp \uplustil \lib_{i-1}^\sharp \uplus \lib_{i+1}^\sharp\uplustil \lib_{n-1}^\sharp} = \\
\MGC{\lib_1^\sharp \uplustil \lib_{i-1}^\sharp \uplus \lib_{i+1}^\sharp
\uplustil \lib_{n}^\sharp}
\end{multline*}
Hence, for every $1\leq i\leq n-1$, we have
$\lib_i \sqsubseteq_{\MGCn_i} \lib_i^\sharp$
 for
 $$\MGCn_i=\client{\MGCn'}{\lib_1^\sharp \uplustil \lib_{i-1}^\sharp \uplus \lib_{i+1}^\sharp
\uplustil \lib_{n-1}^\sharp}.$$
By the induction hypothesis, it follows that
$\lib_1 \uplustil \lib_{n-1} \sqsubseteq_{\MGCn'} \lib_1^\sharp \uplustil \lib_{n-1}^\sharp$.

Let $\lib = \lib_1 \uplustil \lib_{n-1}$
and $\lib^\sharp = \lib_1^\sharp \uplustil \lib_{n-1}^\sharp$.
Then, we have
$\lib \sqsubseteq_{\MGCn'} \lib^\sharp$,
which implies that
$\fhistory{\MGC{\lib \uplus \lib_n^\sharp}}
\suq \fhistory{\MGC{\lib^\sharp \uplus \lib_n^\sharp}}$.
The latter implies that $\MGC{\lib}$ correctly calls $\lib_n$
\wrt $\MGC{\lib^\sharp}$.
In addition, by assumption we have
$\lib_n \sqsubseteq_{\MGCn_n} \lib_n^\sharp$
 for $\MGCn_n=\MGC{\lib^\sharp}$.
Hence, the abstraction theorem ensures that
 $\fhistory{\MGC{\lib \uplus \lib_n}} \suq
 \fhistory{\MGC{\lib \uplus \lib_n^\sharp}}$.
Together with the fact that
$\fhistory{\MGC{\lib \uplus \lib_n^\sharp}}
\suq \fhistory{\MGC{\lib^\sharp \uplus \lib_n^\sharp}}$,
we obtain that
 $\fhistory{\MGC{\lib \uplus \lib_n}} \suq
 \fhistory{\MGC{\lib^\sharp \uplus \lib_n^\sharp}}$,
 which implies that $\lib \uplus \lib_n \sqsubseteq_{\MGCn} \lib^\sharp \uplus \lib_n^\sharp$,
 and concludes our proof.
\end{proof}

\sim*

 \begin{proof}[Outline]
Let $\history \in \fhistory{\MGC{\lib}}$.
Let $\history_1 \til \history_n$ be \emph{crashless} histories such that
$\history = \history_1 \cdot \crash \cdottil \crash \cdot \history_n$.
Let  $\tr_1  \til \tr_n$ be crashless traces of $\cs{\MGC{\lib}}{\PSC}$,
such that $\fhistory{\tr_i} = \history_i$ for every $1\leq i\leq n$.
Let $\mem_0 \til \mem_n$ be non-volatile memories
such that each $\tr_i$ is $\mem_{i-1}$-to-$\mem_i$.
By repeatedly applying the assumption of the lemma (formally, inducting on $n$),
we obtain a sequence of crashless traces $\tr_1^\sharp  \til \tr_n^\sharp$
of $\cs{\MGC{\lib^\sharp}}{\PSC}$
and non-volatile memories $\mem_0^\sharp \til \mem_n^\sharp$
such that each $\tr_i^\sharp$ is $\mem_{i-1}^\sharp$-to-$\mem_i^\sharp$
and satisfies $\fhistory{\tr_i^\sharp} = \history_i$.
Then, it follows that $\history = \fhistory{\tr_1^\sharp} \cdot \crash \cdottil \crash \cdot \fhistory{\tr_n^\sharp}
\in \fhistory{\MGC{\lib^\sharp}}$.
\end{proof}

\pair*
 \begin{proof}[Sketch]
We use \cref{lem:sim} to prove the claim.
We let $\tup{\mem,\mem^\sharp}\in R$ iff the following hold:
\begin{itemize}
\item If $\mem(\xc)$ is even, then $\mem(\xone)= \mem^\sharp(\xone)$ and $\mem(\xtwo)= \mem^\sharp(\xtwo)$.
\item If $\mem(\xc)$ is odd,  then $\mem(\xonenew)= \mem^\sharp(\xone)$ and $\mem(\xtwonew)= \mem^\sharp(\xtwo)$.
\end{itemize}
Clearly, we have $\tup{\mem_\Init,\mem_\Init} \in R$.
Suppose that $\tup{\mem_0,\mem^\sharp_0} \in R$.
Let $\tr$ be an $\mem_0$-to-$\mem$ crashless trace of $\cs{\client{\MGCrec}{\lib_\mpair}}{\PSC}$.
We show that there exist a non-volatile memory $\mem^\sharp$ and an
$\mem^\sharp_0$-to-$\mem^\sharp$ crashless trace $\tr^\sharp$ of $\cs{\client{\MGCrec}{\lib_\mpair^\sharp}}{\PSC}$,
such that
$\tup{\mem,\mem^\sharp} \in R$
and
$\fhistory{\tr} = \fhistory{\tr^\sharp}$.
First, if $\tr$ ends during the execution of the recovery method, then we obtain $\tr^\sharp$ by executing the call of the recovery method,
and take $\mem^\sharp=\mem^\sharp_0$.
Otherwise, if recovery has completed, then after its completion, the invariant ensures that
$\memstate(\xone) = \memstate^\sharp(\xone)$,
$\memstate(\xtwo) = \memstate^\sharp(\xtwo)$,
and $\memstate(\xc) =0$.
Now, when the states are matching, by reusing the standard linearizability proof for the seqlock algorithm (see \cite{Burckhardt_2012_library_tso}),
we can obtain a trace of $\cs{\client{\MGCrec}{\lib_\mpair^\sharp}}{\PSC}$ with the same history as $\tr$.
It remains to handle persistency related steps,
\ie to decide when persist the block in the run of $\lib^\sharp$,
in a way that establishes the required relation on the non-volatile memories in the end of the trace.
For all complete executions of the write method, we persist the specification block just before the step in which the $\flInst{\xone}$ is executed.
For the incomplete invocations of the write method, we first note that at most one of them may manage to acquire the lock
and persist an odd value of $\xc$ (the rest are waiting in the busy loop, and have nothing to persist).
For that invocation, we persist the block at the point corresponding to the step in which the implementation persists the
odd value of $\xc$. (Note that this mean that we may need to exclude the $\flInst{\xone}$-step from the specification trace,
and we can do so since the invocation did not complete.)
This construction ensures that
$R$ holds for the non-volatile memories in the end of the trace.
\qed\vspace{0.5ex}
 \end{proof}

\bpair*

 \begin{proof}[Sketch]
We use \cref{lem:sim} to prove the claim.
We let $\tup{\mem,\mem^\sharp}\in R$ iff the following hold:
\begin{itemize}
\item If $\mem(\xflag)=0$, then $\mem(\xonenext)= \mem^\sharp(\xone)$ and $\mem(\xtwonext)= \mem^\sharp(\xtwo)$.
\item If $\mem(\xflag)=1$,  then $\mem(\xoneprev)= \mem^\sharp(\xone)$ and $\mem(\xtwoprev)= \mem^\sharp(\xtwo)$.
\end{itemize}
Clearly, we have $\tup{\mem_\Init,\mem_\Init} \in R$.
Suppose that $\tup{\mem_0,\mem^\sharp_0} \in R$.
Let $\tr$ be an $\mem_0$-to-$\mem$ crashless trace of $\cs{\client{\MGCrec}{\lib_\mbpair}}{\PSC}$.
We show that there exist a non-volatile memory $\mem^\sharp$ and
an $\mem^\sharp_0$-to-$\mem^\sharp$ crashless trace $\tr^\sharp$ of $\cs{\client{\MGCrec}{\lib_\mbpair^\sharp}}{\PSC}$,
such that
$\tup{\mem,\mem^\sharp} \in R$
and
$\fhistory{\tr} = \fhistory{\tr^\sharp}$.
First, if $\tr$ ends during the execution of the recovery method, then we obtain $\tr^\sharp$ by executing the call of the recovery method,
and take $\mem^\sharp=\mem^\sharp_0$.
Otherwise, if recovery has completed, then after its completion, the invariant ensures that
$\memstate(\vxone) = \memstate^\sharp(\xone)$ and
$\memstate(\vxtwo) = \memstate^\sharp(\xtwo)$.
In addition, since $\vxc$ is volatile, we also have $\memstate(\vxc) =0$.
Now, when the states are matching, by reusing the standard linearizability proof for the seqlock algorithm (see \cite{Burckhardt_2012_library_tso}),
we can obtain a trace of $\cs{\client{\MGCrec}{\lib^\sharp_\mbpair}}{\PSC}$ with the same history as $\tr$
(in particular, note that the read an flush methods do not interfere whatsoever).
It remains to handle persistency related steps,
\ie to decide when persist the block in the run of $\lib^\sharp$,
in a way that establishes the required relation on the non-volatile memories in the end of the trace.
Our construction performs all these persists just before the $\flInst{\xone}$-step from the specification trace (when the flush method is executed).
If there are incomplete invocations of the flush method in $\tr$, we first note that at most one of them may manage to acquire the lock
and persist $0$ for $\xflag$ (the rest are waiting in the busy loop, and have nothing to persist).
For that invocation, we persist the block at the point corresponding to the step in which the implementation persists $0$ for $\xflag$.
(Note that this mean that we may need to exclude the $\flInst{\xone}$-step from the specification trace,
and we can do so since the invocation did not complete.)
This construction ensures that $R$ holds for the non-volatile memories in the end of the trace.
To show this, one shows that $R$ is in fact an invariant of this construction that holds whenever the lock is not held ($\memstate(\xflag)=0$).
\end{proof}

\else

\vfill

{\small\medskip\noindent{\bf Open Access} This chapter is licensed under the terms of the Creative Commons\break Attribution 4.0 International License (\url{http://creativecommons.org/licenses/by/4.0/}), which permits use, sharing, adaptation, distribution and reproduction in any medium or format, as long as you give appropriate credit to the original author(s) and the source, provide a link to the Creative Commons license and indicate if changes were made.}

{\small \spaceskip .28em plus .1em minus .1em The images or other
third party material in this chapter are included in the\break
chapter's Creative Commons license, unless indicated otherwise in a
credit line to the\break material.~If material is not included in
the chapter's Creative Commons license and\break your intended use
is not permitted by statutory regulation or exceeds the
permitted\break use, you will need to obtain permission directly
from the copyright holder.}

\medskip\noindent\includegraphics{cc_by_4-0.eps}

\fi

\end{document}